\documentclass[a4paper,UKenglish,cleveref, autoref, thm-restate, usenames,dvipsnames]{lipics-v2021}
\newtheorem{innercustomthm}{Theorem}
\newenvironment{customthm}[1]
  {\renewcommand\theinnercustomthm{#1}\innercustomthm}
  {\endinnercustomthm}
  
\newtheorem{innercustomlem}{Lemma}
\newenvironment{customlem}[1]
  {\renewcommand\theinnercustomthm{#1}\innercustomlem}
  {\endinnercustomthm}

\newcommand{\msc}[1]{\mbox{\sc #1}} 

\usepackage[strict]{changepage}
\usepackage{tikz}
\usetikzlibrary{arrows}
\usepackage{proof}
\usepackage{mathtools}
\usepackage{tikz-cd}
\usepackage{mathpartir}
\usepackage[ruled,vlined]{algorithm2e}
\usepackage{stmaryrd}
\usepackage{changepage}
\usepackage{graphicx}
\usepackage{booktabs}   

\usepackage{amssymb}
\usepackage{verbatim}
\usepackage{bussproofs-extra}
\usepackage{calc}

\title{Strong Progress for Session-Typed Processes in a Linear Metalogic with Circular Proofs}
\titlerunning{Strong progress for session typed processes}
\author{Farzaneh Derakhshan}{Carnegie Mellon University}{fderakhs@andrew.cmu.edu}{}{}
\author{ Frank Pfenning}{Carnegie Mellon University}{fp@cmu.edu}{}{}

\authorrunning{F. Derakhshan et al.}
\Copyright{Farzaneh Derakhshan, Frank Pfenning} 

\begin{document}
\nolinenumbers
\maketitle              
\begin{abstract}
We introduce an infinitary first order linear logic with least and greatest fixed points. To ensure cut elimination, we impose a validity condition on infinite derivations. Our calculus is designed to reason about rich signatures of mutually defined inductive and coinductive linear predicates. In a major case study we use it to prove the strong progress property for binary session-typed processes under an asynchronous communication semantics.  As far as we are aware, this is the first proof of this property.
\end{abstract}

\section{Introduction}

 \emph{Session types} describe the communication behavior of interacting processes \cite{honda1993types,honda1998language}. Processes are connected via channels and the communication takes place by exchanging messages along the channels. Each channel has a type that governs the protocol of the communication. Generally, session types may be \emph{recursive}, thereby allowing them to capture unbounded interactions.  \emph{Binary session types} are a particular form in which each channel has just two endpoints. A channel connects the provider of a resource to its client. When such a channel connects two processes within a configuration of processes it is considered \emph{internal} and private; other \emph{external} channels provide an interface to a configuration and communication along them may be observed.

The {\it progress property} for a configuration of processes states that during its computation it either (i) takes an internal communication step, or (ii) is empty, or (iii) communicates along one of its external channels. The process typing rules guarantee the progress property even in the presence of recursive session types and recursively defined processes \cite{caires10concur,Toninho2013ESP}. The progress property is enough to ensure the absence of deadlocks, but it does not guarantee that a process will eventually terminate either in an empty configuration or one attempting to communicate along an external channel. The programmer will generally be interested in this stronger form of progress, requiring additional guard conditions on processes beyond typing.
 

In this paper we prove the strong progress property for processes that satisfy such a guard condition in a computational model based on \emph{asynchronous} communication. We restrict ourselves to the subsingleton fragment of recursive binary session types, differentiated into least and greatest fixed points. The processes defined over this fragment may use at most one resource and provide exactly one. Interestingly, the subsingleton fragment with recursive types already has the computational power of Turing machines \cite{deyoung2016substructural}. In prior work \cite{derakhshan2019circular} we already introduced a similar guard condition for this fragment with a \emph{synchronous} interpretation of communication and proved that programs satisfying this condition enjoy the strong progress property.  So our main interest here is not that such a proof exists, but (a) the novel structure of the proof, (b) the fact that the proof is carried out in a substructural metalogic with circular proofs, and (c) that it holds in an asynchronous semantics.

Underlying our development are the proof-theoretic roots of session types.
The recursion-free fragment of session types corresponds to sequent proofs in linear logic under a Curry-Howard interpretation of propositions as types, derivations as programs, and cut reduction as communication \cite{caires10concur,deyoung2016substructural}. In prior work \cite{derakhshan2019circular}, we extended the Curry-Howard interpretation of derivations in Fortier and Santocanale's singleton logic with circular proofs~\cite{Fortier13csl,santocanale2002calculus}. Under this interpretation, strong progress of a process follows from the cut elimination property of its underlying derivation. We showed that if a process satisfies a guard condition, its underlying derivation satisfies Fortier and Santocanale's condition. One shortcoming of building the guard condition upon the condition on its underlying logical derivation is that we cannot recognize some interesting programs as guarded, even though they enjoy the strong progress property.  As a corollary to the halting problem, no decidable condition can recognize all programs with the strong progress property \cite{derakhshan2019circular,deyoung2016substructural}. However, our long term objective is to capture more programs with this property as long as the algorithm is still effective, compositional, and predictable by the programmer. 

In this paper, we take a significant step toward this goal by formalizing strong progress as a predicate indexed by session types in a metalogic. We follow the approach of processes-as-formulas to provide an asynchronous semantics for session-typed processes and to carry out the proof of strong progress in the metalogic.  We use this proof technique as a case study of the guard criterion established in prior work, but also of how to prove properties of programming languages in a metalanguage with circular proofs. In addition, the theorem itself is also new since it deals with an \emph{asynchronous semantics}, which is generally more realistic and more appealing for concurrent programming than the prior synchronous one.
 
{\bf A New Metalogic.}  To carry out our argument in a metalogic  we need a calculus in which we can easily embed session-typed processes and define their operational behavior, which strongly suggests a linear metalogic.  Moreover, the formalization of strong progress inherits the need for using nested least and greatest fixed points from the session types that it is defined upon. Furthermore, we must be able to use these fixed points for proofs.  For these reasons, we decided to introduce a new metalogic: a calculus for intuitionistic linear logic with fixed points and infinitary proofs. In this logic, we can build an elegant derivation for the strong progress property of a process with clearly marked (simultaneous) inductive and coinductive steps. The resulting metalogical derivation is then bisimilar to the typing derivation for the process. This helps us to better understand the interplay between mutual inductive and coinductive steps in the proof of strong progress, and how they relate to the behavior of the program. 

In our first order linear metalogic, we allow circularity in derivations generalizing the approach of Brotherston et al. \cite{brotherston2005cyclic,brotherston2011JLC} by allowing both least and greatest fixed points.  However, we follow their  approach by allowing circular derivations instead of explicitly applying induction and coinduction principles. To ensure soundness of the proofs we impose a validity condition on our derivations. We introduce a cut elimination algorithm and prove its termination on valid derivations.

{\bf Contributions.} In summary, the main contributions of this paper are twofold. First, we introduce a new metalogic: an infinitary sequent calculus for first order multiplicative additive linear logic with (mutual) least and greatest fixed points (Section \ref{sec:metalogic}). We provide a validity condition on derivations of this logic to ensure the cut elimination property (Section~\ref{sec:validitycondition}). Second, we embed session-typed processes and their asynchronous semantics in this metalogic using a processes-as-formulas interpretation (Section~\ref{sec:procs}). We then define the strong progress property as a predicate with nested least and greatest fixed points. We prove strong progress of guarded programs by providing a syntactic proof for this predicate in our calculus and verifiying that this proof ensures strong progress of the underlying program (Section \ref{sec:stronprogress}).

{\bf Other Related Work.} Our metalogic and the validity condition imposed upon its derivations are a generalization of Fortier and Santocanale's singleton logic and their guard condition, respectively.  Baelde et al.~\cite{baelde2016infinitary,doumane2017infinitary} introduced a validity condition on the pre-proofs in multiplicative-additive linear logic with fixed points and proved cut-elimination for valid derivations. Our results, when restricted to the propositional fragment, differ from Baelde et al.'s in the treatment of intuitionistic linear implication ($\multimap$) versus its classical counterpart ($\bindnasrepma$). A stronger condition on the linear implication along with the fact that we are in an intuitionistic setting where exactly one formula is allowed as a succedent, allow us to adapt Fortier and Santocanale's cut elimination proof. This proof is essentially different from Baelde et. al's proof of cut elimination \cite{baelde2016infinitary,doumane2017infinitary}; in particular, we do not need to interpret the logical formulas in a classical truth semantics. 

Furthermore, our calculus is  essentially different from the finitary ones introduced by Baelde for the first order MALL with fixed points \cite{baelde2007least} and  $\text{Linc}^{-}$ \cite{tiu2012jal}, a first order logic with fixed points, since we allow for circularity. 


 Miller \cite{miller92welp} used the processes-as-formulas approach for expressing processes in the $\pi$-calculus as formulas in linear logic with non-logical constants. Revisions of this embedding to more expressive (finitary) extensions of linear logic \cite{horne2019MSCS,bruscoli02ICLP,tiu2010TOCL} are used to prove properties about processes, e.g. proofs of progress (deadlock-freedom) for circular multiparty sessions \cite{Horne20} and  bisimilarity for $\pi$-calculus processes \cite{tiu2010TOCL}.

The strong progress property for session typed processes is of similar nature to strong normalization in typed $\lambda$-calculi. In the setting of non-recursive session types, this property is reduced to termination of the computation and is often proved using logical relations \cite{perez12esop,deyoungsax}. The proofs of strong normalization for the simply typed $\lambda$-calculus, and termination for recursion-free processes both rely on an induction over the type structure and no longer apply after adding recursive types. In response, step-indexed logical relations \cite{appel2001indexed,ahmed2006step,ahmed2004semantics} have been developed to prove properties of typed calculi with recursive types. However,  neither strong normalization nor strong progress can be formalized as a logical relation indexed by the steps of computation as they are both associated to termination.
 
\section{First order intuitionistic linear logic with fixed points}\label{sec:metalogic}
In this section we introduce our metalogic: first order intuitionistic multiplicative additive linear logic with fixed points ({\small $\mathit{FIMALL}^{\infty}_{\mu,\nu}$}). The syntax and calculus of {\small $\mathit{FIMALL}^{\infty}_{\mu,\nu}$} is similar to the first order linear logic, but it is extended to handle predicates that are defined as mutual least and greatest fixed points.  The syntax of formulas follows the grammar
{\small\[
\begin{array}{lcl}
A & ::= & 1 \mid  0 \mid \top \mid A \otimes A \mid A \multimap A \mid A \oplus A \mid A\,\&\, A  \mid \exists x.\, A(x)  \mid \forall x.\, A(x) \mid s=t \mid T(\overline{t})
\end{array}
\]}%
where $s,t$ stand for terms and $x,y$ for term variables. We do not specify a grammar for terms; all terms are of the only type $U$. 
$T(\overline{t})$ is an instance of a predicate. A predicate can be defined using least and greatest fixed points in a \emph{signature} $\Sigma$.
{\small\[\Sigma ::= \cdot \mid \Sigma, T(\overline{x})=^{i}_{\mu} A \mid \Sigma, T(\overline{x})=^{i}_\nu A\]}%
The subscript $a$ of a fixed point $T(\overline{x})=^i_{a}A$ determines whether it is a least or greatest fixed point. If $a=\mu$, then predicate $T(\overline{x})$ is a least fixed point and inductively defined (e.g., the property of being a natural number) and if $a=\nu$ it is a greatest fixed point and coinductively defined (e.g., the lexicographic order on streams). Here we restrict $\Sigma$ to the definitions in which each fixed point predicate occurs only in covariant or contravariant positions, i.e. we do not allow mixed positions {\cite{pierce2002types}}.

The superscript $i\in \mathbb{N}$ is the relative priority of $T(\overline{x})$ in the signature $\Sigma$ with the condition that if $  T_1(\overline{x})=^{i}_{a} A, T_2(\overline{x})=^{i}_b B \in \Sigma$, then $a=b$. We say $T_1(\overline{x})=^{i}_{a} A$ has higher priority than $T_2(\overline{x})=^{j}_b B$ if $i<j$. The priorities determine the order by which the fixed-point equations in $\Sigma$ are solved {\cite{santocanale2002calculus}}. {Similar to prior work \cite{Fortier13csl,derakhshan2019circular} we use priority on predicates to define the validity condition on infinite derivations.}

\begin{example}\label{Nat-pred}
Let signature $\Sigma_1$ be
{\small\[
\begin{array}{lcl}
\mathtt{Stream}(x)&=^1_{\nu}& (\exists y. \exists z.(x = y \cdot z) \otimes\, \mathtt{Nat}(y) \otimes \mathtt{Stream}\, (z))\\
\mathtt{Nat}(x)&=^2_{\mu} &(\exists y. (x=\mathsf{s} y)\, \otimes \, \mathtt{Nat}(y))\, \oplus\, (x=\mathsf{z})\\
\end{array}
\]}%
where predicate $\mathtt{Nat}$ refers to the property of being a natural number, and predicate $\mathtt{Stream}$ refers to the property of being a stream of natural numbers. We interpret it as $\mathtt{Stream}$ having a higher priority relative to $\mathtt{Nat}$.
\end{example}


Derivations in {\small $\mathit{FIMALL}^{\infty}_{\mu,\nu}$} establish judgments of the form $\Gamma \vdash_{\Sigma} A$ where $\Gamma$ is an unordered list of formulas and $\Sigma$ is the signature. We omit $\Sigma$ from the judgments, since it never changes throughout a proof. The infinitary sequent calculus for this logic is given in Figure \ref{fig:rules-1}, in which we generalize $\oplus$ and $\&$ to be $n$-ary connectives $\oplus\{l_i:A_i\}_{i \in I}$ and $\&\{l_i:A_i\}_{i \in I}$. The binary disjunction and conjunction are defined as $A\oplus B=\oplus \{\pi_1:A,\pi_2:B\}$ and $A\& B=\& \{\pi_1:A,\pi_2:B\}$. Constants $0$ and $\top$ are defined as the nullary version of these connectives: $0=\oplus\{\}$ and $\top=\&\{\}$.
{\small\begin{figure*}[!t]
{\begin{equation*}
\hspace{-2.5cm}
\arraycolsep=0.1cm\begin{array}{cccc}
    \infer[\msc{Id}]{A \vdash A}{} &  \infer[\msc{Cut}]{\Gamma, \Gamma' \vdash C}{\Gamma \vdash A & \Gamma' , A \vdash C}  &
     \infer[1R]{\cdot \vdash 1}{} & \infer[1L]{\Gamma, 1 \vdash C}{ \Gamma \vdash C} \\[0.5em]
     \infer[\otimes R]{\Gamma, \Gamma' \vdash A_1 \otimes A_2}{\Gamma \vdash A_1 & \Gamma' \vdash A_2}&\infer[\otimes L]{\Gamma, A_1 \otimes A_2 \vdash B}{\Gamma, A_1, A_2\vdash B } &
      \infer[{\multimap} R]{\Gamma \vdash A_1 \multimap A_2}{\Gamma , A_1 \vdash A_2}&\infer[{\multimap} L]{\Gamma, \Gamma', A_1 \multimap A_2 \vdash B}{\Gamma \vdash A_1 & \Gamma', A_2 \vdash B}\\[0.5em]
     \infer[\oplus R]{\Gamma \vdash \oplus \{l_i: A_i\}_{i \in I}}{\Gamma \vdash A_k & k \in I}&\infer[\oplus L]{\Gamma, \oplus\{l_i:A_i\}_{i\in I} \vdash B}{\Gamma, A_i\vdash B & \forall i \in I} &
     \infer[\& R]{\Gamma \vdash \& \{l_i: A_i\}_{i \in I}}{\Gamma \vdash A_i & \forall i \in I}&\infer[\& L]{\Gamma, \&\{l_i:A_i\}_{i\in I} \vdash B}{\Gamma, A_k\vdash B & k \in I}\\[0.5em]
      \infer[\exists R]{\Gamma \vdash \exists x. P(x)}{\Gamma \vdash P(t)}&\infer[\exists L_x]{\Gamma, \exists x. P(x) \vdash B}{\Gamma, P(x)\vdash B & x\; \mbox{fresh}} & 
     \infer[\forall R_x]{\Gamma \vdash \forall x. P(x)}{\Gamma \vdash P(x) & x\; \mbox{fresh}}&\infer[\forall L]{\Gamma, \forall x. P(x) \vdash B}{\Gamma, P(t)\vdash B }\\[0.5em]
     \infer[\mu_{T} R]{\Gamma \vdash T(\overline{t})}{{\Gamma \vdash [\overline{t}/\overline{x}]A }&{T(\overline{x})=_{\mu} A} } &\infer[\mu_{T} L]{\Gamma, T(\overline{t}) \vdash B}{{\Gamma, [\overline{t}/\overline{x}]A\vdash B}&{ T(\overline{x})=_{\mu} A }}& 
      \infer[\nu_{T} R]{\Gamma \vdash T(\overline{t})}{{\Gamma \vdash [\overline{t}/\overline{x}]A}&{ T(\overline{x})=_{\nu}A} } & \infer[\nu_{T} L]{\Gamma, T(\overline{t}) \vdash B}{{\Gamma, [\overline{t}/\overline{x}]A\vdash B}&{ T(\overline{x})=_{\nu} A} }
      \end{array}
     \end{equation*}
      \vspace{-10pt}
      \[\arraycolsep=1.3cm\begin{array}{cccc}
      &\infer[{=}R]{\cdot \vdash s=s}{}& \infer[{=} L]{\Gamma, s=t \vdash B}{\Gamma[\theta] \vdash B[\theta] & \forall \theta \in \mathtt{mgu}(t,s) }&
      \end{array}\]
       \vspace{-16pt}
    \caption[Caption for LOF]{Infinitary calculus for first order linear logic with fixed points. (In the {=}L rule, the set $\mathtt{mgu}(t,s)$ is either empty, or a singleton set containing a most general unifier.)}
    \label{fig:rules-1}}
     \vspace{-16pt}
\end{figure*}}

{ An inference system characterizes the meaning of a formula as long as it satisfies the cut elimination property. The calculus in Figure~\ref{fig:rules-1} is infinitary, meaning that we allow building infinite derivations. Infinite derivations in our calculus do not necessarily enjoy the cut elimination property and thus are called \emph{pre-proofs} instead of \emph{proofs}. In Section~\ref{sec:validitycondition}, we introduce a validity condition on derivations to identify valid proofs among all pre-proofs. We prove that the cut elimination property holds for the derivations satisfying the condition.

The ability to build an infinite derivation is rooted in the fixed point rules. The least and greatest fixed point rules look similar as they unfold the definition of a predicate defined as fixed points in $\Sigma$. However, we will see in sections~\ref{sec:tower} and \ref{sec:validitycondition} that they each play a different role in determining validity of a derivation.}

A \emph{circular derivation} is the finite representation of an infinite one in which we can identify each open subgoal with an identical interior judgment. In the first order context we may need to use a substitution rule right before a circular edge to make the subgoal and interior judgment exactly identical \cite{brotherston2005cyclic}. See Appendix~\ref{app:logic} for the substitution rule and an how to transform a circular derivation into an infinite one. The definitions and proofs in this paper are based on the infinite system of Figure \ref{fig:rules-1}. But we present the derivations in a circular form.

 It may not be feasible to present a large piece of derivation fully in the calculus of Figure~\ref{fig:rules-1}. For the sake of brevity, we may represent predicates defined as least fixed points in the signature using pattern matching and build equivalent derivations based on that signature \cite{rosu2017matching,brotherston2005cyclic}. Whenever we use pattern matching, we make it clear how to transform the signature and derivations into our main logical system.

\section{An example - Tower of Hanoi}\label{sec:tower}
In this section, we provide an example to illustrate the applicability of our metalogic for formalizing the behavior and properties of recursive programs.



Tower of Hanoi problem consists of a source peg with a number of disks in ascending order of size. The goal is to transfer the entire stack of disks to a target peg using an auxiliary peg such that no larger disk moves onto a smaller one. Algorithm~\ref{alg:tower} shows a well-known recursive program for this puzzle. The input consists of the pegs, i.e. source, target, and auxiliary, and a stack of disks $I$ that needs to be moved in the current call. In each recursive call the size of the stack gets smaller ${\small(|I'|<|I|)}$ which makes the algorithm terminating.
\begin{small}
\begin{minipage}[c]{1.09\linewidth}
{\begin{algorithm}[H] \small
\SetAlgoLined 
{\bf def move(source, target, auxiliary, $I$)}:\\
    \If{$I =I'\,k$}
       {{\color{red}\bf move(source, auxiliary, target, $I'$)} {\small {\color{Brown} \it\#Move disks in $I'$ from source to auxiliary}\;
       {\bf\color{blue} target.push(source.pop())} {\small\color{Brown}\it \#Move disk k from source to target}\;
       {\color{OliveGreen}\bf move(auxiliary, target, source, $I'$)}
        {\small\color{Brown}\it \#Move the disks in $I'$ that we left on auxiliary onto target}\;
}}
\caption{Tower of Hanoi - move disks in $I$ from source to target using auxiliary.}
\vspace{-3pt}\label{alg:tower}
\end{algorithm}}
\end{minipage}
\end{small}
\begin{figure}[!t]
    \centering
{\small\[
\begin{array}{lll}
    \mathsf{move}(s,t,a,I,m)\, =^1_{\mu}  &\oplus\{\mathit{next}:\exists I'\, k.\, (I=I'\,k) & \otimes\,  {\color{red}\mathsf{move}(s,a,t,I',m)}\\ && \otimes\, {\color{blue}\mathsf{pop\_push}(s,t,k,I', m)} \\ && \otimes\, {\color{OliveGreen}\mathsf{move}(a,t,s,I',2^{\mid I' \mid}+m)},   \\
     &\;\;\;\;\mathit{done}: I=\epsilon\;\}
\end{array}\]}
{\small\[
\begin{array}{ll}
\mathsf{pop\_push}(s,t,k,I',m)=^1_{\mu}\forall L.\forall L'\,\,  &(\mathsf{count}(2^{|I'|}+m-1)  \otimes \mathsf{peg}(s,k\,L') \otimes \mathsf{peg}(t,L)) \multimap\\ &(\mathsf{count}(2^{|I'|}+m) \otimes \mathsf{peg}(s,L') \otimes\, \mathsf{peg}(t,k\,L))
\end{array}\]}
\vspace{-15pt}
    \caption{Interpretation of Algorithm~\ref{alg:tower} in {\small $\mathit{FIMALL}^{\infty}_{\mu,\nu}$}.}
    \label{fig:moveind} \vspace{-20pt}
\end{figure}
\indent In Figure~\ref{fig:moveind} we introduce a formula $\mathsf{move}(s,t,a,I,m)$ to interpret the recursive algorithm above in {\small $\mathit{FIMALL}^{\infty}_{\mu,\nu}$}. The formula is defined as a least fixed point with direct correspondence to the steps in Algorithm~\ref{alg:tower}. The predicate $\mathsf{peg}(s,L)$ presents the configuration of peg $s$ with stack $L$ on it.  Predicate $\color{blue}{\mathsf{pop\_push}(s,t,k,I',m)}$ formalizes the pop-push step (the 4th line) in Algorithm~\ref{alg:tower} using the \emph{linear implication}: it \emph{replaces} two configurations  $\mathsf{peg}(s,k\,L')$ and $\mathsf{peg}(t,L)$ with new configurations $\mathsf{peg}(s,L')$ and $\mathsf{peg}(t,k\,L)$. 

{\small $\mathit{FIMALL}^{\infty}_{\mu,\nu}$} is not an ordered logic. There is no a priori ordering on the application of pop-push predicates called (recursively) by $\mathsf{move}$. We use an extra argument $m$ to sort pop-push predicates in the order of instructions in Algorithm~\ref{alg:tower}. We rely on the fact that moving stack $I$ requires ${\small 2^{|I|}{-}1}$ instructions. Predicate $\mathsf{count}(m)$ ensures that pop-push steps are applied in order: its argument is incremented after applying the current pop-push predicate and becomes a pointer to the next one. 


The correctness of the algorithm is stated as: $\mathsf{move}(s,t,a, I_s,0)$ on an initial configuration $\mathsf{peg}(s, I_s)$, $\mathsf{peg}(a, \epsilon)$, and $\mathsf{peg}(t, \epsilon)$ with $\mathsf{count}(0)$, results in the final configuration
$\mathsf{peg}(s, \epsilon)$, $\mathsf{peg}(a, \epsilon)$, and $\mathsf{peg}(t, I_s)$ with $\mathsf{count}(2^{|I_s|}-1)$. We prove a generalization of this property:
{\small\begin{equation*} \hspace{-1cm} \dagger\,\mathsf{peg}(s,I\,L_{s}),\mathsf{peg}(t,L_t), \mathsf{peg}(a,L_a), \mathsf{count}(m), \mathsf{move}(s,t,a,L_s,m)\vdash \mathsf{peg}(s,L_{s})\otimes\,\mathsf{peg}(t,I\,L_t)\otimes\, \mathsf{peg}(a,L_a) \otimes\, \mathsf{count}(2^{|I|}+m-1)
\end{equation*}}
\begin{figure}[!t]
{\footnotesize
\begin{equation*}
\infer[\color{Brown}{\mu L}]{\dagger \,\mathsf{p}(s,I\,L_{s}),\mathsf{p}(t,L_t), \mathsf{p}(a,L_a),\mathsf{c}(m), {\color{Plum}\mathsf{m}(s,t,a,I,m)}\vdash \mathsf{p}(s,L_{s})\otimes\mathsf{p}(t,I\,L_t)\otimes \mathsf{p}(a,L_a)\otimes \mathsf{c}(2^{|I|}+m-1)} {\infer*{}{\deduce{{\color{red}\dagger_1} \,\mathsf{p}(s,I'\,k\,L_{s}),\mathsf{p}(t,L_t), \mathsf{p}(a,L_a), \mathsf{c}(m),   {\color{Plum}\mathsf{m}(s,a,t,I',m)} \vdash  \mathsf{p}(s,k\,L_{s})\otimes\mathsf{p}(t,L_t)\otimes \mathsf{p}(a,I'\,L_a) \otimes \mathsf{c}(2^{|I'|}{+}m{-}1)}{ {\color{OliveGreen}\dagger_2}\, \mathsf{p}(s,L_{s})\otimes\mathsf{p}(t,k\,L_t)\otimes \mathsf{p}(a,I'\,L_a)\otimes \mathsf{c}(2^{|I'|}{+}m),   {\color{Plum}\mathsf{m}(a,t,s,I',m)} \vdash  \mathsf{p}(s,L_{s})\otimes\mathsf{p}(t,I'\,k\,L_t)\otimes \mathsf{p}(a,L_a)\otimes \mathsf{c}(2^{|I'k|}{+}m{-}1)}}}   
\end{equation*}
}
\vspace{-0.9cm}
 \caption{A circular derivation for correctness of Algorithm~\ref{alg:tower}.(See Appendix~\ref{app:hanoi} for full derivation.)}
\label{fig:tower}
\vspace{-0.7cm}
\end{figure}
Using a few obvious logical rules, judgment $\dagger$ can be deduced from judgments ${\color{red}\dagger_1}$ and ${\color{OliveGreen}\dagger_2}$ (Figure~\ref{fig:tower}). We can identify ${\color{red}\dagger_1}$ and ${\color{OliveGreen}\dagger_2}$ with $\dagger$ after the proper substitution and form two {\bf cycles in the derivation}. 
 Predicates ${\color{Plum}\mathsf{move}}$ in the antecedents of judgments ${\color{red}\dagger_1}$ and ${\color{OliveGreen}\dagger_2}$ are both logically derived from ${\color{Plum} \mathsf{move}}$ in $\dagger$. On the thread of rules that transforms this predicate, there is an application of ${\color{Brown} \mu L}$. The ${\color{Brown} \mu L}$ on this thread corresponds to {\color{Brown}reducing an inductive input} before applying the inductive hypothesis; it ensures validity of both cycles.

Adding greatest fixed points to the signature mandates a more complex validity condition. As an example of such signature consider an infinitary variant of Tower of Hanoi: after the original stack (constant $I_s$) is completely transferred from  the source to the target, the program can choose to either terminate or reset the configuration and solve the puzzle again. 

We formalize a recursive solution to the infinitary variant in Figure~\ref{fig:moveco} with a minimal change to the predicate {$\mathsf{move}$}: at the end of each recursive call we add {${\color{Orange} \mathsf{start}(t,s,a,I,m)}$}. Predicate ${\color{Orange}\mathsf{start}}$ either terminates (label $\mathit{term}$) or asserts that the original stack has been transferred to the target ({$I=I_s$}), resets the configuration, and solves the puzzle again (label $\mathit{reset}$).
This solution is not necessarily terminating anymore, but it is {\bf productive}: ${\color{Orange}\mathsf{start}}$ calls $\mathsf{move}$ recursively only when it can assert that $I_s$ has been completely transferred from $s$ to $t$ once more. We define predicate ${\color{Orange}\mathsf{start}}$ as a \emph{greatest fixed point} and give it a \emph{higher priority} than $\mathsf{move}$. The higher priority of the greatest fixed point ${\color{Orange}\mathsf{start}}$ ensures that the overall solution is \emph{coinductive} rather than inductive.

The property asserted by $\dagger$ does not hold for the productive definition of $\mathsf{move}$ in Figure~\ref{fig:moveco}: the program has a possibility to run forever and never reach a \emph{final} configuration. However, we can build a circular derivation for $\dagger$ as sketched in Figure~\ref{fig:cotower}. This derivation should not and does not amount to a valid proof in {\small $\mathit{FIMALL}^{\infty}_{\mu,\nu}$}. There is a third {\bf cycle} in Figure~\ref{fig:cotower} formed from $\dagger_4$ to $\dagger$. This cycle does not correspond to a valid application of an inductive step. Predicate {\color{Plum}$\mathsf{move}$} in ${\dagger_4}$ is derived from ${\color{Plum}\mathsf{move}}$ in $\dagger$ as highlighted in purple. But the purple thread connecting these two predicates contains both $\mu L$ and $\nu L$ rules, with $\nu L$ having a higher priority. The $\nu L$ application with a higher priority cancels out the effect of the $\mu L$ rule on the cycle. The $\mu L$ rule cannot be used to justify validity of an inductive step anymore. In the next section, we provide a machinery to identify valid circular proofs based on this discussion. 
%


\begin{figure}[!t]
    \centering
    {\small\[
\begin{array}{lll}
    \mathsf{move}(s,t,a,I,m)\, =^2_{\mu}  &\oplus\{\mathit{next}:\exists I'\, k.\, (I=I'\,k) & \otimes\,  {\color{red}\mathsf{move}(s,a,t,I',m)}\\ && \otimes\, {\color{blue}\mathsf{pop\_push}(s,t,k,I', m)} \\ && \otimes\, {\color{OliveGreen}\mathsf{move}(a,t,s,I',2^{\mid I' \mid}{+}m)}\\
    && \otimes \, {\color{Orange} \mathsf{start}(t,s,a,I,2^{\mid I \mid}{+}m{-}1)}, \\
     &\;\;\;\;\mathit{done}: I=\epsilon\;\}\\
{\mathsf{start}(t,s,a,I,m)}=^1_{\nu}& \oplus\{\mathit{restart}: I=I_s \otimes\, \forall L,L'& (\mathsf{peg}(s, L) \otimes \mathsf{peg}(t,I_s\,L')\otimes\mathsf{count}(m)) \multimap\\&& \;\;(\mathsf{peg}(s, I_s\,L) \otimes \mathsf{peg}(t,L') \otimes \, \mathsf{count}(0)\,\otimes \, \mathsf{move}(s,t,a,I_s,0)),\\
&\;\;\;\;\mathit{term}:1\;\}
\end{array}
\]}\vspace{-0.7cm}
    \caption{Interpretation of a productive solution for infinitary Tower of Hanoi.}
    \vspace{-0.5cm}
    \label{fig:moveco}
\end{figure}

\begin{figure}[!t]
    \centering
{\footnotesize
\[\hspace{-1.5cm}
\infer[\mu L]{\dagger}{\infer*{}{\infer[]{}{\dagger_1 & \dagger_2 &\infer[\nu L]{{\color{Orange}\dagger_3}\, \mathsf{p}(s,L_{s})\otimes\mathsf{p}(t,I'\,k\,L_t)\otimes \mathsf{p}(a, L_a)\otimes \mathsf{c}(2^{|I_s|}{+}m{-}1),   {\color{Plum}\mathsf{s}(t,s,a,I,2^{|I'k|}{+}m{-}1)} \vdash  \mathsf{p}(s,L_{s})\otimes\mathsf{p}(t,I'\,k\,L_t)\otimes \mathsf{p}(a,L_a)\otimes \mathsf{c}(2^{|I'k|}{+}m{-}1)
}{\infer*{}{\dagger_4\,\mathsf{p}(a, L_a),{\color{Plum} \mathsf{p}(s,I_s\, L), \mathsf{p}(t,L'), \mathsf{c}(0), \mathsf{m}(s,t,a,I_s,0)} \vdash  \mathsf{p}(s,L_{s})\otimes\mathsf{p}(t,I_s\,L_t)\otimes \mathsf{p}(a,L_a)\otimes \mathsf{c}(2^{|I_s|}{+}m{-}1)}}}}}
\]
}

    \vspace{-1cm}
    \caption{An invalid circular derivation for correctness of infinitary Tower of Hanoi solution }
    \vspace{-20pt}
    \label{fig:cotower}
\end{figure}
\section{A validity condition on first order derivations} \label{sec:validitycondition}

In Section~\ref{sec:metalogic}, we introduced an infinitary calculus for first order linear logic with fixed points.
In Section \ref{sec:tower}, we saw that not all infinite derivations amount to a valid proof. This section establishes a concept of validity with respect to cut elimination. As usual, cut elimination for valid derivations ensures consistency: it implies that there is no proof for {\small$\cdot \vdash 0$} in our calculus.


To establish a validity condition we need to keep track of behavior of any particular formula throughout the whole derivation.  We annotate formulas with position variables $\mathbf{x},\mathbf{y},\mathbf{z}$ and track their generations $\alpha, \beta$. Different generations of the same variable allow us to capture the progress of a given formula which is made when an unfolding message is applied on it. We picked generational position variables to track formulas over alternatives, e.g. Baelde et al.'s (pre)formula occurrences, since they resemble channels in session-typed processes \cite{derakhshan2019circular}. We will use this analogy in the proof of strong progress property (Section~\ref{sec:stronprogress}). 

Figure \ref{fig:rules-2} shows the calculus annotated with position variable generations and their relations. A basic judgment in the annotated calculus is of the form {\small$\Delta \vdash_{\Omega} \mathbf{z}^\beta{:}C$} where $\Delta$ is an unordered list of formulas annotated with (unique) generational position variables, i.e. its elements are of the form {\small$\mathbf{x}^\alpha{:}A$}. The set $\Omega$ keeps the relation between different generation of position variables in a derivation for each priority.


The relation of a new generation to its priors is determined by the role of the rule that introduces it in (co)induction. The $\mu L$ rule breaks down an inductive antecedent and $\nu R$  produces a coinductive succedent. They both take a step toward termination/productivity of the proof: we put the new generation of the position variable they introduce to be less than the prior ones in the given priority.
Their counterpart rules $\nu L$ and $\mu R$, however, do not contribute to termination/productivity. They break the relation between the new generation and its prior ones for the given priority. 

In the $\mu L$ rule, for example, we add the relation $\mathbf{y}^{\alpha+1}_i<\mathbf{y}^{\alpha}_i$ to $\Omega'$. It is interpreted as the new generation  $\mathbf{y}^{\alpha+1}$ is less than its prior generation on priority $i$. For the other priorities $j\neq i$ we keep $\mathbf{y}^{\alpha+1}_j=\mathbf{y}^{\alpha}_j$.

In the {\small$\msc{Cut}$} rule we introduce a fresh position variable annotated with a generation, {\small$\mathbf{w}^\eta$}. Since it refers to a new formula, we designate it to be incomparable to other position variables. We consider {\small$\mathbf{w}^\eta$} as a continuation of {\small$\mathbf{z}^\beta$} in the rule $\otimes R$ similar to its counterpart in the guard condition for propositional {\small$\mu$-$\mathit{MALL}^{\infty}$} \cite{baelde2016infinitary}. Similarly, we keep the relation of  {\small$\mathbf{y}^\alpha$} with its continuation {\small$\mathbf{w}^\eta$} in $\Omega$ for the $\otimes L$ rule. The fresh position variable $\mathbf{w}^\eta$  introduced in $\multimap R$ (resp. $\multimap L$) rule switches its polarity from right to left (resp. left to right) so it cannot be equal to $\mathbf{z}^\beta$ (resp. $\mathbf{y}^\alpha$).
 As a result, our condition on $\multimap$ is more restrictive than its classical counterpart $\bindnasrepma$ in \cite{baelde2016infinitary}.

\begin{figure}[t!]
    \centering
   \begin{minipage}[c]{1.09\linewidth} {\small
 \begin{adjustwidth}{-1.5cm}{0cm}
 \begin{equation*}
     \begin{array}{cccc} 
       \infer[\msc{Id}]{ \mathbf{x}^\alpha:A \vdash_{\Omega} \mathbf{z}^\beta:A}{} &  \infer[\msc{Cut}]{\Delta, \Delta' \vdash_{\Omega} \mathbf{z}^\beta:C}{\Delta \vdash_{\Omega}  \mathbf{w}^\eta:A & \Delta', \mathbf{w}^\eta:A \vdash_{\Omega} \mathbf{z}^\beta:C} &
     \infer[1R]{\cdot \vdash_{\Omega}  \mathbf{z}^\beta:1}{} & \infer[1L]{\Delta, \mathbf{y}^\alpha:1 \vdash_{\Omega} \mathbf{z}^\beta:C}{ \Delta\vdash_{\Omega}\mathbf{z}^\beta:C}
     \end{array}
 \end{equation*}
 \end{adjustwidth}
 \vspace{-15pt}
  \begin{adjustwidth}{-1cm}{0cm}
 \begin{equation*}
     \arraycolsep=0.6cm\begin{array}{cc}
      \infer[\otimes R]{\Delta, \Delta' \vdash_{\Omega} \mathbf{z}^\beta:A_1 \otimes A_2}{\Delta \vdash_{\Omega\cup\{\mathbf{w}^\eta=\mathbf{z}^\beta\}} \mathbf{w}^\eta:A_1 & \Delta' \vdash_{\Omega} \mathbf{z}^\beta:A_2}&\infer[\otimes L]{\Delta, \mathbf{y}^\alpha:A_1 \otimes A_2 \vdash_{\Omega}\mathbf{z}^\beta:B}{\Delta, \mathbf{w}^\eta:A_1, \mathbf{y}^\alpha:A_2\vdash_{\Omega\cup\{\mathbf{w}^\eta=\mathbf{y}^\alpha\}} \mathbf{z}^\beta:B }\\[0.75em]
      \infer[\multimap R]{\Delta \vdash_{\Omega} \mathbf{z}^\beta:A_1 \multimap A_2}{\Delta, \mathbf{w}^\eta:A_1 \vdash_{\Omega}  \mathbf{z}^\beta:A_2}&\infer[\multimap L]{\Delta, \Delta', \mathbf{y}^\alpha:A_1 \multimap A_2 \vdash_{\Omega} \mathbf{z}^\beta:B}{\Delta \vdash_{\Omega}\mathbf{w}^\eta:A_1 & \Delta', \mathbf{y}^\alpha:A_2 \vdash_{\Omega} \mathbf{z}^\beta:B}\end{array}
 \end{equation*}
 \end{adjustwidth}
  \vspace{-10pt}
     \begin{equation*}
     \hspace{-2.9cm}
    \arraycolsep=0.1cm\begin{array}{cccc}
     \infer[\oplus R]{\Delta \vdash_{\Omega}  \mathbf{z}^\beta:\oplus \{l_i: A_i\}_{i \in I}}{{\Delta \vdash_{\Omega}\mathbf{z}^{\beta}:A_k}&{k \in I}}&\infer[\oplus L]{\Delta, \mathbf{y}^\alpha:\oplus\{l_i:A_i\}_{i\in I} \vdash_{\Omega} \mathbf{z}^\beta{:}B}{\Delta, \mathbf{y}^\alpha:A_i\vdash_{\Omega}  \mathbf{z}^\beta:B & \forall i \in I}&
     \infer[\& R]{\Delta \vdash_{\Omega}  \mathbf{z}^\beta:\& \{l_i: A_i\}_{i \in I}}{\Delta \vdash_{\Omega}  \mathbf{z}^\beta{:} A_i & \forall i \in I}&\infer[\& L]{\Delta, \mathbf{y}^\alpha:\&\{l_i:A_i\}_{i\in I} \vdash_{\Omega} \mathbf{z}^\beta: B}{{\Delta, \mathbf{y}^\alpha:A_k\vdash_{\Omega} \mathbf{z}^\beta:B}&{k \in I}}\\[0.75em] 
      \infer[\exists R]{\Delta \vdash_{\Omega} \mathbf{z}^\beta:\exists x. P(x)}{\Delta \vdash_{\Omega} \mathbf{z}^\beta:P(t)}& \infer[\exists L_x]{\Delta, \mathbf{y}^\alpha:\exists x. P(x) \vdash_{\Omega} \mathbf{z}^\beta:B}{\Delta, \mathbf{y}^\alpha:P(x) \vdash_{\Omega} \mathbf{z}^\beta:B & x\; \mbox{fresh}}&
     \infer[\forall R_x]{\Delta \vdash_{\Omega}\mathbf{z}^\beta:\forall x. P(x)}{\Delta \vdash_{\Omega} P ::\mathbf{z}^\beta:P(x) & x\; \mbox{fresh}} & \infer[\forall L]{\Delta, \mathbf{y}^\alpha:\forall x. P(x) \vdash_{\Omega} \mathbf{z}^\beta:B}{\Delta, \mathbf{y}^\alpha:P(t)\vdash_{\Omega} \mathbf{z}^\beta:B }
     \end{array}
     \end{equation*}
     \[\begin{array}{cc}
     \infer[\mu R]{\Delta \vdash_{\Omega} \mathbf{z}^\beta:T(\overline{t})}{\deduce{\Delta \vdash_{\Omega'}\mathbf{z}^{\beta+1}:[\overline{t}/\overline{x}]A \hspace{0.5cm} T(\overline{x})=^j_{\mu} A}{\Omega'= \Omega \cup \{\mathbf{z}^{\beta+1}_{i}= \mathbf{z}^{\beta}_{i}\mid i\neq j\}} } & \infer[\mu L]{\Delta, \mathbf{y}^\alpha:T(\overline{t}) \vdash_{\Omega} \mathbf{z}^\beta:B}{\deduce{\Delta, \mathbf{y}^{\alpha+1}: [\overline{t}/\overline{x}]A\vdash_{\Omega'}\mathbf{z}^\beta:B \hspace{1.4cm} T(\overline{x})=^j_{\mu} A}{ \Omega'= \Omega \cup \{\mathbf{y}^{\alpha+1}_{i}= \mathbf{y}^{\alpha}_{i}\mid i\neq j\} \cup \{\mathbf{y}^{\alpha+1}_{j}<\mathbf{y}^{\alpha}_{j}\}}} \\[0.55em]
     \infer[\nu R]{\Delta \vdash_{\Omega} \mathbf{z}^\beta:T(\overline{t})}{\deduce{\Delta \vdash_{\Omega'} \mathbf{z}^{\beta+1}:[\overline{t}/\overline{x}]A \hspace{2.1cm} T(\overline{x})=^j_{\nu} A}{  
      \Omega'= \Omega \cup \{\mathbf{z}^{\beta+1}_{i}= \mathbf{z}^{\beta}_{i}\mid i\neq j\}\cup \{\mathbf{z}^{\beta+1}_{j}<\mathbf{z}^{\beta}_{j}\}}} &  \infer[\nu L]{\Delta, \mathbf{y}^\alpha:T(\overline{t}) \vdash_{\Omega} \mathbf{z}^\beta:B}{\deduce{\Delta, \mathbf{y}^{\alpha+1}:[\overline{t}/\overline{x}]A\vdash_{\Omega'} \mathbf{z}^\beta:B \hspace{0.5cm} T(\overline{x})=^j_{\nu} A}{\Omega'= \Omega \cup \{\mathbf{y}^{\alpha+1}_{i}= \mathbf{y}^{\alpha}_{i}\mid i\neq j\}} }\\[0.35em]
     \infer[= R]{\cdot \vdash_{\Omega} \mathbf{z}^\beta:(s=s)}{}&\infer[= L]{\Delta, \mathbf{y}^\alpha:(s=t) \vdash_{\Omega} \mathbf{z}^\beta:B}{\Delta[\theta]\vdash_{\Omega}  \mathbf{z}^\beta: B[\theta] & \forall \theta \in \mathtt{mgu}(t,s) }
     \end{array}\]}
     \end{minipage}
      \vspace{-15pt}
    \caption{Infinitary calculus annotated with position variables and their generations.}
    \label{fig:rules-2}
    \vspace{-20pt}
\end{figure}

Unlike the existing validity conditions for infinitary calculi defined only over least fixed points \cite{brotherston2005cyclic,brotherston2011JLC,sprenger2003structure, Das2018csl}, priorities are essential in a setting with nested least and greatest fixed points. Here both inductive and coinductive predicates may be unfolded infinitely often along the left and right sides of a branch, but \emph{only the one with the highest priority shall be used to ensure validity of it}. We, similar to Fortier and Santocanale~\cite{Fortier13csl,santocanale2002calculus}, have the flexibility to assign a priority to each predicate variable. In the validity conditions introduced by Dax et al.~\cite{dax2006lineartime} and Baelde et al.~\cite{baelde2016infinitary} the syntactic subformula ordering determines the priorities.


To sort fixed point unfolding rules applied on previous generations of position variable $\mathbf{x}^\alpha$ by their priorities we use a snapshot of $\mathbf{x}^\alpha$. For a given signature $\Sigma$, \emph{snapshot} of an annotated position variable $\mathbf{x}^\alpha$ is a list $\mathsf{snap}(\mathbf{x}^\alpha)=[\mathbf{x}^\alpha_i]_{i \le n}$, where  $n$ is the maximum priority in $\Sigma$. Having the relation between annotated position variables in $\Omega$, we can define a partial order on  snapshots of annotated position variables. We write\,
{\small$\mathsf{snap}(\mathbf{x}^\alpha)=[\mathbf{x}^\alpha_1\cdots \mathbf{x}^\alpha_n]<_{\Omega}[\mathbf{z}^\beta_1\cdots \mathbf{z}^\beta_n]= \mathsf{snap}(\mathbf{z}^\beta)$}%
if the list {$[\mathbf{x}^\alpha_1\cdots \mathbf{x}^\alpha_n]$} is less than {$[\mathbf{z}^\beta_1\cdots \mathbf{z}^\beta_n]$} by the lexicographic order defined by the transitive closure of the relations in $\Omega$.
We adapt the definitions of left $\mu$-trace and right $\nu$-trace from Fortier and Santocanale to our setting.
\begin{definition}\label{def:mu}
{\small An infinite branch of a derivation is a \emph{left $\mu$-trace} if for infinitely many position variables $\mathbf{x1}^{\alpha_1}, \mathbf{x2}^{\alpha_2}, \cdots$ appearing as antecedents of judgments $\mathbf{xi}^{\alpha_i}: A_i, \Delta_i \vdash_{\Omega_i} \mathbf{w}^{\beta}:C_i,$ 
we can form an infinite chain of inequalities 
{\small$
\mathsf{snap}(\mathbf{x1}^{\alpha_1})>_{\Omega_2}\mathsf{snap}(\mathbf{x2}^{\alpha_2})>_{\Omega_3}\cdots
$}.

Dually, an infinite branch of a derivation is a \emph{right $\nu$-trace} if for infinitely many position variables $\mathbf{y1}^{\beta_1}, \mathbf{y2}^{\beta_2}, \cdots$ appearing as the succedents of judgments $\Delta_i \vdash_{\Omega_i} \mathbf{yi}^{\beta_i}:C_i$ in the branch
we can form an infinite chain of inequalities 
{\small$
\mathsf{snap}(\mathbf{y1}^{\beta_1})>_{\Omega_2}\mathsf{snap}(\mathbf{y2}^{\beta_2})>_{\Omega_3}\cdots.
$} }
\end{definition}

\begin{definition}[Validity condition for infinite derivations] An infinite derivation is a \emph{valid proof} if each of its infinite branches is either a left $\mu$-trace or a right $\nu$-trace. A circular derivation is a \emph{proof} if it has a valid underlying infinite derivation.
\end{definition}
 We introduce a cut elimination algorithm for infinite pre-proofs in {\small $\mathit{FIMALL}^{\infty}_{\mu,\nu}$}. We prove that this algorithm is productive for valid derivations: the algorithm receives an infinite valid proof as an input and outputs a cut-free infinite valid proof productively. An algorithm is productive if every piece of its output is generated in a finite number of steps. We present the algorithm and its proof of productivity in Appendix \ref{app:cut}. 
 
\begin{theorem}
A valid (infinite) derivation enjoys the cut elimination property. 
\end{theorem}
\begin{proof}
 Consider a derivation given in the system of Figure~\ref{fig:rules-1}. We first annotate the derivation to get one in the system of Figure \ref{fig:rules-2}. This can be done productively: we start by  annotating the root with arbitrary generational position variables, and continue by replacing the last rule with its annotated version. In the appendix, we prove a lemma which states termination of the principal reductions of the algorithm  (Lemma \ref{thm:main}).  {Fortier and Santocanale's proof of a similar lemma for singleton logic is based on an observation that a $\nu$-trace on a valid derivation tree has only a limited number of branches on its right side created by the cut rule.  In our calculus  $\otimes R$ and $\multimap L$ rules also create branches. We introduce a fresh position variable $\mathbf{w}^\eta$ as a succedent in the branching rules. However, in the branches created by $\otimes R$ we keep the relation between the fresh variable $\mathbf{w}^\eta$ with its parent $\mathbf{z}^\beta$. As a result there may be an infinite $\nu$-trace with infinitely many branches on its right created by the $\otimes R$ rule. To take advantage of a similar observation to Fortier and Santocanale, we distinguish between branching on fresh succedent position variables created by a cut or a $\multimap L$ rule, and a $\otimes R$ rule. After this distinction, our cut elimination algorithm creates a trace which is a chain complete partially ordered set, rather than a complete lattice in Fortier and Santocanale's proof. We show that having a chain complete partially ordered set is enough for proving the lemma.}  As a corollary to this lemma our cut elimination algorithm produces a potentially infinite cut-free proof for the annotated derivation. We further prove that the output of the cut elimination algorithm is valid.  By simply ignoring the annotations of the output, we get a cut free valid derivation in the calculus of Figure \ref{fig:rules-1}. 
\end{proof}
\section{Session-typed processes}\label{sec:procs}
We define {\it session types} with grammar
{\small$\mathtt{A}::= {\oplus}\{\ell : \mathtt{A}_\ell\}_{\ell \in L} \mid {\&}\{\ell : \mathtt{A}_\ell\}_{\ell \in L} \mid 1 \mid \mathtt{t},$}
where $L$ ranges over finite sets of labels denoted by $\ell$ and $k$. $\mathtt{t}$ stands for type variables whose definition is given in a signature $\mathbf{\Sigma}$. The signature $\mathbf{\Sigma}$ is defined as {\small\[\mathbf{\Sigma} ::= \cdot \mid \mathbf{\Sigma}, \mathtt{t}=^{i}_{\mu} \mathtt{A} \mid \mathbf{\Sigma}, \mathtt{t}=^{i}_\nu \mathtt{A},\]} \noindent with the condition that if $  \mathtt{t}=^{i}_{a} \mathtt{A} \in \mathbf{\Sigma}$ and $\mathtt{t}=^{i}_b \mathtt{B} \in \mathbf{\Sigma}$, then $a=b$. If $a=\mu$, then session type $\mathtt{t}=^i_{a}$ is a positive fixed point and if $a=\nu$ it is a negative fixed point.


\begin{table*}[t!]
\begin{adjustwidth}{-1.2cm}{}
\centering
{\small
\begin{tabular}{@{}lllllc@{}}
\toprule
\multicolumn{2}{@{}l}{\textbf{Session type (curr./cont.)}} &
\multicolumn{2}{l}{\textbf{Process term (curr./cont.)}} &
\textbf{Description} &
\textbf{Pol} \\[3pt]
$x {:} \oplus\{\ell{:}\mathtt{A}\}_{\ell \in L}$ & $x {:} \mathtt{A}_k$ & $x.k; \mathsf{P}$ &
$\mathsf{P}$ & provider sends label $k$ along $x$  & + \\
 & & $\mathbf{case}\, x (\ell \Rightarrow \mathsf{Q}_{\ell})_{\ell \in L}$ & $\mathsf{Q}_k$ &
client receives label $k$ along $x$ & \\[2pt]
$x {:} \&\{\ell{:}\mathtt{A}\}_{\ell \in L}$ & $x {:} \mathtt{A}_k$ & $\mathbf{case}\, x (\ell \Rightarrow \mathsf{P}_{\ell})_{\ell \in L}$ &
$\mathsf{P}_k$ & provider receives label $k$ along $x$  & - \\
 & & $x.k\; \mathsf{Q}$ & $\mathsf{Q}$ & client sends label $k$ along $x$   & \\[2pt]
$x : 1$ & - & $\mathbf{close}\, x$ &
- & provider sends ``$\mathsf{close}$'' along $x$ and terminates  & + \\
 & & $\mathbf{wait}\,x;\mathsf{Q}$ & $\mathsf{Q}$ & provider receives ``$\mathsf{close}$'' along $x$  & \\
 
$x {:} \mathtt{t}$ & $x {:} \mathtt{A}$ & $x.\mu_\mathtt{t}; \mathsf{P}$ &
$\mathsf{P}$ & provider sends unfolding message $\mu_\mathtt{t}$ along $x$  & + \\
 ($\mathtt{t}=^i_{\mu}\mathtt{A} \in \mathbf{\Sigma}$) & & $\mathbf{case}\, x (\mu_\mathtt{t} \Rightarrow \mathsf{Q})$ & $\mathsf{Q}$ &
client receives unfolding message $\mu_\mathtt{t}$ along $x$ & \\[2pt]
$x {:} \mathtt{t}$ & $x {:} \mathtt{A}$ & $\mathbf{case}\, x (\nu_\mathtt{t} \Rightarrow \mathsf{P})$ &
$\mathsf{P}$ & provider receives unfolding message $\nu_{\mathtt{t}}$ along $x$   & - \\
($\mathtt{t}=^i_{\nu}\mathtt{A} \in \mathbf{\Sigma}$) & & $x.\nu_\mathtt{t}\; \mathsf{Q}$ & $\mathsf{Q}$ & client sends unfolding message $\nu_\mathtt{t}$ along $x$  & \\[2pt]
\bottomrule
\end{tabular}}
\end{adjustwidth}
\vspace{-2pt}
\caption{Overview of intuitionistic linear session types with their operational meaning.}
\label{tab:session_types}
\vspace{-30pt}
\end{table*}



Several notations introduced for {\small $\mathit{FIMALL}^{\infty}_{\mu,\nu}$} calculus  overlap with the notations we use in the context of session-typed processes.  For example in both,  a signature stores definitions and the relative priority of fixed points. The difference is that in the context of the first order calculus $\Sigma$ contains predicates but in the context of session-typed processes $\mathbf{\Sigma}$ contains (not dependent) session types. In the rest of this paper, we fix a signature $\mathbf{\Sigma}$ with $\mathbf{n}$ being the maximum priority in $\mathbf{\Sigma}$. The overloaded notation is inevitable since our metalogic is a generalization of the infinitary subsingleton logic with fixed points based on which binary session typed processes are defined. We use it for our advantage in the last section to prove our main result using a bisimulation.


A process $\mathsf{P}$ in this context is represented as $x:\mathtt{A} \vdash \mathsf{P} :: (y:\mathtt{C})$, where $x$ and $y$ are its left and right channels of types $\mathtt{A}$ and $\mathtt{C}$, respectively.
One can read this judgment as ``process $\mathsf{P}$ \emph{provides} a service of type $\mathtt{C}$ along channel $y$ while \emph{using} a service that is provided by some other process $\mathsf{Q}$ along channel $x$.'' An interaction along a channel $x$ is of two forms (i) $\mathsf{Q}$ sends a message to the right, and process  $\mathsf{P}$ receives it from the left, or (ii) process $\mathsf{P}$ sends a message to the left and process $\mathsf{Q}$ receives it from the right. In the first case the session type $\mathtt{A}$ is of positive polarity, and in the second case it is negative. Since a process might not use any service provided along its left channel, e.g. $\cdot \vdash \mathsf{P} :: (y:\mathtt{B}),$ we generalize the labelling of processes to be of the form $\bar{x}:\omega \vdash \mathsf{P}::(y:\mathtt{B}),$
where $\bar{x}$ is either empty or $x$, and $\omega$ is empty given that $\bar{x}$ is empty. 

Mutually recursive processes are allowed and implemented using process variables. A process variable $\mathsf{X}$ is defined as process $\mathsf{P}$ using a judgment of the form {\small${\bar{x}:\omega \vdash \mathsf{X}=\mathsf{P}_{\bar{x},y}:: (y:\mathtt{B})}$}.
The syntax for calling a  process variable recursively is {\small${ y \leftarrow \mathsf{X} \leftarrow \overline{x}}$}.  We store all process variable definitions in a finite set $V$. 


A summary of the operational reading of session types is presented  in  Table~\ref{tab:session_types}.  The  first  column  indicates  the session type before the message exchange, the second column the session type after the exchange. The corresponding process terms  are  listed  in  the  third  and  fourth  column,  respectively. The fifth column provides the operational meaning of the type  and  the  last  column  its polarity.  Since the fixed points are isorecursive, processes send and receive explicit fixed point unfolding messages. As the names suggest a positive fixed point has a positive semantics and a negative fixed point has a negative one.


In the next section we introduce an \emph{asynchronous} dynamics for the processes.  In an asynchronous semantics only receivers can be blocked, while senders spawn a message and proceed with their continuation. To ensure that consecutive messages  arrive  in the  order  they  were  sent, we use the overlapping notion of generations from {\small $\mathit{FIMALL}^{\infty}_{\mu,\nu}$} and annotate a process typing judgment with generations.  A judgment in the annotated format is of the form ${\bar{x}^\alpha:\mathtt{A} \vdash_{\Omega} \mathtt{P} :: (y^\beta:\mathtt{B})}$.  The generation $\alpha$ of a channel $x^\alpha$ is incremented to $\alpha+1$ whenever a new message is spawned or received, except for type $1$ because there is no continuation. Figure \ref{fig:stp-order} presents the annotated inference rules for session-typed processes corresponding to derivations in subsingleton logic with fixed points. The set $\Omega$ in the judgment plays the same role as its counterpart in {\small $\mathit{FIMALL}^{\infty}_{\mu,\nu}$}. Later in this section we explain its role in developing a guard condition over processes.

\begin{example}
Consider  process $ y^\beta{:}1 \vdash \mathsf{Loop} :: (x^\alpha:\mathtt{nat})$ in Figure~\ref{fig:1a} defined over the signature $\Sigma_2:=  \mathtt{nat}=^{1}_{\mu} \oplus\{ \mathit{z}:\mathtt{1}, \mathit{s}:\mathtt{nat}\}$. Process $\mathsf{Loop}$ spawns the following  infinite stream of messages: 
$x^\alpha.\mu_{\mathtt{nat}},\,x^{\alpha+1}.s,\, x^{\alpha+2}.\mu_{\mathtt{nat}},\, x^{\alpha+3}.s,\, \cdots.$
\end{example}

\begin{figure}[!t]
\begin{adjustwidth}{-0.8cm}{}
\begin{subfigure}[b]{0.5\textwidth}
{\small
\begin{align}
y \leftarrow \mathtt{Loop} \leftarrow x = \  & Ry.\mu_{nat};  && \%\ \textit{send}\ \mu_{nat}\  \textit{to right} \notag\\
& \phantom{low}  Ry.s;  && \%\ \textit{send label}\ \mathit{s} \  \textit{to right} \notag\\
& \phantom{low s} y \leftarrow \mathtt{Loop} \leftarrow x &&{\color{red} \%\ \textit{recursive call}}\ \notag
\end{align}}
\vspace{-18pt}
    \caption{Process Loop} \label{fig:1a}
\vspace{-20pt}%
\end{subfigure}
\hspace{5mm}%
\begin{subfigure}[b]{0.5\textwidth}
{\small\[\begin{aligned}
 \Sigma_3:=\ & \mathsf{ack}=^{1}_{\mu} \oplus\{ \mathit{ack}:\mathsf{astream}\},\\
& \mathsf{astream}=^{2}_{\nu} \& \{\mathit {head}: \mathsf{ack}, \ \ \mathit{tail}: \mathsf{astream}\}
 \end{aligned}\]}
\vspace{-18pt}
\caption{Signature $\Sigma_3$} \label{fig:1c}
\vspace{-6pt}
\end{subfigure}
\begin{subfigure}[b]{0.5\textwidth}
{\small \begin{align}
 & y\leftarrow \mathtt{Pong} \leftarrow w=&& \notag\\ & \phantom{s} Lw.\nu_{astream};&&  \% \  \textit{send}\ \mathit{\nu_{astream}} \ \textit{to } w \notag\\& \phantom{s} Lw.\mathit{head};&&  \% \  \textit{send label}\ \mathit{head} \ \textit{to } w \notag\\ & \phantom{s} \mathbf{case}\, Lw\ (\mu_{ack} \Rightarrow && \% \  \textit{receive}\ \mathit{\mu_{ack}} \ \textit{from } w \notag\\ &\phantom{sma} \mathbf{case}\, Lw\ ( && \% \  \textit{receive a label from }w \notag \\ & \phantom{sma s} y \leftarrow  \mathtt{Pong} \leftarrow w))&& {\color{red} \% \  \textit{recursive call}} \notag \end{align}}
 \vspace{-20pt}
\caption{Ping and Pong processes} \label{fig:1b}
\end{subfigure}
\begin{subfigure}[b]{0.5\textwidth}
{\small \begin{align}
 & w \leftarrow \mathtt{Ping} \leftarrow x= && \notag \\ &  \phantom{s} \mathbf{case}\, Rw\  (\nu_{astream} \Rightarrow && \% \  \textit{receive } \mathit{\nu_{astream}} \textit{from } w \notag\\ & \phantom{sma} \mathbf{case}\, Rw\ ( &&\notag \% \  \textit{receive a label from }w \\ &\phantom{small} \mathit{head} \Rightarrow Rw.\mu_{ack};&&   \% \  \textit{send}\ \mu_{ack} \ \textit{to }w \notag\\ & \phantom{small sp}Rw.\mathit{ack}; \notag && \% \  \textit{send label}\ \mathit{ack}\ \textit{to }w\\ & \phantom{small sp} w\leftarrow \mathtt{Ping}\leftarrow x && {\color{red} \% \  \textit{recursive call}} \notag \\ & \phantom{smal} \mid \mathit{tail} \Rightarrow w \leftarrow  \mathtt{Ping}\leftarrow x)) \notag&& {\% \  \textit{recursive call}}
\end{align}}
\end{subfigure}
\vspace{-20pt}
    \caption{}
\vspace{-20pt}
    \label{fig:loop}
\end{adjustwidth}
\end{figure}

Recall that our goal is to identify processes that terminates either in an empty configuration or one attempting to receive along an external channel. Obviously, process $\mathsf{Loop}$ does not have this property; it keeps spawning a message along channel $x$ while ignoring the channel $y$ on its left. This is not surprising as $\mathsf{Loop}$ does not receive any unfolding messages before calling itself recursively. However, receiving an unfolding message before each recursive call is not enough to guarantee the strong progress property. Consider processes $x{:}1\vdash\mathsf{Ping}::(w{:}\mathtt{astream})$ and $w{:}\mathtt{astream}\vdash\mathsf{Pong}::(y{:}1)$ in Figure \ref{fig:1b} defined over signature $\Sigma_3$ (Figure \ref{fig:1c}).  
$\mathsf{ack}$ is a type with {\it positive} polarity that, upon unfolding, describes a protocol requiring an {\it acknowledgment} message to be sent to the right (or be received from the left).  $\mathsf{astream}$ is a type with {\it negative} polarity of a potentially infinite stream where its $\mathit{head}$ is always followed by an acknowledgement. $\mathsf{Ping}$ and $\mathsf{Pong}$ are connected along their private channel $w$. 
$\mathsf{Ping}$ receives an $\mathtt{astream}$ unfolding message followed by a request for the head from $\mathsf{Pong}$. Next, $\mathsf{Pong}$ receives an $\mathtt{ack}$ unfolding message followed by an acknowledgement from $\mathsf{Ping}$ and they both call themselves recursively. The back and forth communication along $w$ continues indefinitely while the external channels $x$ and $y$ are ignored. 

The key to avoiding such indefinite internal exchanges is to use priorities of the fixed point unfolding messages and define $\mu$- and $\nu$-traces over process typing derivations. \emph{Receiving} an unfolding message ($\mu$L or $\nu$R) along channel $x$ ensures that a recursive call is guarded as long as the process \emph{does not send} an unfolding message ($\nu$L or $\mu$R) with \emph{higher priority} along the same channel. With this definition $\mathsf{Pong}$ is a guarded process but not $\mathsf{Ping}$.

Similar to  {\small $\mathit{FIMALL}^{\infty}_{\mu,\nu}$}, the set $\Omega$ in Figure~\ref{fig:stp-order} tracks sending and receiving fixed point unfolding messages. We use the overloaded operator $\mathsf{snap}$ to define $\mu$- and $\nu$-traces over process typing derivations (see Appendix~\ref{app:guard}). The guard condition for processes is defined similarly in Definition~\ref{def:valproc}. In Section \ref{sec:stronprogress} we show that a guarded process has strong progress.
\begin{definition}[Guard condition for processes]\label{def:valproc} A program defined over signature $\mathbf{\Sigma}$ and the set of process definitions $V$ is \emph{guarded} if any infinite branch in the derivation of $\bar{x}^\alpha:\omega \vdash y \leftarrow \mathsf{X} \leftarrow x ::y^\beta:\mathtt{B}$ for every $\bar{x}:\omega \vdash \mathsf{X}=\mathsf{P}_{\bar{x},y} ::y:\mathtt{B} \in V$ is either a left $\mu$-trace or a right $\nu$-trace.
\end{definition}

The calculus introduced in this section for processes is different from the metalogic infinitary calculus in Sections \ref{sec:metalogic} and \ref{sec:validitycondition}. Among other things, the metalogic is \emph{synchronous} as implied by its cut elimination procedure. While, the session type calculus is asynchronous in its operational semantics.  Moreover, the properties that are relevant for the two calculi are different: \emph{standard cut elimination} for an infinitary system in the metalogic and a specific operational property of \emph{strong progress} in the object language.

\begin{figure*}[t!]
\begin{center}
{\small\[
\begin{tabular}{c c}
\infer[\msc{Id}]{x^\alpha: \mathtt{A} \vdash_{\Omega} y \leftarrow x :: (y^{\beta}: \mathtt{A})}{} & \infer[\msc{Cut}^{w}]{ \bar{x}^{\alpha}: \omega \vdash_{\Omega}  ((w:\mathtt{A}) \leftarrow \mathsf{P}_{w} ; \mathsf{Q}_{w}) :: (y^{\beta}: \mathtt{C})}{ \bar{x}^{\alpha}:  \omega \vdash_{\Omega} \mathsf{P}_{w} ::(w^\eta:\mathtt{A}) & w^\eta: \mathtt{A} \vdash_{\Omega} \mathsf{Q}_{w} :: (y^{\beta}: \mathtt{C})}\\
\infer[1R]{\cdot \vdash_{\Omega} \mathbf{close}\, Ry :: (y^{\beta}: 1)}{} & \infer[1L]{x^{\alpha}: 1 \vdash_{\Omega} \mathbf{wait}\, Lx;\mathsf{Q} :: (y^{\beta}: \mathtt{A})}{  \cdot \vdash_{\Omega} \mathsf{Q} :: (y^{\beta}: \mathtt{A})}\\
    \infer[\oplus R]{\bar{x}^{\alpha}:\omega \vdash_{\Omega} Ry.k; \mathsf{P} :: (y^{\beta}: \oplus\{\ell:\mathtt{A}_{\ell}\}_{\ell \in L})}{\bar{x}^{\alpha}: \omega \vdash_{\Omega \cup \{y^\beta=y^{\beta+1}\}} \mathsf{P} :: (y^{\beta+1}: \mathtt{A}_{k}) \quad (k \in L)} &  \infer[\oplus L]{x^{\alpha}:\oplus\{ \ell:\mathtt{A}_\ell \}_{ \ell \in L} \vdash_{\Omega} \mathbf{case}\, Lx \ (\ell\Rightarrow \mathsf{P}_{\ell}):: (y^{\beta}: \mathtt{C})}{\forall \ell\in L \quad x^{\alpha+1}:\mathtt{A}_{\ell} \vdash_{\Omega\cup \{x^\alpha=x^{\alpha+1}\}} \mathsf{P}_\ell :: (y^{\beta}:\mathtt{C})}\\
    \infer[\& R]{\bar{x}^{\alpha}: \omega \vdash_{\Omega} \mathbf{case}\, Ry\ (\ell \Rightarrow \mathsf{P}_\ell) :: (y^{\beta}: \& \{\ell:\mathtt{A}_\ell\}_{\ell \in L})}{\forall \ell\in L \quad \bar{x}^{\alpha}: \omega \vdash_{\Omega\cup \{y^\beta=y^{\beta+1}\}} \mathsf{P}_\ell :: (y^{\beta+1}:\mathtt{A}_{\ell})} & \infer[\& L]{x^{\alpha}: \&\{ \ell:\mathtt{A}_l \}_{ \ell \in L} \vdash_{\Omega} Lx.k; \mathsf{P} :: (y^{\beta}:\mathtt{C})}{k\in L \quad x^{\alpha+1}:  \mathtt{A}_{k} \vdash_{\Omega\cup \{x^\alpha=x^{\alpha+1}\}}  \mathsf{P} :: (y^{\beta}:\mathtt{C})}\\
\end{tabular}
\]}
\[\begin{tabular}{c c}
  \infer[\mu R]{ \bar{x}^{\alpha}: \omega \vdash_{\Omega} Ry.\mu_t; \mathsf{P}_{y^{\beta}} :: (y^{\beta}:t)}{\deduce{  \bar{x}^{\alpha}: \omega \vdash_{\Omega'} \mathsf{P}_{y} :: (y^{\beta+1}:\mathtt{A})}{\Omega'= \Omega \cup \{(y^{\beta})_{j} =(y^{\beta+1})_{j} \mid j\neq i\}} & t=^i_{\mu}\mathtt{A} & }   & \infer[\mu L]{x^{\alpha}: t \vdash_{\Omega} \mathbf{case}\, Lx\ (\mu_{t} \Rightarrow \mathsf{Q}_{x^{\alpha}}):: (y^{\beta}: \mathtt{C})}{\deduce{x^{\alpha+1}: \mathtt{A} \vdash_{\Omega'} \mathsf{Q}_{x^{\alpha+1}} :: (y^{\beta}:\mathtt{C})}{ \Omega'=\Omega \cup \{x^{\alpha+1}_{i} < x^{\alpha}_{i}\} \cup \{x^{\alpha+1}_{j} =x^{\alpha}_{j} \mid j\neq i\} }  &  t=^i_{\mu} \mathtt{A} }
\end{tabular}\]
\[\begin{tabular}{cc}
  \infer[\nu R]{\bar{x}^{\alpha}: \omega \vdash_{\Omega} \mathbf{case}\, Ry \ (\nu_t \Rightarrow \mathsf{P}) :: (y^{\beta}: t)}{\deduce{\bar{x}^{\alpha}: \omega \vdash_{\Omega'} \mathsf{P} :: (y^{\beta+1}: \mathtt{A})}{\Omega'= \Omega \cup \{y^{\beta+1}_{i} < y^{\beta}_{i}\} \cup \{y^{\beta+1}_{j} = y^{\beta}_{j} \mid i\neq j\}}  & t=^i_{\nu}\mathtt{A}  } & \infer[\nu L]{x^{\alpha}: t \vdash_{\Omega} Lx.\nu_{t}; \mathsf{Q}:: (y^{\beta}: \mathtt{C})}{\deduce{x^{\alpha+1}: \mathtt{A} \vdash_{\Omega'} \mathsf{Q} :: (y^{\beta}: \mathtt{C})} {\Omega'=\Omega \cup \{(x^{\alpha+1})_{j} =(x^{\alpha})_{j} \mid j\neq i \}} & t=^i_{\nu} \mathtt{A} }  \\
\end{tabular}\]
\vspace{0cm}
{\[\qquad \qquad \qquad \qquad \qquad \qquad\infer[\msc{Def}(X)]{\bar{x}^{\alpha}: \omega \vdash_{\Omega} y \leftarrow \mathsf{X} \leftarrow \bar{x}:: (y^{\beta}: \mathtt{C})}{\bar{x}^{\alpha}: \omega \vdash_{\Omega} \mathsf{P}_{\bar{x}, y} :: (y^{\beta}: \mathtt{C}) & \bar{u}:\omega \vdash \mathsf{X}=\mathsf{P}_{\bar{u},w} :: (w:\mathtt{C}) \in V }\]}
\end{center}
 \vspace{-15pt}
\caption{Infinitary typing rules for processes with an ordering on channels.}
\label{fig:stp-order}
\vspace{-20pt}
\end{figure*}

\section{Asynchronous Dynamics}
In this section, we follow the approach of processes-as-formulas to provide an asynchronous semantics for session-typed processes. We express the computational behaviour of configurations as a recursive predicate in our metalogic. A configuration $\mathcal{C}$ is a list of processes that communicate with each other along their private channels. It is defined with the grammar ${\mathcal{C} ::= \cdot \mid \mathsf{P} \mid (\mathcal{C}_1 \mid_{x: \mathtt{A}}\ \mathcal{C}_2)}$, where $\mid$ is an associative, noncommutative operator and $(\cdot)$ is the unit. The type checking rules for configurations
$\bar{x}^\alpha: \omega \Vdash \mathcal{C}:: (y^\beta:\mathtt{B})$ are:
{\small\[
\begin{tabular}{c c  c }
\infer[\msc{emp}]{x^\alpha:\mathtt{A} \Vdash \cdot :: (x^\alpha:\mathtt{A}) }{} &  \infer[\msc{comp}]{\bar{x}^\alpha:\omega \Vdash \mathcal{C}_1 |_{z: \mathtt{A}} \ \mathcal{C}_2 :: (y^\beta:\mathtt{B})}{\bar{x}^\alpha:\omega \Vdash \mathcal{C}_1:: (z^\eta:\mathtt{A}) & z^\eta:\mathtt{A} \Vdash \mathcal{C}_2 :: (y^\beta:\mathtt{B})} &
\infer[\msc{proc}]{\bar{x^\alpha}:\omega \Vdash \mathsf{P} :: (y^\beta:\mathtt{B})}{\bar{x}^\alpha:\omega \vdash \mathsf{P} :: (y^\beta:\mathtt{B})} 
\end{tabular}
\]}
 An asynchronous computational semantics for configuration $(\bar{x}^\alpha:\omega) \Vdash \mathcal{C}::(y^\beta:\mathtt{B})$ is defined in Figure~\ref{fig:comp} under processes-as-formulas as the recursive predicate $\mathsf{Cfg}_{x^\alpha:\omega, y^\beta:\mathtt{B}}(\mathcal{C})$. We use pattern matching to define $\mathsf{Cfg}$. This predicate is parameterized by the left ($x^\alpha$) and right ($y^\beta$) endpoints of the configuration. For the sake of brevity, we drop these parameters for some cases in Figure~\ref{fig:comp}. For those cases we provide in the last column the type and generation of the channel along which the communication takes place. The interested reader is referred to Appendix~\ref{app:async} for a more detailed definition of $\mathsf{Cfg}$ with all its parameters, and conversion of the definition to a formula in the language of {\small $\mathit{FIMALL}^{\infty}_{\mu,\nu}$}.  
 
 The first two cases in the definition of $\mathsf{Cfg}$ reflect the rules for composition of configurations and an empty configuration. In the second line where two configurations are composed, a fresh channel $z^\eta$ is created. Channel $z^\eta$ is an internal channel in the composition of configurations $\mathcal{C}_1$ and $\mathcal{C}_2$ and is used by them to communicate with each other. Channels $\bar{x}^\alpha$ and $y^\beta$ are the externals channels of this composition.

The rest of the cases in Figure~\ref{fig:comp} refer to the configuration consists of a single process $(\bar{x}^\alpha:\omega) \vdash \mathsf{P}::(y^\beta:\mathtt{B})$. For identity (row 3) we put the generational channels to be equal to each other. Cut (row 4) spawns a new process $\mathsf{Q}_1$ offering along a fresh variable $z^\eta$ and continues as $\mathsf{Q}_2$ which is using the resource offered along $z^\eta$. Processes $\mathsf{Q}_1$ and $\mathsf{Q}_2$ communicate along their private channel $z^\eta$. 

The definition of predicate in rows 5-14 captures the operational meaning of session types presented in Table~\ref{tab:session_types}.
For the cases in which the process sends a message along a channel (rows 5,8,9,12,13), we first declare the message and then proceed the computation with the rest of the process. We use predicate $\mathsf{Msg}(x^\alpha.b)$ to formalize the declaration of message $b$ along $x^\alpha$. 
In the cases where the process needs to receive a message to continue (rows 6,7,10,11,14), the predicate is defined as a conjunction of the possible continuations. The definition may proceed with each label provided that the label is declared via a message predicate {$\mathsf{Msg}$}. $\mathsf{Cfg}$ provides an asynchronous semantic: a process spawns a message via $\mathsf{Msg}$ formula and continues its computation. The message may sit around for a while until another process receives it.

\begin{figure}[h]
\newcommand{\semi}{\mathrel{;}}
\small  \centering
 {\begin{equation*}
\hspace{-1cm} 
\begin{array}{llcll}
1. &\mathsf{Cfg}_{x^\alpha:\mathtt{A},x^\alpha:\mathtt{A}} (\mathsf{\cdot}) & =^{\mathbf{n}+2}_{\mu}& 1 & \text{empty}\\
2. &\mathsf{Cfg}_{\bar{x}^\alpha:\omega,y^\beta:\mathtt{B}} (\mathcal{C}_1|_{z:\mathtt{C}}\mathcal{C}_2) & =^{\mathbf{n}+2}_{\mu}& \exists z. \exists \eta.  \mathsf{Cfg}_{\bar{x}^\alpha:\omega,z^\eta:\mathtt{C}}(\mathcal{C}_1) \otimes  \mathsf{Cfg}_{z^\eta:\mathtt{C},y^\beta:\mathtt{B}}(\mathcal{C}_2) & \text{composition}\\
3.&\mathsf{Cfg}_{x^\alpha:\mathtt{A},y^\beta:\mathtt{A}} (y\leftarrow x) & =^{\mathbf{n}+2}_{\mu}& (x^\alpha=y^\beta)& \text{identity}\\
4.&\mathsf{Cfg}_{\bar{x}^\alpha:\omega,y^\beta:\mathtt{B}}( (z:\mathtt{C}) \leftarrow \mathsf{Q}_1 \semi \mathsf{Q}_2) & =^{\mathbf{n}+2}_{\mu} &  \exists z. \exists \eta.  \mathsf{Cfg}_{\bar{x}^\alpha:A,z^\eta:\mathtt{C}}(\mathsf{Q}_1) \otimes  \mathsf{Cfg}_{z^\eta:\mathtt{C},y^\beta:\mathtt{B}}(\mathsf{Q}_2) & \text{cut}\\
5.&\mathsf{Cfg}(\mathbf{close} Ry)& =^{\mathbf{n}+2}_{\mu}&  \mathsf{Msg}(y^{\beta}.\mbox{closed}) \otimes 1 & y^\beta{:}1\\
6.&\mathsf{Cfg}(\mathbf{wait} Lx; \mathsf{Q})& =^{\mathbf{n}+2}_{\mu} &   \mathsf{Msg}(x^{\alpha}.\mbox{closed}) \multimap \mathsf{Cfg}(\mathsf{Q}) & x^\alpha{:}1\\
7.&\mathsf{Cfg}(\mathbf{case} Ry(\ell \Rightarrow \mathsf{Q}_\ell)_{\ell\in L})& =^{\mathbf{n}+2}_{\mu}& \&\{\ell:  \mathsf{Msg}(y^\beta.\ell) \multimap \mathsf{Cfg}(\mathsf{Q}_\ell)\}_{\ell \in L} & y^\beta{:}\&\{\ell:\mathtt{B}_\ell\}_{\ell \in L}\\
8.&\mathsf{Cfg}({Lx}.k;\mathsf{Q}) & =^{\mathbf{n}+2}_{\mu} &  \mathsf{Msg}(x^\alpha.k) \otimes \mathsf{Cfg}(\mathsf{Q})& x^\alpha{:}\&\{\ell:\mathtt{A}_\ell\}_{\ell \in L} \\
9.&\mathsf{Cfg}({Rw}.k;\mathsf{Q}) & =^{\mathbf{n}+2}_{\mu} &  \mathsf{Msg}(y^\beta.k) \otimes \mathsf{Cfg}(\mathsf{Q})& y^{\beta}{:} \oplus \{\ell:\mathtt{B}_\ell\}_{\ell \in L}\\
10.&\mathsf{Cfg}(\mathbf{case} Lx(\ell \Rightarrow \mathsf{Q}_\ell)_{\ell\in L})& =^{\mathbf{n}+2}_{\mu}& \&\{\ell:  \mathsf{Msg}(x^\alpha.\ell) \multimap \mathsf{Cfg}(\mathsf{Q}_\ell)\}_{\ell\in L} & x^\alpha{:}\oplus\{\ell:\mathtt{A}_\ell\}_{\ell \in L} \\
11.&\mathsf{Cfg}(\mathbf{case} Ry(\nu_\mathtt{t} \Rightarrow \mathsf{Q}))& =^{\mathbf{n}+2}_{\mu}&   \mathsf{Msg}(y^\beta.\nu_t) \multimap \mathsf{Cfg}(\mathsf{Q})& y^\beta{:}\mathtt{t}\qquad \qquad \mathtt{t}=^i_{\nu} \mathtt{C} \in \mathbf{\Sigma}\\ 
12.&\mathsf{Cfg}({Lx}.\nu_\mathtt{t};\mathsf{Q}) & =^{\mathbf{n}+2}_{\mu} &  \mathsf{Msg}(x^\alpha.\nu_t) \otimes \mathsf{Cfg}(\mathsf{Q}) &  x^\alpha{:}\mathtt{t} \qquad \qquad \mathtt{t}=^i_{\nu} \mathtt{C} \in \mathbf{\Sigma}\\
13.&\mathsf{Cfg}({Ry}.\mu_\mathtt{t};\mathsf{Q}) & =^{\mathbf{n}+2}_{\mu} &  \mathsf{Msg}(y^\beta.\mu_t) \otimes \mathsf{Cfg}(\mathsf{Q}) & y^\beta{:}\mathtt{t} \qquad \qquad \mathtt{t}=^i_{\mu} \mathtt{C} \in \mathbf{\Sigma}\\
14.&\mathsf{Cfg}(\mathbf{case} Lx(\mu_\mathtt{t} \Rightarrow \mathsf{Q}))& =^{\mathbf{n}+2}_{\mu}&    \mathsf{Msg}(x^\alpha.\mu_t) \multimap \mathsf{Cfg}(\mathsf{Q})& x^\alpha{:}\mathtt{t} \qquad \qquad \mathtt{t}=^i_{\mu} \mathtt{C} \in \mathbf{\Sigma}\\
15.&\mathsf{Cfg}_{\bar{x}^\alpha:\omega,y^\beta:\mathtt{B}}(\mathsf{Y}) & =^{\mathbf{n}+1}_{\nu}& \mathsf{Cfg}_{\bar{x}^\alpha:\omega, y^\beta:\mathtt{B}}(\mathsf{Q} [y/w, \overline{x}/\bar{u}] ) & \bar{u}:\omega \vdash \mathsf{Y}=\mathsf{\mathsf{Q}}:: (w:\mathtt{B}) \in V \end{array}
\end{equation*}
}    \vspace{-15pt}
    \caption{Definition of predicate $\mathsf{Cfg}$.}
    \label{fig:comp}
    \vspace{-10pt}
\end{figure}
In the last case a process variable is unfolded while instantiating the left and right channels $\bar{u}$ and $w$ in the process definition with proper channel names $\bar{x}$ and $y$, respectively. All cases of $\mathsf{Cfg}$ except the last one are defined recursively on a process with a smaller size; we define the predicate in these cases as a least fixed point. In the last case a process variable $\mathsf{Y}$ is replaced by its definition $\mathsf{Q}$ of a possibly larger size; accordingly the predicate is defined as a greatest fixed point in this case. Since in the presence of recursion the behavior of a configuration is defined coinductively, we put the priority of the $\nu$-term in the last case to be higher than $\mu$-terms in the other cases. The first $\mathbf{n}$ priorities in the signature are reserved for the recursive cases in the definition of the strong progress predicate which we will introduce in Section \ref{sec:stronprogress}.


\section{The strong progress property}\label{sec:stronprogress}

Strong progress in an asynchronous setting states that a program eventually terminates either in an empty configuration or one attempting to receive along an external channel. In Section \ref{sec:procs}, we showed that recursion destroys strong progress similar to the way adding circularity to a calculus breaks down cut elimination. In this section, we give a direct proof for the strong progress property of guarded programs with an \emph{asynchronous communication}. 


 Our goal is to formalize the concept of strong progress for a configuration of processes in the language of {\small $\mathit{FIMALL}^{\infty}_{\mu,\nu}$}. We first focus on a particular case in which the configuration is closed, i.e. it does not use any resources. We need to show the computational behavior of configuration $\cdot \Vdash\mathcal{C} ::(y^\beta:\mathtt{B})$  as defined in Figure \ref{fig:comp} ensures that its external channel $y^\beta{:}\mathtt{B}$ will be eventually closed or blocked by waiting to receive a message. In the later case as soon as the message becomes available $y^\beta$ evolves to its new generation $\beta+1$ of a correct type and the continuation $y^{\beta+1}$ has to maintain the same property. In Figure \ref{fig:[]} we define a predicate $[y^\beta:\mathtt{B}]$ to formalize this property. Similar to the definition of $\mathsf{Cfg}$ in Figure \ref{fig:comp} we use pattern matching here. We explain the conversion of this definition to a formula in the language of {\small $\mathit{FIMALL}^{\infty}_{\mu,\nu}$} in Appendix~\ref{app:derivation}.
 
 The first line in the definition of $[y^\alpha:\mathtt{B}]$ (Figure \ref{fig:[]}) corresponds to the case in which $y^\beta$ terminates; in this case we announce the formula that $y^\beta$ is closed ($\mathsf{Msg}(y^\beta.\mbox{closed})$) and proceed with the predicate $[\cdot:\cdot]$.  The predicate $[\cdot:\cdot]$ is defined as constant $1$. For positive types ($\oplus$ and positive fixed points) the provider declares a message of a correct type along $y^\beta$ while the continuation channel $y^{\beta+1}$ has to enjoy the strong progress property. For negative types ($\&$ and negative fixed points) the channel waits on a message along $y^\beta$. Upon receive of such message the continuation channel $y^{\beta+1}$ has to maintain the strong progress property.


The recursive definitions in all the cases except for the fixed points (rows 3 and 5) are defined based on a session type with a smaller size; the predicate is defined as a least fixed point in these cases. When the session type is a positive fixed point $\mathtt{t}$, the predicate is inductively defined on the structure of $\mathtt{t}$ and thus inherits its priority from $\mathtt{t}$. When the session type is a negative fixed point $\mathtt{t}=^i_{\nu}\mathtt{A}$, the predicate is defined as a greatest fixed point based on the structure of $\mathtt{t}$ and ensures a productive property: as desired by strong progress the channel $y^\beta$ is blocked until it receives a message and when the message is received the continuation $y^{\beta+1}$ continues to maintain the desired property. The priority of the predicate is inherited from its underlying session type $\mathtt{t}$. We will see the priority assigned to each case ensures that the if a program is guarded then the derivation for its strong progress is valid.


\begin{figure}
\small
    \centering
\[
\begin{array}{lcl}
{[y^\beta:1 ]} & =^{\mathbf{n}+2}_{\mu} & \mathsf{Msg}(y^\beta.\mbox{closed}) \otimes[\cdot: \cdot]\\
{[y^\beta: \&\{\ell:\mathtt{A}_\ell\}_{\ell\in L}]} & =^{\mathbf{n}+2}_{\mu} &  \&\{\ell:{ (\mathsf{Msg}(y^\beta.\ell) \multimap[y^{\beta+1}:\mathtt{A}_\ell])}\}_{\ell \in L} \\
 {[y^\beta: \mathtt{t}]} & =^i_{\nu} & ( \mathsf{Msg}(y^\beta.\nu_t) \multimap[y^{\beta+1}:\mathtt{A}]) \quad  \mathtt{t}=^{i}_{\nu}\mathtt{A}\in \mathbf{\Sigma}\\
{[y^\beta:\oplus \{\ell:\mathtt{A}_i\}]} & =^{\mathbf{n}+2}_{\mu} & \oplus\{\ell:(\mathsf{Msg}(y^\beta.\ell) \otimes [y^{\beta+1}:\mathtt{A}_\ell])\}_{\ell \in L} \\
{[y^\beta: \mathtt{t}]} & =^{i}_{\mu} & ( \mathsf{Msg}(y^\beta.\mu_t) \otimes[y^{\beta+1}:\mathtt{A}]) \quad  \mathtt{t}=^i_{\mu} \mathtt{A}\in \mathbf{\Sigma}\\
{[\cdot: \cdot]} & =^{\mathbf{n}+2}_{\mu} &  1\\
\end{array}
\]
 \vspace{-15pt}
    \caption{Definition of predicate $[ y^\beta:\mathtt{A} ]$.}
    \label{fig:[]}
    \vspace{-20pt}
\end{figure}

 Using the predicate defined in Figure~\ref{fig:[]}, strong progress for closed configuration $\cdot \vdash \mathcal{C}:: (y^\beta{:}B)$ is formalized as {\small$\mathsf{Cfg}_{\cdot,y^\beta:B}(\mathcal{C}) \multimap  [y^\beta:\mathtt{B}].$} We use $\mathcal{C} \in \llbracket \cdot \vdash y^\beta:\mathtt{B} \rrbracket$ as an abbreviation for this formula. The generalization of this predicate to an open configuration ${x^\alpha:\mathtt{A} \Vdash \mathcal{C} ::(y^\beta:\mathtt{B})}$ is straightforward. Configuration $\mathcal{C}$ satisfies strong progress if its composition with any closed configuration with strong progress has the strong progress property: 
 
 \noindent
 {\small$
     \forall X.\,( \forall z. \forall \eta.\, X \in \llbracket \cdot \vdash z^\eta{:}\mathtt{A} \rrbracket \multimap X\mid_{x{:}\mathtt{A}}\mathcal{C} \in \llbracket \cdot \vdash y^\beta{:}\mathtt{B} \rrbracket).$}
     
\noindent We put {\small$\mathcal{C} \in \llbracket x^\alpha{:}\mathtt{A} \vdash y^\beta{:}\mathtt{B} \rrbracket$} to abbreviate this formula.
\begin{lemma}\label{lem:derivation}
For a configuration of processes $\bar{x}^\alpha:\omega \vdash \mathcal{C} ::(y^\beta:\mathtt{B})$, there is a (possibly infinite) derivation for $ \mathcal{C} \in \llbracket \bar{x}^\alpha:\omega \vdash y^\beta:\mathtt{B} \rrbracket$.
\end{lemma}
\begin{proof}
 Here we provide a summary of the proof. For the complete proof see Appendix \ref{app:derivation}. To get the derivation we aim for, it is enough to build a circular derivation  $\star\, [x^\alpha:\mathtt{A}],  \mathsf{Cfg}_{x^\alpha:\mathtt{A},y^\beta:\mathtt{B}}(\mathcal{C}) \vdash  [y^\beta:\mathtt{B}]$ for an open configuration $x^\alpha:\mathtt{A} \vdash \mathcal{C}:: y^\beta{:}\mathtt{B}$, and derivation $ \star \;\mathsf{Cfg}_{\cdot,y^\beta:\mathtt{B}}(\mathcal{C}) \vdash  [y^\beta{:}\mathtt{B}]$ for a closed configuration $\cdot \vdash \mathcal{C}:: y^\beta:\mathtt{B}$. We provide a circular derivation for each possible pattern of $\mathcal{C}$. 
In all cases, a rule is applied on $[x^\alpha:\mathtt{A}]$ or $[y^\beta:\mathtt{B}]$ only if there is a process willing to receive or send a message along channels $x$ or $y$, respectively, in the configuration. 
\end{proof}

We need to show that when defined over a guarded program the derivation introduced in Lemma~\ref{lem:derivation} is a valid proof. 
Since a configuration is always finite, it is enough to prove validity of the annotated derivation for a single process $\mathsf{P}$. The idea is to introduce a validity-preserving bisimulation between the annotated derivation given in Lemma~\ref{lem:derivation} and the typing derivation of process $ \mathsf{P}$. The key is to assign a matching position variable $\mathbf{y}^\beta$ to every predicate $[y^\beta{:}\mathtt{t}]$ occurring in the derivation. A sketch of the bisimulation is given below. For more details see Appendix \ref{app:derivation}. 
{\small\[
\begin{tikzcd}
 \bar{x}^\alpha:\omega \vdash_{\Omega} \mathsf{P} :: (y^\beta:\mathtt{B}) \arrow[r, dash, "\mathcal{R}"] \arrow[d, hookrightarrow] &  \overline{\mathbf{x}^\alpha:[\bar{x}^\alpha:\omega]}, \mathbf{c}^\eta:\mathsf{Cfg}_{\bar{x}^\alpha:\omega, y^\beta:\mathtt{B}}(\mathsf{P})\vdash_{\Omega'}  \mathbf{y}^\beta:[y^\beta:\mathtt{B}] \arrow[d, Rightarrow, dashed] \\
  \bar{z}^\gamma:\omega' \vdash_{\Lambda} \mathsf{Q} :: (w^\delta:\mathtt{D}) \arrow[r, dash, dashed, "\mathcal{R}"'] & \overline{\mathbf{z}^\gamma:{[\bar{z}^\gamma:\omega']}}, \mathbf{d}^\theta:\mathsf{Cfg}_{\bar{z}^\gamma:\omega', w^\delta:\mathtt{D}}(\mathsf{Q})\vdash_{\Lambda'} \mathbf{w}^{\delta}: {[w^\delta:\mathtt{D}]}
\end{tikzcd}
\]}
Finally, we show that predicate $ \mathcal{C} \in \llbracket \bar{x}^\alpha:\omega \vdash y^\beta:\mathtt{B} \rrbracket$ is sound in the sense that its valid cut free proof ensures strong progress of $\mathcal{C}$.
\begin{theorem}\label{thm:strongprogress}
For the judgment $\bar{x}^\alpha:\omega\Vdash \mathcal{C} ::(y^\beta:\mathtt{B})$ where $\mathcal{C}$ is a configuration of guarded processes, there is a valid cut-free proof for $ \mathcal{C} \in \llbracket \bar{x}^\alpha:\omega \vdash y^\beta:\mathtt{B} \rrbracket$ in {\small $\mathit{FIMALL}^{\infty}_{\mu,\nu}$}. Validity of this proof ensures the strong progress property of $\mathcal{C}$.
\end{theorem}
 \begin{proof}
It is enough to show that the valid proof we built for {\small$\star {\,\overline{[\bar{x}^\alpha:\omega]}, \mathsf{Cfg}_{\bar{x}^\alpha:\omega, y^\beta:\mathtt{B}}(\mathcal{C})\vdash  [y^\beta:\mathtt{B}]}$} in {\small $\mathit{FIMALL}^{\infty}_{\mu,\nu}$} ensures the strong progress property of configuration $\mathcal{C}$. We run the cut elimination algorithm on the valid proof built above. The output is a cut-free valid proof. If there are no infinite branches in the proof then by linearity of the calculus we know that a rule is applied on $\overline{[\bar{x}^\alpha:\omega]}$ and $[y^\beta:\mathtt{B}]$. In the infinite case, recall that the predicate $\mathsf{Cfg}$ for a recursive process is defined coinductively; no subformula of it in the antecedents can be a part of an infinite $\mu$-trace. Thus a rule has to be applied on either $[\bar{x}^\alpha:\omega]$ or $[y^\beta:\mathtt{B}]$. By the structure of the $\star$ proof, a rule is applied on $[\bar{x}^\alpha:\omega]$ or $[y^\beta:\mathtt{B}]$ only if there is a process communicating along $\bar{x}^\alpha$ or $y^\beta$, respectively, in the configuration. In Appendix~\ref{app:derivation} we further show that $\mathcal{C}$ eventually waits to receive a message from one of its external channels.
\end{proof}

\section{Conclusion}
In this paper we introduced an infinitary sequent calculus for first order intuitionistic multiplicative additive linear logic with fixed points. Inspired by the work of Fortier and Santocanale \cite{Fortier13csl} we provide an algorithm to identify valid proofs among all infinite derivations. 

Our main motivation for introducing this calculus is  to give a direct proof for the strong progress property of  guarded binary session typed processes. The importance of a direct proof other than its elegance is that it can be  adapted for a more general guard condition on processes, e.g. a condition that takes into account the relation between a resource and the service that is offered by a process, without the need to prove cut elimination productivity for their underlying derivations. 

One recent approach in type theory integrates
induction and coinduction by pattern and copattern matching and explicit
well-founded induction on ordinals\cite{abel2016well}, following a number of earlier representations of induction and coinduction in type theory \cite{abel2013wellfounded}. The connection to this type theoretic approach is an interesting item for future research. A first step in this general direction was taken by Sprenger and Dam~\cite{sprenger2003structure}
who justify cyclic inductive proofs using inflationary iteration.

The results presented in this paper are proved for an asynchronous semantics, while the prior work~\cite{derakhshan2019circular} is based on a synchronous semantics. The strong progress proof in the prior work is based on the Curry-Howard correspondence (propositions as types) we built between recursive session types and infinitary linear logic, in which strong progress corresponds to the productivity of the cut-elimination algorithm introduced in~\cite{Fortier13csl}. In this work, we use the processes-as-formulas approach. The strong progress property is formalized as a formula, and the proof is given for it in a metalogic without appealing to a Curry-Howard correspondence.

\bibliography{ref}
\newpage
\appendix
\section{Appendix}\label{appendix}

\subsection{$\mathit{FIMALL}^{\infty}_{\mu,\nu}$}\label{app:logic}
A \emph{circular derivation} is the finite representation of an infinite one in which we can identify each open subgoal with an identical interior judgment. In the first order context we may need to use a substitution rule right before a circular edge to make the subgoal and interior judgment exactly identical \cite{brotherston2005cyclic}:

{\small\[ \infer[\mathtt{subst}_{\theta}]{\Gamma[\theta] \vdash B[\theta]}{\Gamma\vdash B}\]}%
We can transform a circular derivation to its underlying infinite derivation in a productive way by deleting the $\mathtt{subst}_{\theta}$ rule and the circular edge. We need to instantiate the derivation to which the circular edge pointed with substitution $\theta$. This instantiation exists and does not change the structure of the derivation. We prove the next lemma for the annotated system of Figure~\ref{fig:rules-2}. It is straightforward to see that it also holds for the system in Figure~\ref{fig:rules-1}.
\begin{lemma}[Substitution]\label{lem:subst}
For a valid derivation 
\[\infer{\Delta \vdash \mathbf{w}^\alpha:A}{\Pi}\]
in the infinite system and substitution $\theta$, there is a valid derivation for \[\infer{\Delta[\theta] \vdash \mathbf{w}^\alpha:A[\theta]}{\Pi[\theta]}\]
where $\Pi[\theta]$ is the whole derivation $\Pi$ or a prefix of it instantiated by $\theta$.
\end{lemma}
\begin{proof}
The proof is by coinduction on the structure of 
\[\Delta \vdash \mathbf{w}^\alpha:A.\]
The only interesting case is where we get to the $= L$ rule. \[\infer[=L]{\Gamma, s=t \vdash B}{ \deduce{\Gamma[\theta']\vdash B[\theta']}{\Pi'}& \forall \theta' \in \mathsf{mgu}(s,t)}\]
If the set {$\mathsf{mgu}(t[\theta],s[\theta])$}is empty then so is {$\Pi[\theta]$}. Otherwise if $\eta$ is the single element of {$ \mathsf{mgu}(t[\theta],s[\theta])$}, then for some substitution $\lambda$ we have {$\theta\eta=\theta'\lambda,$} and we can form the rest of derivation  for substitution $\lambda$ as {$\Pi'[\lambda]$} coinductively. 
\end{proof}

\subsection{Derivations for the correctness of Tower of Hanoi solution} \label{app:hanoi}

A circular derivation for correctness of Algorithm~\ref{alg:tower} is given below. The derivation for ${\color{blue}\star}$ is straightforward and {\bf finite}. The purple thread highlights the descendants of predicate $\mathsf{move}$ in $\dagger$. 

 {\footnotesize 
\begin{equation*}
   \hspace{-2.9cm}  \inferLineSkip=2pt
  \infer[\color{Brown}{\mu L}]{\dagger \,\mathsf{p}(s,I\,L_{s}),\mathsf{p}(t,L_t), \mathsf{p}(a,L_a),\mathsf{c}(m), {\color{Plum}\mathsf{m}(s,t,a,I,m)}\vdash \mathsf{p}(s,L_{s})\otimes\mathsf{p}(t,I\,L_t)\otimes \mathsf{p}(a,L_a)\otimes \mathsf{c}(2^{|I|}+m-1)}{\infer[\oplus L]{\mathsf{p}(s,I\,L_{s}),\mathsf{p}(t,L_t), \mathsf{p}(a,L_a),\mathsf{c}(m), {\color{Plum}\oplus\{\mathit{next}: \exists k, I'.\,(I= I'\,k)  \otimes \cdots, \mathit{done}: I=\epsilon\}}\vdash \mathsf{p}(s,L_{s})\otimes\mathsf{p}(t,I\,L_t)\otimes \mathsf{p}(a,L_a) \otimes \mathsf{c}(2^{|I|}+m-1)}{\infer[\exists L]{\mathsf{p}(s,I\,L_{s}),\mathsf{p}(t,L_t), \mathsf{p}(a,L_a), \mathsf{c}(m), {\color{Plum}\exists k,I'.\, (I=I'\,k)  \otimes \cdots} \vdash \mathsf{p}(s,L_{s})\otimes\mathsf{p}(t,I\,L_t)\otimes \mathsf{p}(a,L_a) \otimes \mathsf{c}(2^{|I|}+m-1)}{\infer[\otimes L]{\mathsf{p}(s,I\,L_{s}),\mathsf{p}(t,L_t), \mathsf{p}(a,L_a),\mathsf{c}(m), {\color{Plum} (I=I'\,k)  \otimes \mathsf{m}(s,a,t,I') \otimes \cdots} \vdash \mathsf{p}(s,L_{s})\otimes\mathsf{p}(t,I\,L_t)\otimes \mathsf{p}(a,L_a)\otimes \mathsf{c}(2^{|I|}+m-1)}{\infer[=L]{\mathsf{p}(s,I\,L_{s}),\mathsf{p}(t,L_t), \mathsf{p}(a,L_a), \mathsf{c}(m), {\color{Plum} (I=I'\,k), \mathsf{m}(s,a,t,I') \otimes \cdots} \vdash \mathsf{p}(s,L_{s})\otimes\mathsf{p}(t,I\,L_t)\otimes \mathsf{p}(a,L_a)\otimes \mathsf{c}(2^{|I|}+m-1)}{\infer[\otimes L]{\mathsf{p}(s,I'\,k\,L_{s}),\mathsf{p}(t,L_t), \mathsf{p}(a,L_a), \mathsf{c}(m),   {\color{Plum}\mathsf{m}(s,a,t,I')\otimes\mathsf{p\_p}(s,t,k, I', m)\otimes \mathsf{m}(a,t,s,I')} \vdash \mathsf{p}(s,L_{s})\otimes\mathsf{p}(t,I'\,k\,L_t)\otimes \mathsf{p}(a,L_a)\otimes \mathsf{c}(2^{|I'k|}+m-1)}{\infer[\msc{Cut} \textit{(x2)}]{\mathsf{p}(s,I'\,k\,L_{s}),\mathsf{p}(t,L_t), \mathsf{p}(a,L_a), \mathsf{c}(m), {\color{Plum}\mathsf{m}(s,a,t,I',m), \mathsf{p\_p}(s,t,k,I',m), \mathsf{m}(a,t,s,I',2^{|I'|}{+}m)} \vdash \mathsf{p}(s,L_{s})\otimes\mathsf{p}(t,I'\,k\,L_t)\otimes \mathsf{p}(a,L_a) \otimes \mathsf{c}(2^{|I'k|}{+}m{-}1)}{\deduce{{\color{OliveGreen}\dagger_2}\, \mathsf{p}(s,L_{s})\otimes\mathsf{p}(t,k\,L_t)\otimes \mathsf{p}(a,I'\,L_a)\otimes \mathsf{c}(2^{|I'|}{+}m),   {\color{Plum}\mathsf{m}(a,t,s,I',m)} \vdash  \mathsf{p}(s,L_{s})\otimes\mathsf{p}(t,I'\,k\,L_t)\otimes \mathsf{p}(a,L_a)\otimes \mathsf{c}(2^{|I'k|}{+}m{-}1)}{   \deduce{{\color{blue}\star}\,\mathsf{p}(,k\,L_{s}) \otimes \mathsf{p}(t,L_t)\otimes \mathsf{p}(a,I'\,L_a) \otimes \mathsf{c}(2^{|I'|}{+}m{-}1),   {\color{Plum}\mathsf{p\_p}(s,t,k, I',m)} \vdash  \mathsf{p}(s,L_{s})\otimes\mathsf{p}(t,k\,L_t)\otimes \mathsf{p}(a,I'\,L_a)\otimes \mathsf{c}(2^{|I'|}{+}m)}{{\color{red}\dagger_1} \,\mathsf{p}(s,I'\,k\,L_{s}),\mathsf{p}(t,L_t), \mathsf{p}(a,L_a), \mathsf{c}(m),   {\color{Plum}\mathsf{m}(s,a,t,I',m)} \vdash  \mathsf{p}(s,k\,L_{s})\otimes\mathsf{p}(t,L_t)\otimes \mathsf{p}(a,I'\,L_a) \otimes \mathsf{c}(2^{|I'|}{+}m{-}1) } }}}}}} &\hspace{-80pt} \infer*{\cdots (I=\epsilon) \vdash \cdots}{}}}
\end{equation*}}

An invalid circular derivation for correctness of infinitary Tower of Hanoi solution is given below. The derivation for ${\dagger_5}$ is straightforward and {\bf finite}. The purple thread highlights the descendants of predicate $\mathsf{move}$ in $\dagger$. 

 {\footnotesize   \[
 \hspace{-2.5cm}
    \infer[\mu L]{\dagger}{\infer*{}{\infer[\msc{Cut}]{}{{\color{red}\dagger_1} & {\color{blue}\star} &  {\color{OliveGreen}\dagger_2} &\hspace{-1.8cm} \infer[\nu L]{{\color{Orange}\dagger_3}\, \mathsf{p}(s,L_{s})\otimes\mathsf{p}(t,I'\,k\,L_t)\otimes \mathsf{p}(a, L_a)\otimes \mathsf{c}(2^{|I_s|}{+}m{-}1),   {\color{Plum}\mathsf{s}(t,s,a,I,2^{|I'k|}{+}m{-}1)} \vdash  \mathsf{p}(s,L_{s})\otimes\mathsf{p}(t,I'\,k\,L_t)\otimes \mathsf{p}(a,L_a)\otimes \mathsf{c}(2^{|I'k|}{+}m{-}1)}{\infer[\oplus L]{\mathsf{p}(s,L_{s})\otimes\mathsf{p}(t,I'\,k\,L_t)\otimes \mathsf{p}(a, L_a)\otimes \mathsf{c}(2^{|I'k|}{+}m{-}1),   {\color{Plum}\oplus \{\mathit{restart}{:} I'\,k= I_s \otimes \cdots, \mathit{term}{:}1}\} \vdash  \mathsf{p}(s,L_{s})\otimes\mathsf{p}(t,I'\,k\,L_t)\otimes \mathsf{p}(a,L_a)\otimes \mathsf{c}(2^{|I'k|}{+}m{-}1)}{\infer[\otimes L, = L]{\mathsf{p}(s,L_{s})\otimes\mathsf{p}(t,I'\,k\,L_t)\otimes \mathsf{p}(a, L_a)\otimes \mathsf{c}(2^{|I'k|}{+}m{-}1),   {\color{Plum} I'\,k= I_s \otimes \cdots} \vdash  \mathsf{p}(s,L_{s})\otimes\mathsf{p}(t,I'\,k\,L_t)\otimes \mathsf{p}(a,L_a)\otimes \mathsf{c}(2^{|I'k|}{+}m{-}1)}{\infer[\forall L]{\mathsf{p}(s,L_{s})\otimes\mathsf{p}(t,I_s\,L_t)\otimes \mathsf{p}(a, L_a)\otimes \mathsf{c}(2^{|I_s|}{+}m{-}1),   {\color{Plum} \forall L,L'. (\cdots) \multimap (\cdots)} \vdash  \mathsf{p}(s,L_{s})\otimes\mathsf{p}(t,I_s\,L_t)\otimes \mathsf{p}(a,L_a)\otimes \mathsf{c}(2^{|I_s|}{+}m{-}1)}{\infer[\multimap L]{\mathsf{p}(s,L_{s})\otimes\mathsf{p}(t,I_s\,L_t)\otimes \mathsf{p}(a, L_a)\otimes \mathsf{c}(2^{|I_s|}{+}m{-}1),   {\color{Plum} (\mathsf{p}(s,L) \otimes \cdots) \multimap (\mathsf{p}(s,I_s\, L) \otimes \cdots)} \vdash  \mathsf{p}(s,L_{s})\otimes\mathsf{p}(t,I_s\,L_t)\otimes \mathsf{p}(a,L_a)\otimes \mathsf{c}(2^{|I_s|}{+}m{-}1)}{\deduce{\dagger_4\,\mathsf{p}(a, L_a),{\color{Plum} \mathsf{p}(s,I_s\, L), \mathsf{p}(t,L'), \mathsf{c}(0), \mathsf{m}(s,t,a,I_s,0)} \vdash  \mathsf{p}(s,L_{s})\otimes\mathsf{p}(t,I_s\,L_t)\otimes \mathsf{p}(a,L_a)\otimes \mathsf{c}(2^{|I_s|}{+}m{-}1)} {\dagger_5\, {\mathsf{p}(s,L_{s}), \mathsf{p}(t,I_s\,L_t),  \mathsf{c}(2^{|I_s|}{+}m{-}1) \vdash \mathsf{p}(s,L_{s})\otimes\mathsf{p}(t,I_s\,L_t)\otimes \mathsf{c}(2^{|I_s|}{+}m{-}1)}}}}} &\hspace{-1cm} \infer*{\cdots, {\color{Plum} 1} \vdash \cdots}{}}}}}}
    \]}

\subsection{Cut elimination}\label{app:cut}


We adapt Fortier and Santocanale's \cite{Fortier13csl} cut elimination algorithm to {\small $\mathit{FIMALL}^{\infty}_{\mu,\nu}$} and prove its productivity for valid derivations. The algorithm receives an infinite proof as an input and outputs a cut-free infinite proof. 
Since we are dealing with infinite derivations, to make the algorithm productive\footnote{An algorithm is productive if every piece of its output is generated in a finite number of steps.} we need to push every cut away from the root with a lazy strategy. With this strategy we may need to permute two consecutive cuts which results into a loop. To overcome this problem, similarly to Fortier and Santocanale and also Baelde et al. \cite{baelde2016infinitary} we generalize binary cuts to $n$-ary cuts using the notion of a \emph{branching tape}. 

\begin{definition}\label{def:tape}
A \emph{branching tape} $\mathcal{C}$ is a finite list of sequents\footnote{For brevity we elide the set $\Omega$ in the judgments.} $\Delta \vdash \mathbf{w}^\beta: A$, such that
\begin{itemize}
    \item Every two judgments $\Delta \vdash \mathbf{w}^\beta: A$ and $\Delta' \vdash \mathbf{w}'^{\beta'}: A'$ on the tape share at most one position variable $\mathbf{z}^\alpha:B$. If they share such a position variable, we call them connected. Moreover, assuming that $\Delta \vdash \mathbf{w}^\beta: A$ appears before $\Delta' \vdash \mathbf{w}'^{\beta'}: A'$ on the list, we have  $\mathbf{z}^\alpha:B\in \Delta'$ and $\mathbf{z}^\alpha:B= \mathbf{w}^\beta:A$.
    \item Each position variable $\mathbf{z}^\beta$ appears at most twice in a tape and if it appears more than once it  connects two judgments.
    \item Every tape is connected.
\end{itemize}
The \emph{conclusion} $\mathsf{conc}_{\mathcal{M}}$ of a branching tape $\mathcal{M}$ is a sequent $\Delta \vdash \mathbf{x}^\alpha:A$ such that 
\begin{itemize}
    \item there is a sequent $\Delta' \vdash \mathbf{x}^\alpha:A$ in the tape that $\mathbf{x}^\alpha:A$ does not connect it to any other sequent in the tape. 
    \item For every $\mathbf{y}^\beta:B\in \Delta$ there is a sequent $\Delta', \mathbf{y}^\beta:B \vdash \mathbf{z}^\gamma:C$ on the tape such that $\mathbf{y}^\beta:B$ does not connect it to any other sequent in the tape. 
\end{itemize}
We call $\Delta$ the set of \emph{leftmost formulas} of $\mathcal{M}$: $\mathsf{lft}(\mathcal{M})$. And $x^\alpha:A$ is the \emph{rightmost formula} of tape $\mathcal{M}$: $\mathsf{rgt}(\mathcal{M})$. By definition a branching tape is connected and acyclic. Therefore its conclusion always exists and is unique. An $n$-ary cut is a rule formed from a tape $\mathcal{M}$ and its conclusion $\mathsf{conc}_{\mathcal{M}}$:\qquad
$\infer[nCut]{\mathsf{conc}_{\mathcal{M}}}{\mathcal{M}}$
\end{definition}     

We generalize  Fortier and Santocanale's set of  primitive operations to account for {\small $\mathit{FIMALL}^{\infty}_{\mu,\nu}$}. They closely resemble the reduction rules given by Doumane \cite{doumane2017infinitary}. Figure~\ref{fig:cred} depicts a few interesting internal ($\mathsf{PRd}$) and external ($\mathsf{Flip}$) reductions, identity elimination, and merging a cut.  It is straightforward to adapt the rest of reduction steps from the previous works \cite{doumane2017infinitary,Fortier13csl}.

{\small \begin{figure*}
    \centering

\[\hspace{-2cm}{\small
\begin{array}{lcl}
    \infer[\mathit{nCut}]{\Delta \vdash \mathbf{v}:C}{\mathcal{C}_1 & \infer[\exists R]{ \Delta' \vdash \mathbf{z}^{\beta}:\exists x. P(x)}{ \Delta'\vdash \mathbf{z}^{\beta}:P(t) }& \mathcal{C}_2 &\infer[\exists L]{\Delta'', \mathbf{z}^{\beta}:\exists x. P(x) \vdash \mathbf{w}^\alpha:B}{\deduce{\Delta'', \mathbf{z}^{\beta}: P(x) \vdash \mathbf{w}^\alpha:B}{\Pi'} } & \mathcal{C}_3 }&
     \xRightarrow{\mathsf{PRd}} &
     \hspace*{14em} \\
     \multicolumn{3}{r}{
     \infer[\mathit{nCut}]{\Delta \vdash \mathbf{v}:C}{\mathcal{C}_1 & \Delta'\vdash \mathbf{z}^{\beta}:P(t) &\mathcal{C}_2 &\deduce{\Delta'', \mathbf{z}^{\beta}: P(t) \vdash \mathbf{w}^\alpha:B}{\Pi'[t/x]}&  \mathcal{C}_3 }}\\
    \infer[\mathit{nCut}]{\Delta \vdash \mathbf{v}:C}{\mathcal{C}_1 & \infer[=R]{ \cdot \vdash \mathbf{z}^{\beta}:s=s}{}&\mathcal{C}_2 &\infer[=L]{\Delta'', \mathbf{z}^{\beta}:s=s \vdash \mathbf{w}^\alpha:B }{\Delta'' \vdash \mathbf{w}^\alpha:B} & \mathcal{C}_3 }& \xRightarrow{\mathsf{PRd}} &
     \infer[\mathit{nCut}]{\Delta \vdash \mathbf{v}:C}{\mathcal{C}_1 & \mathcal{C}_2 & \Delta'' \vdash \mathbf{w}^\alpha:B  & \mathcal{C}_3 }\\
\end{array}
}\]    
    {
   \[ \hspace{-2cm}
\begin{array}{ccc}
     \infer[\mathit{nCut}]{\Delta \vdash \mathbf{v}:C}{\mathcal{C}_1& \infer[\otimes R]{ \Delta' \vdash \mathbf{z}^{\beta}:A_1 \otimes A_2}{\Delta'_1\vdash \mathbf{u}^{\eta}:A_1 & \Delta'_2\vdash \mathbf{z}^{\beta}:A_2 }& \mathcal{C}_2 & \infer[\otimes L]{\Delta'', z^{\beta}:A_1 \otimes A_2 \vdash  \mathbf{w}^\alpha:B}{\Delta'', \mathbf{u}^{\eta}:A_1, \mathbf{z}^{\beta}: A_2 \vdash \mathbf{w}^\alpha:B } & \mathcal{C}_3 }  &   \xRightarrow{\mathsf{PRd}}&
     \hspace*{6em} \\
     \multicolumn{3}{r}{
     \infer[\mathit{nCut}]{\Delta \vdash \mathbf{v}:C}{\mathcal{C}_1& \Delta'_1\vdash \mathbf{u}^{\eta}:A_1 & \Delta'_2\vdash \mathbf{z}^{\beta}:A_2 & \mathcal{C}_2 & \Delta'', \mathbf{u}^{\eta}:A_1, \mathbf{z}^{\beta}: A_2 \vdash  \mathbf{w}^\alpha:B  & \mathcal{C}_3 }}\\
     \end{array}\]}%
         {
   \[\hspace{-2cm}
\begin{array}{ccc}
     \infer[\mathit{nCut}]{\Delta \vdash \mathbf{v}:C}{\mathcal{C}_1& \infer[\multimap R]{ \Delta' \vdash \mathbf{z}^{\beta}:A_1 \multimap A_2}{\Delta', \mathbf{u}^\eta: A_1 \vdash \mathbf{z}^{\beta}:A_2}& \mathcal{C}_2 & \infer[\multimap L]{\Delta'', z^{\beta}:A_1 \multimap A_2 \vdash  \mathbf{w}^\alpha:B}{\Delta_1''\vdash  \mathbf{u}^{\eta}:A_1 & \Delta_2'', \mathbf{z}^{\beta}: A_2 \vdash \mathbf{w}^\alpha:B } & \mathcal{C}_3 }  &   \xRightarrow{\mathsf{PRd}}&
     \hspace*{4em} \\
     \multicolumn{3}{r}{
     \infer[\mathit{nCut}]{\Delta \vdash \mathbf{v}:C}{\mathcal{C}_1& \mathcal{C}_2 & \Delta_1''\vdash  \mathbf{u}^{\eta}:A_1 & \Delta', \mathbf{u}^\eta: A_1 \vdash \mathbf{z}^{\beta}:A_2 & \Delta_2'', \mathbf{z}^{\beta}: A_2 \vdash \mathbf{w}^\alpha:B   & \mathcal{C}_3 }}\\
     \end{array}\]}%
         \[\hspace{-2cm}{\small
\begin{array}{lcl}
    \infer[\mathit{nCut}]{\Delta \vdash \mathbf{v}:C}{\mathcal{C}_1 & \infer[\mu R]{ \Delta' \vdash \mathbf{z}^{\beta}:T(\overline{t})}{ \Delta'\vdash \mathbf{z}^{\beta+1}:[\overline{t}/\overline{x}]A & T(\overline{x})=_{\mu}A }& \mathcal{C}_2 &\infer[\mu L]{\Delta'', \mathbf{z}^{\beta}:T(\overline{t}) \vdash \mathbf{w}^\alpha:B}{\Delta'', \mathbf{z}^{\beta+1}: [\overline{t}/\overline{x}]A  \vdash \mathbf{w}^\alpha:B & T(\overline{x})=_{\mu}A} & \mathcal{C}_3 }&
     \xRightarrow{\mathsf{PRd}}&
     \hspace*{4em} \\
     \multicolumn{3}{r}{
     \infer[\mathit{nCut}]{\Delta \vdash \mathbf{v}:C}{\mathcal{C}_1 & \Delta'\vdash \mathbf{z}^{\beta+1}:[\overline{t}/\overline{x}]A  &\mathcal{C}_2 &\Delta'', \mathbf{z}^{\beta+1}: [\overline{t}/\overline{x}]A  \vdash \mathbf{w}^\alpha:B&  \mathcal{C}_3 }}
\end{array}
}\]
\[\hspace{-2cm}{
\begin{array}{lcl}
     \infer[\mathit{nCut}]{\Delta_1,\Delta_2 \vdash \mathbf{z}^\beta:A_1 \otimes A_2}{\mathcal{C}& \infer[\otimes R]{ \Delta_1', \Delta'_2 \vdash \mathbf{z}^{\beta}:A_1 \otimes A_2}{\Delta'_1\vdash \mathbf{u}^{\eta}:A_1 & \Delta'_2\vdash \mathbf{z}^{\beta}:A_2}  }&  \xRightarrow{\mathsf{RFLIP}} & \infer[\otimes R]{\Delta_1, \Delta_2 \vdash \mathbf{z}^\beta:A_1 \otimes A_2}{\infer[nCut]{\Delta_1\vdash \mathbf{u}^{\eta}:A_1}{ \mathcal{C}_{\Delta'_1} &  \Delta'_1\vdash \mathbf{u}^{\eta}:A_1} & \infer[\mathit{nCut}]{\Delta_2\vdash \mathbf{z}^{\beta}:A_2}{\mathcal{C}_{\Delta'_2} & \Delta'_2\vdash \mathbf{z}^{\beta}:A_2 } }\\
          \end{array}}\]
         \[\hspace{-3cm}\mbox{
$\mathcal{C}_{\Delta_1'}$ in the above reduction is a subset of the tape $\mathcal{C}$ connected to $\Delta'_1$.} \mbox{By definition of tape, two sets  $\mathcal{C}_{\Delta_1'}$ and  $\mathcal{C}_{\Delta_2'}$ partition $\mathcal{C}$.
}\]
     \[\hspace{-2cm} {\begin{array}{lcl}
     \infer[\mathit{nCut}]{\Delta, \mathbf{z}^\beta: A_1 \otimes A_2  \vdash \mathbf{v}:C}{\mathcal{C}_1 & \infer[\otimes L]{ \Delta',\mathbf{z}^\beta: A_1 \otimes A_2 \vdash \mathbf{w}^{\alpha}:B}{\Delta', \mathbf{u}^\eta: A_1, \mathbf{z}^\beta: A_2 \vdash \mathbf{w}^{\alpha}:B} & \mathcal{C}_2 }&  \xRightarrow{\mathsf{LFLIP}} & \infer[\otimes L]{\Delta, \mathbf{z}^\beta: A_1 \otimes A_2  \vdash \mathbf{v}:C}{ \infer[\mathit{nCut}]{\Delta, \mathbf{u}^\eta: A_1 , \mathbf{z}^\beta:A_2 \vdash \mathbf{v}:C}{ \mathcal{C}_1 &  \Delta', \mathbf{u}^\eta: A_1, \mathbf{z}^\beta: A_2 \vdash \mathbf{w}^{\alpha}:B & \mathcal{C}_2}}\\\\
     \end{array}
}\]
\[\hspace{-2cm}{
\begin{array}{lcl}
     \infer[\mathit{nCut}]{\Delta, \mathbf{z}^\beta:s=t \vdash  \mathbf{w}^\alpha:B}{ \mathcal{C}_1 & \infer[=L]{\Delta',\mathbf{z}^\beta:s=t \vdash \mathbf{w}^\alpha:B}{\Delta'[\theta] \vdash \mathbf{w}^\alpha:B'[\theta] & \forall \theta \in \mathtt{mgu}(t,s)}& \mathcal{C}_2}  & \xRightarrow{\mathsf{LFLIP}} & \hspace*{20em} \\
     \multicolumn{3}{r}{
     \infer[=L]{\Delta, \mathbf{z}^\beta:s=t \vdash \mathbf{w}^\alpha:B}{ \infer[\mathit{nCut}]{\Delta[\theta] \vdash \mathbf{w}^\alpha:B[\theta] }{\mathcal{C}_1[\theta] &  \Delta'[\theta] \vdash \mathbf{w}^\alpha:B'[\theta] & \mathcal{C}_2[\theta]  } & \forall \theta \in \mathtt{mgu}(t,s) }}
     \end{array}
}\]
\vspace{-10pt}%
\[\hspace{-2.2cm}{
\begin{array}{lcl}
     \infer[\mathit{nCut}]{\Delta \vdash \mathbf{z}^\beta:C }{\mathcal{C}_1 & \infer[ID]{\mathbf{x}^\alpha:A \vdash \mathbf{w}^\gamma:A }{} & \mathcal{C}_2 }  & \xRightarrow{\mathsf{ID-Elim}} & \infer[\mathit{nCut}]{\Delta \vdash \mathbf{z}^\beta:C}{\mathcal{C}_1 & \mathcal{C}_2[\mathbf{x}^\alpha/\mathbf{w}^\gamma] }\\
      \infer[\mathit{nCut}]{\Delta\vdash \mathbf{v}:C}{  {\mathcal{C}_1} &  \infer[\mathsf{Cut}]{ \Delta',\Delta''\vdash \mathbf{w}^\alpha:B }{  \Delta' \vdash \mathbf{z}^\beta: A & \Delta'', \mathbf{z}^\beta: A  \vdash \mathbf{w}^\alpha:B } & {\mathcal{C}_2} }  & \xRightarrow{\mathsf{Merge}} &
   \infer[\mathit{nCut}]{ \Delta \vdash \mathbf{v}:C}{    {\mathcal{C}_1} &   \Delta' \vdash \mathbf{z}^\beta: A & \Delta'', \mathbf{z}^\beta: A \vdash  \mathbf{w}^\alpha:B   & {\mathcal{C}_2}}\\\\
      \end{array} }\]
    \caption{Primitive operations.}
    \label{fig:cred}
\end{figure*}}

\begin{algorithm}[]
\SetAlgoLined
 Initialization: $\Lambda \leftarrow \emptyset; Q \leftarrow [(\epsilon, [v])]$; $v$ is the root sequent. $\rho(s)$ is the rule applied on formula annotated with  position variable $s$, it can either be an $\mathsf{ID}$, $\mathsf{Cut}$, a $L$ rule, or a $R$ rule. $\mathsf{lft}(M)$ and $\mathsf{rgt}(M)$ are defined in Definition \ref{def:tape}. The $\mathsf{Flip}$ rules will return a rule that they permuted down, the sequent corresponding to that, and a list $\mathit{List}$ of one or two tapes.  The output of the algorithm is a tree labelled by $\{0,1\}$. For each node $w \in \{0,1\}^*$ of the tree it also identifies the corresponding sequent, $\mathit{s}(w)$, and the rule applied on the node, $\mathit{r}(w)$. \\ 
 \While{$Q \neq \emptyset$}{
 $(w,M)\leftarrow \mathit{pull}(Q)$\;
 $\Lambda \leftarrow \Lambda \cup \{w\}$\;
 $M \leftarrow \mathit{Treat}(M)$\;
  \eIf{$|M| = 1\, \mathsf{and}\, \rho(\mathsf{lft}(M))=\mathsf{ID}$ }   {
  $(\mathit{r}(w), \mathit{s}(w),\mathit{List})\leftarrow \mathsf{IdOut}(M)$\;
   }{
   \eIf{$\exists s\in \mathsf{lft}(M). \rho(s)\in L$}{
   $(\mathit{r}(w), \mathit{s}(w),\mathit{List}) \leftarrow \mathsf{LFlip}(M)$\;
   }{ \If{$\exists s \in \mathsf{rgt}(M). \rho(s)\in R$}{
   $(\mathit{r}(w), \mathit{s}(w),\mathit{List}) \leftarrow \mathsf{RFlip}(M)$\;
   }}
   \eIf{$\mathit{List}= [M']$}{
   $\mathit{push}((w0,M'),Q)$\;}{
   \If{$\mathit{List}= [M'0,M'1]$}{$\mathit{push}((w0,M'0),Q)$\;
   $\mathit{push}((w1,M'1),Q$\;
   }}
  }
 }
 \caption{Cut elimination algorithm} \label{algorithm}
\end{algorithm}

  {\small
 \begin{algorithm}[]
\SetAlgoLined
Initialization: $M$ is a branching tape. $i$ and $j$ in $\mathsf{PRd}(M,i,j)$ are the index of the two sequents in tape on which the reduction rules are applied. Similarly $i$ in $\mathsf{Merge}(M,i)$ and $\mathsf{idElim}$ is the  index of the sequent in the tape on which the corresponding rule is applied. $\rho'(i)$ is the rule applied on the $i$-th sequent of the tape, it can either be an $\mathsf{ID}$, $\mathsf{Cut}$, a $L$ rule, or a $R$ rule.\\
 \While{$\rho(\mathsf{lft}(M))\not  \in L\, \mathsf{and}\, \rho(rgt(M))\not  \in R  $}{
 \eIf{$|M|>1 \, \mathsf{and}\, \exists i\in M:\rho'(i) =\mathsf{ID} $ }   {
  $M\leftarrow \mathsf{IdElim}(M,i)$\;
   }{
   \eIf{$\exists i\in M:\rho'(i) =\mathsf{Cut}$}{
   $M\leftarrow \mathsf{Merge}(M,i)$\;
   }{ \If{$\exists i.\exists j.\exists\, \circ\in\{1,\oplus, \&, \otimes, \multimap\}.\rho'(i)=\circ R$ and $\rho'(i)=\circ L$ } {
   $M\leftarrow \mathsf{PRd}(M,i,j)$\;
   }}
  }
 }
 \caption{Treat Function}\label{alg:treat}
\end{algorithm}
}

Our cut elimination algorithm is given as Algorithms \ref{algorithm} and \ref{alg:treat}. We define a function \emph{Treat} that reduces the sequence in a branching tape with principal reductions ($\mathsf{PRd}$) until either a left rule is applied on one of its leftmost formulas or a right rule is applied on its rightmost formula. While this condition holds, the algorithm applies a \emph{flip} rule on a leftmost/rightmost formula of the tape. The flipping step is always productive since it pushes a cut one step up. It suffices to show that the principal reductions are terminating to prove productivity of the algorithm. We prove termination of the principal reductions in Lemma \ref{thm:main}. 


\begin{lemma}\label{thm:main} For every input tape $M$, computation of $\mathit{Treat}(M)$ halts.

\end{lemma}
\begin{proof}
We show that $\mathit{Treat}(M)$ does not have an infinite computation tree. Assume for the sake of contradiction that $\mathit{Treat}(M)$ has an infinite computation tree and gets into an infinite loop.
Put  $M_i$ for $i \ge 1$ to be the branching tape in memory before the $i$-th turn of the loop, with $M_1=M$. We build the full trace $T$ of the algorithm. $T$ is a tree which its root is $(0,0)$ and the branches are built using the following transition rules. We use essentially the same transition rules as in  Fortier and Santocanale\cite{Fortier13csl}. 
\begin{itemize}
    \item For $1 \le i \le |M_1|$,  $(0,0) \rightarrow^{i} (1,i)$.\\
    \item If $M_{n+1}= \mathsf{ID-Elim}(M_{n},i)$ then \begin{itemize}
        \item $(n,k)\rightarrow^\bot (n+1,k)$ for $k<i$,
        \item $(n,k)\rightarrow^{0} (n+1,k-1)$ for $k>i$.\\
    \end{itemize}
    \item If $M_{n+1}= \mathsf{Merge}(M_{n},i)$ then \begin{itemize}
        \item $(n,k)\rightarrow^\bot (n+1,k)$ for $k<i$,
        \item $(n,i) \rightarrow^1 (n+1, i)$,
        \item $(n,i) \rightarrow^2 (n+1, i+1)$,
        \item $(n,k)\rightarrow^{0} (n+1,k+1)$ for $k>i$.
    \end{itemize}
\end{itemize}
One difference between our algorithm, and Fortier and Santocanale's is that in our algorithm the sequents subject to reduction may not be next to each other. Thus, the $\mathsf{PRd}$ function needs to receive two indices in the tape to find the sequents for reduction. All reductions except those corresponding to $\otimes$ and $\multimap$ are non-branching($\mathit{nb}$) and their transition rules are quite similar to the one introduced by Fortier and Santocanale\cite{Fortier13csl}.
\begin{itemize}
    \item If $M_{n+1}= \mathsf{PRd}_{\mathit{nb}}(M_{n},i,j)$ then \begin{itemize}
        \item $(n,k)\rightarrow^\bot (n+1,k)$ for $k\not\in\{i,j\}$,
        \item $(n,i) \rightarrow^0 (n+1, i)$,
        \item  $(n,j)\rightarrow^{0} (n+1,j)$.
    \end{itemize}
\end{itemize}
The reductions corresponding to $\otimes$ and $\multimap$, however, produce a branch and need to be defined separately:
\begin{itemize}
    \item If $M_{n+1}= \mathsf{PRd}_{\otimes}(M_{n},i,j)$ then \begin{itemize}
        \item $(n,k)\rightarrow^\bot (n+1,k)$ for $k<i$,
        \item $(n,i) \rightarrow^{1_a} (n+1, i)$ and $(n,i) \rightarrow^{1_b} (n+1, i+1)$,
        \item  $(n,j)\rightarrow^{0} (n+1,j+1)$,
        \item $(n,k)\rightarrow^{\bot} (n+1,k+1)$ for $i<k<j$ or $k>j$. 
    \end{itemize}.  
     \item If $M_{n+1}= \mathsf{PRd}_{\multimap}(M_{n},i,j)$ then 
     \begin{itemize}
         \item $(n,k)\rightarrow^\bot (n+1,k)$ for $k<i$,
         \item $(n,i) \rightarrow^0 (n+1, j)$,
          \item $(n,k)\rightarrow^\bot (n+1,k-1)$ for $i<k<j$,
         \item $(n,j) \rightarrow^1 (n+1, j-1)$ and $(n,j) \rightarrow^2 (n+1, j+1)$, 
         \item $(n,k)\rightarrow^{\bot} (n+1,k+1)$ for $k>j$.
     \end{itemize}
\end{itemize}

A node $(n,k)$ in $T$ is read as  the $k$-th element of branching tape $M_n$. Transitions labelled by $\bot$ mean that the sequent has not evolved by a reduction rule, while other labels show that the sequent is evolved into one or two (in the case of branching rules) new sequents in the next tape. We get the real trace $\Psi$ by collapsing the transitions labelled by $\bot$.
In addition to the labels $0,1,2$ used by Fortier and Santocanale, we use two distinct labels $1_a$ and $1_b$ in the transitions of $\mathsf{PRd}_{\otimes}$. The order $<$ is defined over the labels as $0<1<1_a<2$ and $0<1<1_b<2$. $\Psi$ is an infinite, finitely branching labelled tree with prefix order $\sqsubseteq$ and lexicographical order $<$. A branch in $\Psi$ is a maximal path. The set of all branches of $\Psi$ ordered lexicographically forms a chain complete partially ordered set. 

The labels $1_a, 1_b$ are to distinguish between two types of branching in $\Psi$: (i) the $1/2$ branching that occurs in  $\mathsf{Merge}$ and $\mathsf{PRd}_{\multimap}$ rules, and (ii) the $1_a/1_b$ branching in the $\mathsf{PRd}_{\otimes}$. In the first case the left branch is lexicographically less than the right branch, while in the second case the branches are incomparable. It is important to observe that two branches may share a position variable $x^\alpha$ only if one is less than the other one: the sequents in the branches created by $\mathsf{PRd}_{\otimes}$ are disjoint. 

An infinite branch is a $\mu$- branch (resp. $\nu$-branch) if its corresponding derivation is a $\mu$- trace (resp. $\nu$-trace). By our validity condition $\Psi$ satisfies the property that a $\nu$-branch $\alpha$ can only admit finitely many  branches $\beta>\alpha$ on its right side. This restriction corresponds to the $1/2$ branching by cuts and $\multimap$ reductions.

We prove the following three contradictory statements:
 \begin{enumerate}
     \item[(i)] An infinite branch of $\Psi$ which is not less than any other infinite branches (a maximal infinite branch) is a $\mu$-branch:
     
     Assume that we add $1_a<1_b$ to the ordering, then the set of all branches of $\Psi$ forms a complete lattice, and by Konig's lemma it has a greatest infinite branch $\gamma$. This branch is maximal if we dismiss the relation $1_a<1_b$ from the ordering. The branch $\gamma$ is either a $\mu$- or a $\nu$- branch. Assume it is a $\nu$- branch.  Then either it forms infinitely many branches $\{\beta_i\}_{i\in I}$ on its right such that $\gamma<\beta_i$ or there is an infinite branch $\beta$ greater than it. In both cases we can form a contradiction.
     
    Note that if $\gamma$ is a $\nu$-branch, then there is an infinite chain of inequalities for  position variables $\mathbf{x1}^{\alpha_1}, \mathbf{x2}^{\alpha_2}, \cdots$ on the succedents of $\gamma$: 
{\[
\mathsf{snap}(\mathbf{x1}^{\alpha_1})>_{\Omega_1^\gamma}\mathsf{snap}(\mathbf{x2}^{\alpha_2})>_{\Omega_2^\gamma}\cdots.
\]}%
   These position variables can connect sequents in $\gamma$ to the sequents in one or more branches \{$\beta_i\}_{i \in I}$ only if they also occur as antecedents in $\{\beta_i\}_{i \in I}$, i.e. $\gamma < \{\beta_i\}_{i \in I}$.

     \item[(ii)] Let $\gamma$ be a maximal infinite branch, and form a decreasing chain of $\mu$-branches in $\Psi$ starting from $\gamma$: $\cdots < \beta_2 < \beta_1 < \gamma$.  Put $E$ to be the elements of this chain.  Then $\eta=\bigwedge E$ exists since $\Psi$ is complete chain and it is a $\mu$-branch:
      If $\eta \in E$ then it is trivially true. Otherwise, by the way we constructed $\Psi$, it means that $\eta$ has infinitely many branches $\{\beta_i\}_{i \in I}$ on its right such that $\eta < \beta_i$ and thus cannot be a $\nu$ branch.
     
     \item[(iii)] If $\beta$ is a $\mu$-branch, then there exists another $\mu$-branch $\beta' < \beta$:
     
      $\beta$ is a $\mu$-branch so for infinitely many position variables  $\mathbf{x1}^{\alpha_1}, \mathbf{x2}^{\alpha_2}, \cdots$ on the antecedents of $\beta$ we can form an infinite chain of inequalities 
{\[
\mathsf{snap}(\mathbf{x1}^{\alpha_1})>_{\Omega_1^\beta}\mathsf{snap}(\mathbf{x2}^{\alpha_2})>_{\Omega_2^\beta}\cdots.
\]}%
      There are two possibilities here:
      \begin{enumerate}
          \item[(a)] There is an infinite branch $\beta'<\beta$ with infinitely many position variables $\mathbf{xi}^{\alpha_{i}}, \mathbf{x\{i+1\}}^{\alpha_{i+1}}, \cdots$ as its succedents. Note that these position variables connect sequents in $\beta$ to the sequents in $\beta'$ infinitely many times. So every $\mu/\nu L$ rule in $\beta$ reduces with a $\mu/\nu R$ rule in $\beta'$. This means that a $\mu R$ rule with priority $i$ is applied on the succedent of $\beta'$ infinitely often but no priority $j<i$ has an infinitely many $\nu R$ rule in $\beta'$. %
      \item[(b)] There is an infinite branch $\beta'<\beta$ with infinitely many ($1/2$) branches $\{\delta_i\}_{i \in I}>\beta$ on its right. Note that if we only have infinitely $1a/1b$ right branches corresponding to $\otimes$- reductions, the chain of inequalities $\mathsf{snap}(\mathbf{x1}^{\alpha_1})>_{\Omega_1^\beta}\mathsf{snap}(\mathbf{x2}^{\alpha_2})>_{\Omega_2^\beta}\cdots$ does not hold.
      \end{enumerate}
     In both cases $\beta'$ cannot be a $\nu$-branch and thus is a $\mu$-branch.
 \end{enumerate} 
 
 Items (i)-(iii) form a contradiction. We can form the nonempty collection $E$ of all $\mu$- branches in $\Psi$ that from a maximal decreasing chain starting from $\gamma$ by (i). By (ii) we get $(\eta= \bigwedge E)\in E$ is the minimum of this chain. This forms a contradiction with (iii) and maximality of $E$.
 
 With a similar reasoning, we can prove that the output of the cut elimination algorithm is also a valid derivation. Since the reasoning of the proof is similar to the above, we only provide a high level description here. Consider a branch $b$ in the output derivation of Algorithm \ref{algorithm}. Using a similar set of transition rules in the above proof we can build a trace for the full algorithm (including the treating part). For the flip (external reduction) rules that create two branches we continue with the branch corresponding to $b$ and dismiss the other one. The transition rules for non-branching external reductions are similar to the non-branching principal ones. Consider $\Psi_b$ to be the real trace built based on branch $b$ in the derivation. If $b$ is finite or a $\nu$-trace, we are done. Assume that branch $b$ is not finite and is not a $\nu$-trace.  We know that a branch $s$ containing the succedents of $b$  exists in trace $\Psi_b$ and is a $\mu$-trace. By a similar reasoning to item (iii) above, we get that either there is a $\mu$-trace $t<s$ or the antecedents that make $s$ a $\mu$-trace are the antecedents  of branch $b$ in the output derivation. In the second case, we are done. In the first case, we can form a contradiction by proving a similar statement to item (ii) above.
\end{proof}

\begin{customthm}{1}
A valid (infinite) derivation enjoys the cut elimination property. 
\end{customthm}
\begin{proof}
 We annotate a given derivation in the system of Figure \ref{fig:rules-1} to get a derivation in the system of Figure \ref{fig:rules-2} productively (as described in Section \ref{sec:validitycondition}).  As a corollary to Lemma~\ref{thm:main} the cut elimination algorithm (Algorithm \ref{algorithm}) produces a potentially infinite valid cut free proof for the annotated derivation. By simply ignoring the annotations of the output, we get a cut free proof in the calculus of Figure \ref{fig:rules-1}.
\end{proof}

\subsection{Guard condition for processes}\label{app:guard}

\begin{definition}
An infinite branch of a derivation is a \emph{left $\mu$-trace} if for infinitely many channels $x1^{\alpha_1}, x2^{\alpha_2}, \cdots$ appearing as antecedents of judgments in the branch as
$${\deduce{\vdots}{\infer{{x1}^{\alpha_1}: \mathtt{A}_1 \vdash_{\Omega} \mathtt{Q}_1 :: ({w}^{\beta}:\mathtt{C}_1)}{ \infer{\vdots}{\deduce{{{x2}^{\alpha_2}: \mathtt{A}_2 \vdash_{\Omega_1} \mathtt{Q}_2 :: (y^{\delta}:\mathtt{C}_2)}}{\vdots}}}  }}$$
we can form an infinite chain of inequalities 
{\small $
\mathsf{snap}(x1^{\alpha_1})>_{\Omega_1}\mathsf{snap}(x2^{\alpha_2})>_{\Omega_2}\cdots.
$}%

 Dually, an infinite branch of a derivation is 
a \emph{right $\nu$-trace} if for infinitely many channels $y1^{\beta_1}, y2^{\beta_2}, \cdots$ appearing as the succedents of judgments in the branch as
$${\deduce{\vdots}{\infer{ \bar{x}^\alpha: \omega_1\vdash_{\Omega} \mathtt{Q}_1 :: ({y1}^{\beta_1}:\mathtt{C}_1)}{ \infer{\vdots}{\deduce{{\bar{w}^\delta: \omega_2 \vdash_{\Omega_1} \mathtt{Q}_2:: ({y2}^{\beta_2} :\mathtt{C}_2)}}{\vdots}}}  }}$$
we can form an infinite chain of inequalities 
{\small$
\mathsf{snap}(y1^{\beta_1})>_{\Omega_1}\mathsf{snap}(y2^{\beta_2})>_{\Omega_2}\cdots.
$}

\end{definition}

\subsection{Asynchronous dynamics}\label{app:async}
Here we provide the more detailed definition for predicate  $\mathsf{Cfg}$. To ensure that the messages sent along channel $x$ are received in the correct order, the generation $\alpha$ of $x$ is incremented to $\alpha+1$ with each communication.
\begin{figure}[h]
\newcommand{\semi}{\mathrel{;}}
\small  \centering
 {\begin{equation*}
\hspace{-3cm} 
\begin{array}{llcll}
1. &\mathsf{Cfg}_{x^\alpha:\mathtt{A},x^\alpha:\mathtt{A}} (\mathsf{\cdot}) & =^{\mathbf{n}+2}_{\mu}& 1 & \msc{Emp}\\
2. &\mathsf{Cfg}_{\bar{x}^\alpha:\omega,y^\beta:\mathtt{B}} (\mathcal{C}_1|_{z:\mathtt{C}}\mathcal{C}_2) & =^{\mathbf{n}+2}_{\mu}& \exists z. \exists \eta.  \mathsf{Cfg}_{\bar{x}^\alpha:\omega,z^\eta:\mathtt{C}}(\mathcal{C}_1) \otimes  \mathsf{Cfg}_{z^\eta:\mathtt{C},y^\beta:\mathtt{B}}(\mathcal{C}_2) & \msc{Comp}\\
3.&\mathsf{Cfg}_{x^\alpha:\mathtt{A},y^\beta:\mathtt{A}} (y\leftarrow x) & =^{\mathbf{n}+2}_{\mu}& (x^\alpha=y^\beta)& \msc{Id}\\
4.&\mathsf{Cfg}_{\bar{x}^\alpha:\omega,y^\beta:\mathtt{B}}( (z:\mathtt{C}) \leftarrow \mathsf{Q}_1 \semi \mathsf{Q}_2) & =^{\mathbf{n}+2}_{\mu} &  \exists z. \exists \eta.  \mathsf{Cfg}_{\bar{x}^\alpha:A,z^\eta:\mathtt{C}}(\mathsf{Q}_1) \otimes  \mathsf{Cfg}_{z^\eta:\mathtt{C},y^\beta:\mathtt{B}}(\mathsf{Q}_2) & \msc{Cut}\\
5.&\mathsf{Cfg}_{\cdot,y^\beta:1}(\mathbf{close} Ry)& =^{\mathbf{n}+2}_{\mu}&  \mathsf{Msg}(y^{\beta}.\mbox{closed}) \otimes 1 & 1\\
6.&\mathsf{Cfg}_{x^\alpha:1,y^\beta:\mathtt{A}}(\mathbf{wait} Lx; \mathsf{Q})& =^{\mathbf{n}+2}_{\mu} &   \mathsf{Msg}(x^{\alpha}.\mbox{closed}) \multimap \mathsf{Cfg}_{ \cdot, y^\beta:\mathtt{A}}(\mathsf{Q}) & 1\\
7.&\mathsf{Cfg}_{\bar{x}^\alpha:\omega, y^\beta:\&\{\ell:\mathtt{B}_\ell\}_{\ell \in L}}(\mathbf{case} Ry(\ell \Rightarrow \mathsf{Q}_\ell)_{\ell\in L})& =^{\mathbf{n}+2}_{\mu}& \&\{\ell:  \mathsf{Msg}(y^\beta.\ell) \multimap \mathsf{Cfg}_{\bar{x}^\alpha:\omega, y^{\beta+1}:\mathtt{B}_\ell}(\mathsf{Q}_\ell)\}_{\ell \in L} & \&\\
8.&\mathsf{Cfg}_{x^\alpha:\&\{\ell:\mathtt{A}_\ell\}_{\ell \in L}, y^\beta:\mathtt{B}}({Lx}.k;\mathsf{Q}) & =^{\mathbf{n}+2}_{\mu} &  \mathsf{Msg}(x^\alpha.k) \otimes \mathsf{Cfg}_{x^{\alpha+1}:\mathtt{A}_k, y^\beta:\mathtt{B}}(\mathsf{Q})& \& \\
9.&\mathsf{Cfg}_{ \bar{x}^\alpha:\omega, y^{\beta}: \oplus \{\ell:\mathtt{B}_\ell\}_{\ell \in L}}({Rw}.k;\mathsf{Q}) & =^{\mathbf{n}+2}_{\mu} &  \mathsf{Msg}(y^\beta.k) \otimes \mathsf{Cfg}_{ \bar{x}^\alpha:\omega, y^{\beta+1}:\mathtt{B}_k}(\mathsf{Q})& \oplus\\
10.&\mathsf{Cfg}_{x^\alpha:\oplus\{\ell:\mathtt{A}_\ell\}_{\ell \in L}, y^\beta:\mathtt{B}}(\mathbf{case} Lx(\ell \Rightarrow \mathsf{Q}_\ell)_{\ell\in L})& =^{\mathbf{n}+2}_{\mu}& \&\{\ell:  \mathsf{Msg}(x^\alpha.\ell) \multimap \mathsf{Cfg}_{x^{\alpha+1}:\mathtt{A}_\ell, y^\beta:\mathtt{B}}(\mathsf{Q}_\ell)\}_{\ell\in L} & \oplus \\
11.&\mathsf{Cfg}_{\bar{x}^\alpha:\omega, y^\beta: \mathtt{t}}(\mathbf{case} Ry(\nu_\mathtt{t} \Rightarrow \mathsf{Q}))& =^{\mathbf{n}+2}_{\mu}&   \mathsf{Msg}(y^\beta.\nu_t) \multimap \mathsf{Cfg}_{\bar{x}^\alpha:\omega, y^{\beta+1}:\mathtt{C}}(\mathsf{Q})& \mathtt{t}=^i_{\nu} \mathtt{C} \in \mathbf{\Sigma}\\ 
12.&\mathsf{Cfg}_{ x^\alpha: \mathtt{t}, y^\beta:\mathtt{B}}({Lx}.\nu_\mathtt{t};\mathsf{Q}) & =^{\mathbf{n}+2}_{\mu} &  \mathsf{Msg}(x^\alpha.\nu_t) \otimes \mathsf{Cfg}_{x^{\alpha+1}:\mathtt{C}, y^\beta:B}(\mathsf{Q}) & \mathtt{t}=^i_{\nu} \mathtt{C} \in \mathbf{\Sigma}\\
13.&\mathsf{Cfg}_{ \bar{x}^\alpha:\omega, y^\beta:\mathtt{t}}({Ry}.\mu_\mathtt{t};\mathsf{Q}) & =^{\mathbf{n}+2}_{\mu} &  \mathsf{Msg}(y^\beta.\mu_t) \otimes \mathsf{Cfg}_{\bar{x}^{\alpha}:\omega, y^{\beta+1}:\mathtt{C}}(\mathsf{Q}) & \mathtt{t}=^i_{\mu} \mathtt{C} \in \mathbf{\Sigma}\\
14.&\mathsf{Cfg}_{x^\alpha:\mathtt{t}, y^\beta:\mathtt{B}}(\mathbf{case} Lx(\mu_\mathtt{t} \Rightarrow \mathsf{Q}))& =^{\mathbf{n}+2}_{\mu}&    \mathsf{Msg}(x^\alpha.\mu_t) \multimap \mathsf{Cfg}_{x^{\alpha+1}:\mathtt{C},y^\beta:\mathtt{B}}(\mathsf{Q})& \mathtt{t}=^i_{\mu} \mathtt{C} \in \mathbf{\Sigma}\\
15.&\mathsf{Cfg}_{\bar{x}^\alpha:\omega,y^\beta:\mathtt{B}}(\mathsf{Y}) & =^{n+1}_{\nu}& \mathsf{Cfg}_{\bar{x}^\alpha:\omega, y^\beta:\mathtt{B}}(\mathsf{Q} [y/w, \overline{x}/\bar{u}] ) & \bar{u}:\omega \vdash \mathsf{Y}=\mathsf{\mathsf{Q}}:: (w:\mathtt{B}) \in V \end{array}
\end{equation*}
}        \caption{Definition of predicate $\mathsf{Cfg}$.}
    \label{fig:comp-2}
 \end{figure}

Here, we show a part of the equivalent formula corresponding to the first two lines of Figure \ref{fig:comp-2},{ where $c_i$ and $d_i$ stand for variables in {\small $\mathit{FIMALL}^{\infty}_{\mu,\nu}$}}:
{\small\[
\begin{array}{cl}
   \mathsf{Cfg}_{\bar{x}^\alpha:d_1,y^\beta:d_2}(d_3)=& (\exists c_1,c_2,z,c.\, (d_3=c_1|_{z:c}c_2)\,\otimes\, \exists z.\exists \eta. \mathsf{Cfg}_{\bar{x}^\alpha:d_1,z^\eta:c}(c_1) \otimes  \mathsf{Cfg}_{z^\eta:c,y^\beta:d_2}(c_2))  \\
     &  \oplus ((d_3= \cdot) \otimes 1) \oplus \cdots
\end{array}\]}
We use an abbreviation for the last case (row 15) where the predicate is a greatest fixed point. We can unfold its definition using finitely many intermediate predicates $\mathsf{Call}_Y$ (for each $\mathsf{Y} \in V$) as: {\small\[\begin{array}{lcl}\mathsf{Cfg}_{\bar{x}^\alpha:\omega,y^\beta:\mathtt{B}}(\mathsf{Y})& =_\mu^{\mathbf{n}+2} &
\mathsf{Call}_Y(\bar{x}^\alpha:\omega,y^\beta:\mathtt{B})\\
\mathsf{Call}_Y(\bar{x}^\alpha:\omega,y^\beta:\mathtt{B})& =^{n+1}_{\nu}& \mathsf{Cfg}_{\bar{x}^\alpha:\omega, y^\beta:\mathtt{B}}(\mathsf{Q} [y/w, \overline{x}/\bar{u}] ) \quad \bar{u}:\omega \vdash \mathsf{Y}=\mathsf{\mathsf{Q}}:: (w:\mathtt{B}) \in V. 
\end{array}
\]}

\subsection{Strong progress in session types}\label{app:derivation}

\subsubsection{Converting definition of Figure~\ref{fig:[]} to a formula in $\mathit{FIMALL}^{\infty}_{\mu,\nu}$}

For $\mathtt{t}=^i_{\nu}\mathtt{A}\in \mathbf{\Sigma}$, the definition $[y^\beta:\mathtt{t}]$ is an abbreviation. We can unfold the abbreviation, using finitely many intermediate predicates $\mathsf{Unfold}_\mathtt{s}$ (for each $s=_{\nu}\mathtt{A} \in \mathbf{\Sigma}$):
{\small\[\begin{array}{lcll}
 {[y^\beta: \mathtt{t}]} & =^{n+2}_{\mu} &  \mathsf{Unfold}_\mathtt{t}(y^{\beta})\\
 \mathsf{Unfold}_\mathtt{t}(y^{\beta}) & =^i_{\nu} & ( \mathsf{Msg}(y^\beta.\nu_t) \multimap[y^{\beta+1}:\mathtt{A}]) & \qquad \mathtt{t}=^i_{\nu}\mathtt{A}\in \mathbf{\Sigma}

\end{array}\]}
It means that we do not have a single closed encoding of session-typed programs and strong progress, but we have a different encoding for every signature $\mathbf{\Sigma}$. 
Having this abbreviation helps us in the proof of the strong progress theorem: we can assign a matching position variable $\mathbf{y}^\beta$ to every predicate $[y^\beta:\mathtt{t}]$ occurring in the derivation. The structure of the derivation or its validity won't be affected by the abbreviation, but it ensures that the generation of position variables steps at the same pace as the generation of channels (See Appendix \ref{app:derivation}).

\subsubsection{Proof of strong progress}
\begin{customlem}{1.}\label{lem:derivation}
For a configuration of processes $\bar{x}^\alpha:\omega \vdash \mathcal{C} ::(y^\beta:\mathtt{B})$, there is a (possibly infinite) derivation for $ \mathcal{C} \in \llbracket \bar{x}^\alpha:\omega \vdash y^\beta:\mathtt{B} \rrbracket$.
\end{customlem}

\begin{proof}
 Here, we build a derivation for $\mathcal{C} \in \llbracket \bar{x}^\alpha:\omega \vdash y^\beta:\mathtt{B} \rrbracket$. Recall that this is an abbreviation (abb.) for
 \[\forall \mathrm{C}.\, (\forall z. \forall \eta. \mathrm{C} \in \llbracket z^\eta:\omega \rrbracket \multimap \mathrm{C}\mid_{\bar{x}:\mathtt{A}}\mathcal{C} \in \llbracket  y^\beta:\mathtt{B} \rrbracket).\]
And $\mathcal{C} \in \llbracket  y^\beta:\mathtt{B} \rrbracket$ is an abbreviation for {\small$\mathsf{Cfg}_{\cdot,y^\beta:B}(\mathcal{C}) \multimap  [y^\beta:\mathtt{B}].$}

\[\infer[\textit{abb.}]{\cdot \vdash \mathcal{C} \in \llbracket \bar{x}^\alpha:\omega  \vdash y^\beta:\mathtt{B} \rrbracket}{\infer[\forall R]{\cdot \vdash\forall \mathrm{C}.\, (\forall z. \forall \eta. \mathrm{C} \in \llbracket z^\eta:\omega \rrbracket \multimap \mathrm{C}\mid_{\bar{x}^\alpha:\omega}\mathcal{C} \in \llbracket  y^\beta:\mathtt{B} \rrbracket)}{\infer[\multimap R]{\cdot \vdash \forall z. \forall \eta. \mathrm{C} \in \llbracket z^\eta:\omega \rrbracket \multimap \mathrm{C}\mid_{\bar{x}^\alpha:\omega}\mathcal{C} \in \llbracket  y^\beta:\mathtt{B} \rrbracket}{\infer[\textit{abb.}]{\forall z. \forall \eta. \mathrm{C} \in \llbracket z^\eta:\omega \rrbracket \vdash \mathrm{C}\mid_{\bar{x}^\alpha:\omega}\mathcal{C} \in \llbracket  y^\beta:\mathtt{B} \rrbracket}{\infer[\multimap R]{\forall z. \forall \eta. \mathrm{C} \in \llbracket z^\eta:\omega \rrbracket \vdash \mathsf{Cfg}_{\cdot,y^\beta:\mathtt{B}}(\mathrm{C}\mid_{\bar{x}^\alpha:A}\mathcal{C}) \multimap  [y^\beta:\mathtt{B}]}{\infer[\mu L]{\forall z. \forall \eta. \mathrm{C} \in \llbracket z^\eta:\omega \rrbracket, \mathsf{Cfg}_{\cdot,y^\beta:\mathtt{B}}(\mathrm{C}\mid_{\bar{x}^\alpha:A}\mathcal{C}) \vdash  [y^\beta:\mathtt{B}]}{\infer[\exists L]{\forall z. \forall \eta. \mathrm{C} \in \llbracket z^\eta:\omega \rrbracket, \exists z. \exists \eta.   \mathsf{Cfg}_{\cdot,z^\eta:\omega}(\mathrm{C}) \otimes  \mathsf{Cfg}_{z^\eta:\omega,y^\beta:\mathtt{B}}(\mathcal{C}) \vdash  [y^\beta:\mathtt{B}]}{\infer[\forall L]{\forall z. \forall \eta. \mathrm{C} \in \llbracket z^\eta:\omega \rrbracket,   \mathsf{Cfg}_{\cdot,z^\eta:\omega}(\mathrm{C}) \otimes  \mathsf{Cfg}_{z^\eta:\omega,y^\beta:\mathtt{B}}(\mathcal{C}) \vdash  [y^\beta:\mathtt{B}]}{{\infer[\otimes L]{\mathrm{C} \in \llbracket z^\eta:\omega \rrbracket,   \mathsf{Cfg}_{\cdot,z^\eta:\omega}(\mathrm{C}) \otimes  \mathsf{Cfg}_{z^\eta:\omega,y^\beta:\mathtt{B}}(\mathcal{C}) \vdash  [y^\beta:\mathtt{B}]}{\infer[\textit{abb.}]{\mathrm{C} \in \llbracket z^\eta:\omega \rrbracket,   \mathsf{Cfg}_{\cdot,z^\eta:\omega}(\mathrm{C}), \mathsf{Cfg}_{z^\eta:\omega,y^\beta:\mathtt{B}}(\mathcal{C}) \vdash  [y^\beta:\mathtt{B}]}{\infer[\multimap L]{\mathsf{Cfg}_{\cdot,z^\eta:\omega}(\mathrm{C}) \multimap  [z^\eta:\omega],  \mathsf{Cfg}_{\cdot,z^\eta:\omega}(\mathrm{C}), \mathsf{Cfg}_{z^\eta:\omega,y^\beta:\mathtt{B}}(\mathcal{C}) \vdash  [y^\beta:\mathtt{B}]}{\infer[\msc{Id}]{\mathsf{Cfg}_{\cdot,z^\eta:\omega}(\mathrm{C}) \vdash \mathsf{Cfg}_{\cdot,z^\eta:\omega}(\mathrm{C})}{}  &\deduce{[z^\eta:\omega],  \mathsf{Cfg}_{z^\eta:\omega,y^\beta:\mathtt{B}}(\mathcal{C}) \vdash  [y^\beta:\mathtt{B}]}{\star}}}}}}}}}}}}}\]

 To complete the derivation,  it is enough to build a circular derivation  $\star\, [x^\alpha:\mathtt{A}],  \mathsf{Cfg}_{x^\alpha:\mathtt{A},y^\beta:\mathtt{B}}(\mathcal{C}) \vdash  [y^\beta:\mathtt{B}]$ for an open configuration $x^\alpha:\mathtt{A} \vdash \mathcal{C}:: y^\beta{:}\mathtt{B}$. if the configuration is closed, it is enough to build a derivation $ \star \;\mathsf{Cfg}_{\cdot,y^\beta:\mathtt{B}}(\mathcal{C}) \vdash  [y^\beta{:}\mathtt{B}]$ for a closed configuration $\cdot \vdash \mathcal{C}:: y^\beta:\mathtt{B}$. For the sake of brevity we write $\star\,  \overline{[\bar{x}^\alpha:\omega]},  \mathsf{Cfg}_{\bar{x}^\alpha:\omega,y^\beta:\mathtt{B}}(\mathcal{C}) \vdash  [y^\beta:\mathtt{B}]$ as a generalization for both open and closed configurations, where $\overline{[\bar{x}^\alpha:\omega]}$ is either empty or $[\bar{x}^\alpha:\omega]$, and it is empty if and only if $\bar{x}:\omega$ is empty. We provide a circular derivation for each possible pattern of $\mathcal{C}$. Here are the circular derivations when $\mathcal{C}$ is an empty configuration and a composition of configurations, respectively:

{\small \[ \hspace*{-1cm}\begin{array}{lr}
\infer[\mu L]{\star \; {[x^\alpha:\mathtt{A}]}, \mathsf{Cfg}_{x^\alpha:\mathtt{A},x^\alpha: \mathtt{A}}(\cdot) \vdash {[x^\beta:\mathtt{A}]}}{\infer[1 L]{{[x^\alpha:\mathtt{A}]}, 1 \vdash {[x^\alpha:\mathtt{A}]}}{\infer[\msc{ID}]{[x^\alpha:\mathtt{A}] \vdash [x^\alpha:\mathtt{A}]}{}}} & \infer[\mu L]{\star \; \overline{[\bar{x}^\alpha:\omega]}, \mathsf{Cfg}_{\bar{x}^\alpha:\omega,y^\beta:\mathtt{B}}(\mathcal{C}_1 \mid_{z:\mathtt{C}} \mathcal{C}_2) \vdash {[y^\beta:\mathtt{B}]}}{\infer[\exists L]{\overline{[\bar{x}^\alpha:\omega]}, \exists z. \exists \zeta. (\mathsf{Cfg}_{\bar{x}^\alpha:\omega, z^\zeta:\mathtt{C}}(\mathcal{C}_1) \otimes \mathsf{Cfg}_{z^\zeta:\mathtt{C}, y^\beta:\mathtt{B}}(\mathcal{C}_2)) \vdash {[y^\beta:\mathtt{B}]}}{\infer[\otimes L]{\overline{[\bar{x}^\alpha:\omega]}, \mathsf{Cfg}_{\bar{x}^\alpha:\omega,z^\zeta:\mathtt{C}}(\mathcal{C}_1)\otimes \mathsf{Cfg}_{z^\zeta:\mathtt{C}, y^\beta:\mathtt{B}}(\mathcal{C}_2)) \vdash{[y^\beta:\mathtt{B}]} }{\infer[\msc{Cut}]{\overline{[\bar{x}^\alpha:\omega]},  \mathsf{Cfg}_{\bar{x}^\alpha:\omega,z^\zeta:\mathtt{C}}(\mathcal{C}_1),\mathsf{Cfg}_{z^\zeta:\mathtt{C}, y^\beta:\mathtt{B}}(\mathcal{C}_2)) \vdash{[y^\beta:\mathtt{B}]} }{\deduce{\overline{[\bar{x}^\alpha:\omega]}, \mathsf{Cfg}_{\bar{x}^\alpha:\omega,z^\zeta:\mathtt{C}}(\mathcal{C}_1) \vdash {[z^\zeta:\mathtt{C}]}}{\star} & \deduce{{[z^\zeta:\mathtt{C}]}, \mathsf{Cfg}_{z^\zeta:\mathtt{C}, y^\beta:\mathtt{B}}(\mathcal{C}_2) \vdash{[y^\beta:\mathtt{B}]}}{\star} }}}}
 \end{array}\]}
  
It should be clear how to unfold the derivation given based on patterns to a circular definition in the system of Figure \ref{fig:rules-2}. Consider the first two lines in expanded definition of predicate $\mathsf{Cfg}$:
{\small\[
\begin{array}{cl}
   \mathsf{Cfg}_{\bar{x}^\alpha:d_1,y^\beta:d_2}(d_3)=& (\exists c_1,c_2,z,c.\, (d_3=c_1|_{z:c}c_2)\,\otimes\, \exists z.\exists \eta. \mathsf{Cfg}_{\bar{x}^\alpha:d_1,z^\eta:c}(c_1) \otimes  \mathsf{Cfg}_{z^\eta:c,y^\beta:d_2}(c_2))  \\
     &  \oplus ((d_3= \cdot) \otimes 1) \oplus \cdots
\end{array}\]}
To prove the judgment $\star [x^\alpha:\mathtt{A}],  \mathsf{Cfg}_{x^\alpha:\mathtt{A},y^\beta:\mathtt{B}}(\mathcal{C}) \vdash  [y^\beta:\mathtt{B}]$ without pattern matching, we first need to unfold the definition of $\mathsf{Cfg}_{x^\alpha:\mathtt{A},y^\beta:\mathtt{B}}$ with a $\mu L$ rule. Next, we apply an $\oplus L$ rule on the resulting formula which is the right side of the above definition. After this step, we are in a quite similar situation to the pattern matching argument: we have to prove several branches each corresponding to a pattern of $\mathcal{C}$. In some cases, we may need to apply extra equality and $\oplus$ rules in the derivation without pattern matching. However, for all cases the fixed point rules applied in a cycle are the same in the derivations built with and without pattern matching.

For the cases in which the configuration consists of a single process we give an annotated derivation with position variables and track the relationship between their generations. Without loss of generality, we annotate predicates of the form $[z^\eta:\mathtt{C}]$ with a matching position variable $\mathbf{z}^\eta$ at the start (the bottom) of each cycle. We leave it to the reader to check that this assumption holds as an invariant at the end (the top) of each cycle in the derivation.

\begin{description}
\item {\bf Case 1.} $(y \leftarrow x)$ \[\infer[\mu L]{\star \; {\mathbf{x^\alpha}:[x^\alpha:\mathtt{A}]}, \mathbf{c}^\eta:\mathsf{Cfg}_{x^\alpha:\mathtt{A},y^\beta:\mathtt{A}}(y \leftarrow x) \vdash_{\Lambda} \mathbf{y}^\beta:{[y^\beta:\mathtt{A}]}}{\infer[=L]{\mathbf{x}^\alpha:{[x^\alpha:\mathtt{A}]},\mathbf{c}^{\eta+1}:(x^\alpha=y^\beta) \vdash_{\Lambda_1} \mathbf{y}^\beta:{[y^\beta:\mathtt{A}]}}{\infer[\msc{ID}]{\mathbf{x}^\alpha: {[y^\beta:\mathtt{A}]}\vdash_{\Lambda_1} \mathbf{y}^\beta: {[y^\beta:\mathtt{A}]}}{}}}\]
 $ \Lambda_1=\Lambda \cup\{\mathbf{c}^{\eta+1}_{n+2}<\mathbf{c}^{\eta}_{n+2}\} \cup \{\mathbf{c}^{\eta+1}_{i}=\mathbf{c}^{\eta}_{i} \mid i \neq n+2\}.$\\
 

 \item {\bf Case 2.} $(\mathtt{close}\,Ry)$
 
 \[\infer[\mu L]{\star \; \mathbf{c}^\eta: \mathsf{Cfg}_{\cdot, y^\beta:1}(\mathtt{close}\, Ry) \vdash_{\Lambda} \mathbf{y}^\beta: {[y^\beta:1]}}{\infer[\mu R]{\mathbf{c}^{\eta+1}:\mathsf{Msg}(y^\beta.\mbox{closed})  \otimes 1 \vdash_{\Lambda_1} \mathbf{y}^\beta: {[y^\beta:1]}}{\infer[\otimes L]{\mathbf{c}^{\eta+1}:\mathsf{Msg}(y^\beta.\mbox{closed})  \otimes 1 \vdash_{\Lambda_2} \mathbf{y}^{\beta+1}: \mathsf{Msg}(y^\beta.\mbox{closed}) \otimes [\cdot:\cdot]}{{\infer[\otimes R]{\mathbf{w}^\delta:\mathsf{Msg}(y^\beta.\mbox{closed}) , \mathbf{c}^{\eta+1}:1 \vdash_{\Lambda_3} \mathbf{y}^{\beta+1}: \mathsf{Msg}(y^\beta.\mbox{closed})  \otimes [\cdot:\cdot]}{\infer[\msc{ID}]{\mathbf{w}^\delta:\mathsf{Msg}(y^\beta.\mbox{closed})  \vdash_{\Lambda_4} \mathbf{z}^\gamma:\mathsf{Msg}(y^\beta.\mbox{closed}) }{} & \infer[\mu R]{\mathbf{c}^{\eta+1}:1 \vdash_{\Lambda_3} \mathbf{y}^{\beta+1}:[\cdot:\cdot]}{\infer[1L]{\mathbf{c}^{\eta+1}:1 \vdash_{\Lambda_5} \mathbf{y}^{\beta+2}:1}{
\infer[1R]{\cdot \vdash_{\Lambda_5} \mathbf{y}^{\beta+2}:1}{}}}}}}}}\]
$\Lambda_1=\Lambda \cup\{\mathbf{c}^{\eta+1}_{n+2}<\mathbf{c}^{\eta}_{n+2}\} \cup \{\mathbf{c}^{\eta+1}_{i}=\mathbf{c}^{\eta}_{i} \mid i \neq n+2\}, \quad $
$ \Lambda_2=\Lambda_1 \cup\{\mathbf{y}^{\beta+1}_{n+2}<\mathbf{y}^{\beta}_{n+2}\} \cup \{\mathbf{y}^{\beta+1}_{i}=\mathbf{y}^{\beta}_{i} \mid i \neq n+2\}\, \quad$
$\Lambda_3=\Lambda_2 \cup\{\mathbf{c}^{\eta+1}=\mathbf{w}^\delta\}, $ $\Lambda_4=\Lambda_3 \cup\{\mathbf{y}^{\beta+1}=\mathbf{z}^\gamma\}, \Lambda_5=\Lambda_3 \cup\{\mathbf{y}^{\beta+2}_{n+2}<\mathbf{y}^{\beta+1}_{n+2}\} \cup \{\mathbf{y}^{\beta+2}_{i}=\mathbf{y}^{\beta+1}_{i} \mid i \neq n+2\}\, \quad$\\

\item {\bf{Case 3.}} $(\mathtt{wait}\, Lx;\mathsf{Q})$

{\small \begin{equation}
    \hspace*{-2cm}\infer[\mu L]{\star \; \mathbf{x}^\alpha:{[x^\alpha:1]}, \mathbf{c}^\eta:\mathsf{Cfg}_{x^\alpha:1, y^\beta:\mathtt{B}}(\mathtt{wait}\, Lx;\mathsf{Q}) \vdash_{\Lambda} \mathbf{y}^\beta:{[y^\beta:\mathtt{B}]}}{\infer[\mu L]{\mathbf{x}^\alpha:{[x^\alpha:1]}, \mathbf{c}^{\eta+1}: \mathsf{Msg}(x^\alpha.\mbox{closed}) \multimap \mathsf{Cfg}_{\cdot, y^\beta:\mathtt{B}}(\mathsf{Q}) \vdash_{\Lambda_1} \mathbf{y}^\beta:{[y^\beta:\mathtt{B}]}}{\infer[\otimes L]{\mathbf{x}^{\alpha+1}:\mathsf{Msg}(x^\alpha.\mbox{closed}) \otimes [\cdot:\cdot], \mathbf{c}^{\eta+1}: \mathsf{Msg}(x^\alpha.\mbox{closed}) \multimap \mathsf{Cfg}_{\cdot, y^\beta:\mathtt{B}}(\mathsf{Q}) \vdash_{\Lambda_2} \mathbf{y}^\beta:{[y^\beta:\mathtt{B}]}}{\infer[\multimap L]{\mathtt{w}^\omega:\mathsf{Msg}(x^\alpha.\mbox{closed}), \mathbf{x}^{\alpha+1}: [\cdot:\cdot], \mathbf{c}^{\eta+1}: \mathsf{Msg}(x^\alpha.\mbox{closed}) \multimap \mathsf{Cfg}_{\cdot, y^\beta:\mathtt{B}}(\mathsf{Q}) \vdash_{\Lambda_3} \mathbf{y}^\beta:{[y^\beta:\mathtt{B}]}}{\infer[\msc{ID}]{\mathtt{w}^\omega:\mathsf{Msg}(x^\alpha.\mbox{closed}) \vdash_{\Lambda_3} \mathbf{z}^\kappa:\mathsf{Msg}(x^\alpha.\mbox{closed})}{} & \infer[\mu L]{\mathbf{x}^{\alpha+1}: [\cdot:\cdot], \mathbf{c}^{\eta+1}: \mathsf{Cfg}_{\cdot, y^\beta:\mathtt{B}}(\mathsf{Q}) \vdash_{\Lambda_3} \mathbf{y}^\beta:{[y^\beta:\mathtt{B}]}}{\infer[1 L]{\mathbf{x}^{\alpha+2}: 1, \mathbf{c}^{\eta+1}: \mathsf{Cfg}_{\cdot, y^\beta:\mathtt{B}}(\mathsf{Q}) \vdash_{\Lambda_4} \mathbf{y}^\beta:{[y^\beta:\mathtt{B}]}}{\deduce{\mathbf{c}^{\eta+1}: \mathsf{Cfg}_{\cdot, y^\beta:\mathtt{B}}(\mathsf{Q}) \vdash_{\Lambda_4} \mathbf{y}^\beta:{[y^\beta:\mathtt{B}]}}{\star}}}}}}}\notag\end{equation}}

{\small
$\Lambda_1=\Lambda \cup\{\mathbf{c}^{\eta+1}_{n+2}<\mathbf{c}^{\eta}_{n+2}\} \cup \{\mathbf{c}^{\eta+1}_{i}=\mathbf{c}^{\eta}_{i} \mid i \neq n+2\},$
$ \Lambda_2=\Lambda_1 \cup\{\mathbf{x}^{\alpha+1}_{n+2}<\mathbf{x}^{\alpha}_{n+2}\} \cup \{\mathbf{x}^{\alpha+1}_{i}=\mathbf{x}^{\alpha}_{i} \mid i \neq n+2\}\,,$
$\Lambda_3=\Lambda_2 \cup\{\mathbf{x}^{\alpha+1}=\mathbf{w}^\omega\}, $ $ \Lambda_4=\Lambda_3 \cup\{\mathbf{x}^{\alpha+2}_{n+2}<\mathbf{x}^{\alpha+1}_{n+2}\} \cup \{\mathbf{x}^{\alpha+2}_{i}=\mathbf{x}^{\alpha+1}_{i} \mid i \neq n+2\}\, \quad$}\\

\item {\bf Case 4.} $((z:\mathtt{C}) \leftarrow \mathsf{Q}_1;\mathsf{Q}_2)$

{\small
\begin{equation}
    \hspace*{-2cm}\infer[\mu L]{\star \; \overline{\mathbf{x}^\alpha:{[\bar{x}^\alpha:\omega]}}, \mathbf{c}^\eta:\mathsf{Cfg}_{\bar{x}^\alpha:\omega,y^\beta:\mathtt{B}}((z:\mathtt{C}) \leftarrow \mathsf{Q}_1;\mathsf{Q}_2) \vdash_{\Lambda} \mathbf{y}^\beta:{[y^\beta:\mathtt{B}]}}{\infer[\exists L]{\overline{\mathbf{x}^\alpha:{[\bar{x}^\alpha:\omega]}}, \mathbf{c}^{\eta+1}:\exists z. \exists \zeta. (\mathsf{Cfg}_{\bar{x}^\alpha:\omega, z^\zeta:\mathtt{C}}(\mathsf{Q}_1) \otimes \mathsf{Cfg}_{z^\zeta:\mathtt{C}, y^\beta:\mathtt{B}}(\mathsf{Q}_2)) \vdash_{\Lambda_1} \mathbf{y}^\beta:{[y^\beta:\mathtt{B}]}}{\infer[\otimes L]{\overline{\mathbf{x}^\alpha:{[\bar{x}^\alpha:\omega]}}, \mathbf{c}^{\eta+1}:\mathsf{Cfg}_{\bar{x}^\alpha:\omega,z^\zeta:\mathtt{C}}(\mathsf{Q}_1)\otimes \mathsf{Cfg}_{z^\zeta:\mathtt{C}, y^\beta:\mathtt{B}}(\mathsf{Q}_2)) \vdash_{\Lambda_1} \mathbf{y}^\beta:{[y^\beta:\mathtt{B}]} }{\infer[\msc{Cut}]{\overline{\mathbf{x}^\alpha:{[\bar{x}^\alpha:\omega]}},  \mathbf{w}^{\omega}:\mathsf{Cfg}_{\bar{x}^\alpha:\omega,z^\zeta:\mathtt{C}}(\mathsf{Q}_1),\mathbf{c}^{\eta+1}:\mathsf{Cfg}_{z^\zeta:\mathtt{C}, y^\beta:\mathtt{B}}(\mathsf{Q}_2)) \vdash_{\Lambda_2} \mathbf{y}^\beta:{[y^\beta:\mathtt{B}]} }{\deduce{\overline{\mathbf{x}^\alpha:{[\bar{x}^\alpha:\omega]}}, \mathbf{w}^{\omega}: \mathsf{Cfg}_{\bar{x}^\alpha:\omega,z^\zeta:\mathtt{C}}(\mathsf{Q}_1) \vdash_{\Lambda_2} \mathbf{z}^\zeta:{[z^\zeta:\mathtt{C}]}}{\star} & \deduce{\mathbf{z}^\zeta:{[z^\zeta:\mathtt{C}]}, \mathbf{c}^{\eta+1}:\mathsf{Cfg}_{z^\zeta:\mathtt{C}, y^\beta:\mathtt{B}}(\mathsf{Q}_2) \vdash_{\Lambda_2} \mathbf{y}^\beta:{[y^\beta:\mathtt{B}]}}{\star} }}}}\notag\end{equation}}

{\small$ \Lambda_1=\Lambda \cup\{\mathbf{c}^{\eta+1}_{n+2}<\mathbf{c}^{\eta}_{n+2}\} \cup \{\mathbf{c}^{\eta+1}_{i}=\mathbf{c}^{\eta}_{i} \mid i \neq n+2\}\, \mathit{and}$
$\Lambda_2=\Lambda_1 \cup \{\mathbf{w}^\omega=\mathbf{c}^{\eta+1}\}.$}\\

\item {\bf Case 5.}$(Lx.k;\mathsf{Q})$
{\small
\begin{equation}
    \hspace*{-2cm}\infer[\mu L]{\star \; \mathbf{x}^\alpha:{[x^\alpha:\&\{\ell:\mathtt{A}_\ell\}_{\ell \in L}]}, \mathbf{c}^\eta:\mathsf{Cfg}_{x^\alpha:\&\{\ell:\mathtt{A}_\ell\}_{\ell \in L},y^\beta:\mathtt{B}}(Lx.k;\mathsf{Q}) \vdash_{\Lambda} \mathbf{y}^\beta:{[y^\beta:\mathtt{B}]}}{\infer[\mu L]{\mathbf{x}^\alpha:{[x^\alpha:\&\{\ell:\mathtt{A}_\ell\}_{\ell \in L}]}, \mathbf{c}^{\eta+1}:\mathsf{Msg}(x^\alpha.k) \otimes \mathsf{Cfg}_{x^{\alpha+1}:\mathtt{A}_k,y^\beta:\mathtt{B}}(\mathsf{Q}) \vdash_{\Lambda_1} \mathbf{y}^\beta:{[y^\beta:\mathtt{B}]}}{\infer[\&L]{\mathbf{x}^{\alpha+1}:\&\{ \ell:{\mathsf{Msg}(\ell, x^{\alpha+1}:\mathtt{A}_\ell) \multimap[x^{\alpha+1}:A_\ell]}\}_{\ell \in L}, \mathbf{c}^{\eta+1}:\mathsf{Msg}(x^\alpha.k) \otimes \mathsf{Cfg}_{x^{\alpha+1}:\mathtt{A}_k,y^\beta:\mathtt{B}}(\mathsf{Q})\vdash_{\Lambda_2}  \mathbf{y}^\beta:{[y^\beta:\mathtt{B}]}}{\infer[\otimes L]{\mathbf{x}^{\alpha+1}:\mathsf{Msg}(x^\alpha.k) \multimap {[x^{\alpha+1}:\mathtt{A}_k]}\}, \mathbf{c}^{\eta+1}:\mathsf{Msg}(x^\alpha.k) \otimes \mathsf{Cfg}_{x^{\alpha+1}:\mathtt{A}_k,y^\beta:\mathtt{B}}(\mathsf{Q}) \vdash_{\Lambda_2} \mathbf{y}^\beta:{[y^\beta:\mathtt{B}]}}{\infer[\multimap L]{\mathbf{x}^{\alpha+1}:\mathsf{Msg}(x^\alpha.k) \multimap {[x^{\alpha+1}:\mathtt{A}_k]}, \mathbf{w}^\omega:\mathsf{Msg}(x^\alpha.k), \mathbf{c}^{\eta+1}:\mathsf{Cfg}_{x^{\alpha+1}:\mathtt{A}_k,y^\beta:\mathtt{B}}(\mathsf{Q}) \vdash_{\Lambda_3} \mathbf{y}^\beta:{[y^\beta:\mathtt{B}]}}{\infer[\msc{ID}]{\mathbf{w}^\omega:\mathsf{Msg}(x^\alpha.k)\vdash_{\Lambda_3} \mathbf{z}^{\kappa}:\mathsf{Msg}(x^\alpha.k)}{} & \deduce{\mathbf{x}^{\alpha+1}:{[x^{\alpha+1}:\mathtt{A}_k]},  \mathbf{c}^{\eta+1}:\mathsf{Cfg}_{x^{\alpha+1}:\mathtt{A}_k,y^\beta:\mathtt{B}}(\mathsf{Q}) \vdash_{\Lambda_3} \mathbf{y}^\beta:{[y^\beta:\mathtt{B}]}}{\star}}}}}}
    \notag\end{equation}
}

{\small
$\Lambda_1=\Lambda \cup\{\mathbf{c}^{\eta+1}_{n+2}<\mathbf{c}^{\eta}_{n+2}\} \cup \{\mathbf{c}^{\eta+1}_{i}=\mathbf{c}^{\eta}_{i} \mid i \neq n+2\},$ 
$ \Lambda_2=\Lambda_1 \cup\{\mathbf{x}^{\alpha+1}_{n+2}<\mathbf{x}^{\alpha}_{n+2}\} \cup \{\mathbf{x}^{\alpha+1}_{i}=\mathbf{x}^{\alpha}_{i} \mid i \neq n+2\}\,,$
$\Lambda_3=\Lambda_2 \cup\{\mathbf{c}^{\eta+1}=\mathbf{w}^\omega\}.$ }

\item {\bf Case 6.} $(Ry.k;\mathsf{Q})$ Dual to {\bf Case 5}.
\item {\bf Case 7.} $(\mathbf{case}Lx(\ell \Rightarrow \mathsf{Q}_\ell)_{\ell \in L})$

{\footnotesize \begin{equation}
    \hspace*{-2cm}\infer[\mu L]{\star \;  \mathbf{x}^\alpha: {[x^\alpha:\oplus\{\ell:\mathtt{A}_\ell\}_{\ell \in L}]}, \mathbf{c}^\eta:\mathsf{Cfg}_{x^\alpha:\oplus\{\ell:A_\ell\}_{\ell \in L}, y^\beta:\mathtt{B}}(\mathbf{case}Lx(\ell \Rightarrow \mathsf{Q}_\ell)_{\ell \in L}) \vdash_{\Lambda} {\mathbf{y}^\beta:[y^\beta:\mathtt{B}]}}{\infer[\mu L]{\mathbf{x}^\alpha: {[x^\alpha:\oplus\{\ell:\mathtt{A}_\ell\}_{\ell \in L}]}, \mathbf{c}^{\eta+1}:\&\{\ell: \mathsf{Msg}(y^\beta.\ell) \multimap \mathsf{Cfg}_{x^{\alpha+1}:\mathtt{A}_\ell, y^\beta:\mathtt{B}}(\mathsf{Q}_\ell)\}_{\ell \in L} \vdash_{\Lambda_1} {\mathbf{y}^\beta:[y^\beta:\mathtt{B}]}}{\infer[\oplus L]{\mathbf{x}^{\alpha+1}:\oplus\{ \ell:\mathsf{Msg}(y^\beta.\ell) \otimes {[x^{\alpha+1}:A_\ell]}\}_{\ell \in L}, \mathbf{c}^{\eta+1}:\&\{\ell: \mathsf{Msg}(y^\beta.\ell) \multimap \mathsf{Cfg}_{x^{\alpha+1}:\mathtt{A}_\ell, y^\beta:\mathtt{B}}(\mathsf{Q}_\ell)\}_{\ell \in L} \vdash_{\Lambda_2} {\mathbf{y}^\beta:[y^\beta:\mathtt{B}]}}{\forall k \in L &  \infer[\& L]{{\mathbf{x}^{\alpha+1}:\mathsf{Msg}(y^\beta.k) \otimes [x^{\alpha+1}:\mathtt{A}_k]}, \mathbf{c}^{\eta+1}:\&\{\ell:\mathsf{Msg}(y^\beta.\ell) \multimap \mathsf{Cfg}_{x^{\alpha+1}:\mathtt{A}_\ell, y^\beta:\mathtt{B}}(\mathsf{Q}_\ell)\}_{\ell \in L} \vdash_{\Lambda_2} {\mathbf{y}^\beta:[y^\beta:\mathtt{B}]}}{\infer[\otimes L]{{\mathbf{x}^{\alpha+1}:\mathsf{Msg}(y^\beta.k)\otimes[x^{\alpha+1}:\mathtt{A}_k]}, \mathbf{c}^{\eta+1}:\mathsf{Msg}(y^\beta.k) \multimap\mathsf{Cfg}_{x^{\alpha+1}:\mathtt{A}_k, y^\beta:\mathtt{B}}(\mathsf{Q}_k) \vdash_{\Lambda_2} {\mathbf{y}^\beta:[y^\beta:\mathtt{B}]}}{\infer[\multimap L]{\mathbf{w}^{\omega}:\mathsf{Msg}(y^\beta.k), \mathbf{x}^{\alpha+1}:[x^{\alpha+1}:\mathtt{A}_k], \mathbf{c}^{\eta+1}:\mathsf{Msg}(y^\beta.k) \multimap\mathsf{Cfg}_{x^{\alpha+1}:\mathtt{A}_k, y^\beta:\mathtt{B}}(\mathsf{Q}_k) \vdash_{\Lambda_3} {\mathbf{y}^\beta:[y^\beta:\mathtt{B}]}}{\infer[\msc{Id}]{\mathbf{w}^{\omega}:\mathsf{Msg}(y^\beta.k)\vdash_{\Lambda_3} \mathbf{z}^\zeta:\mathsf{Msg}(y^\beta.k)}{} & \hspace{-0.3cm}\deduce{ \mathbf{x}^{\alpha+1}:[x^{\alpha+1}:\mathtt{A}_k], \mathbf{c}^{\eta+1}:\mathsf{Cfg}_{x^{\alpha+1}:\mathtt{A}_k, y^\beta:\mathtt{B}}(\mathsf{Q}_k) \vdash_{\Lambda_3} {\mathbf{y}^\beta:[y^\beta:\mathtt{B}]}}{\star}}}}}}}\notag\end{equation}
}
{\small
$\Lambda_1=\Lambda \cup\{\mathbf{c}^{\eta+1}_{n+2}<\mathbf{c}^{\eta}_{n+2}\} \cup \{\mathbf{c}^{\eta+1}_{i}=\mathbf{c}^{\eta}_{i} \mid i \neq n+2\},$
$ \Lambda_2=\Lambda_1 \cup\{\mathbf{x}^{\alpha+1}_{n+2}<\mathbf{x}^{\alpha}_{n+2}\} \cup \{\mathbf{x}^{\alpha+1}_{i}=\mathbf{x}^{\alpha}_{i} \mid i \neq n+2\}\,,$
$\Lambda_3=\Lambda_2 \cup\{\mathbf{x}^{\alpha+1}=\mathbf{w}^\omega\}. $}\\

\item{\bf Case 8.} $(\mathbf{case}Ry(\ell \Rightarrow \mathsf{Q}_\ell)_{\ell \in L})$ Dual to {\bf Case 7}.
 
\item {\bf Case 9.} $(Lx.\nu_\mathtt{t};\mathsf{Q})$ where $\mathtt{t}=^k_{\nu}\mathtt{C}$

\begin{equation}
    \hspace*{-3cm}\infer[\mu L]{\star \;  \mathbf{x}^\alpha:{[x^\alpha:\mathtt{t}]}, \mathbf{c}^\eta:\mathsf{Cfg}_{x^\alpha:\mathtt{t},y^\beta:B}(Lx.\nu_\mathtt{t};\mathsf{Q}) \vdash_{\Lambda} \mathbf{y}^\beta: {[y^\beta:\mathtt{B}]}}{\infer[\nu L]{\mathbf{x}^\alpha:{[x^\alpha:\mathtt{t}]}, \mathbf{c}^{\eta+1}:\mathsf{Msg}(x^\alpha.\nu_t)\otimes \mathsf{Cfg}_{x^{\alpha+1}:\mathtt{C},y^\beta:\mathtt{B}}(\mathsf{Q}) \vdash_{\Lambda_1}  \mathbf{y}^\beta:{[y^\beta:\mathtt{B}]}}{\infer[\otimes L]{\mathbf{x}^{\alpha+1}:\mathsf{Msg}(x^\alpha.\nu_t) \multimap {[x^{\alpha+1}:\mathtt{C}]}, \mathbf{c}^{\eta+1}:\mathsf{Msg}(x^\alpha.\nu_t)\otimes \mathsf{Cfg}_{x^{\alpha+1}:\mathtt{C},y^\beta:\mathtt{B}}(\mathsf{Q}) \vdash_{\Lambda_2}  \mathbf{y}^\beta:{[y^\beta:B]} }{\infer[\multimap L]{\mathbf{x}^{\alpha+1}:\mathsf{Msg}(x^\alpha.\nu_t)\multimap{[x^{\alpha+1}:\mathtt{C}]}, \mathbf{w}^{\omega}:\mathsf{Msg}(x^\alpha.\nu_t), \mathbf{c}^{\eta+1}: \mathsf{Cfg}_{x^{\alpha+1}:\mathtt{C},y^\beta:\mathtt{B}}(\mathsf{Q}) \vdash_{\Lambda_3}  \mathbf{y}^\beta:{[y^\beta:\mathtt{B}]}}{\infer[\msc{ID}]{\mathbf{w}^{\omega}:\mathsf{Msg}(x^\alpha.\nu_t)\vdash_{\Lambda_3} \mathbf{z}^\zeta:\mathsf{Msg}(x^\alpha.\nu_t)}{} & \deduce{\mathbf{x}^{\alpha+1}: {[x^{\alpha+1}:\mathtt{C}]}, \mathbf{c}^{\eta+1}: \mathsf{Cfg}_{x^{\alpha+1}:\mathtt{C},y^\beta:\mathtt{B}}(\mathsf{Q}) \vdash_{\Lambda_3}  \mathbf{y}^\beta:{[y^\beta:\mathtt{B}]}}{\star}}}}}\notag\end{equation}

$\Lambda_2=\Lambda_1 \cup\{\mathbf{c}^{\eta+1}_{n+2}<\mathbf{c}^{\eta}_{n+2}\} \cup \{\mathbf{c}^{\eta+1}_{i}=\mathbf{c}^{\eta}_{i} \mid i \neq n+2\},$
$ \Lambda_1=\Lambda \cup \{\mathbf{x}^{\alpha+1}_{i}=\mathbf{x}^{\alpha}_{i} \mid i \neq k\}\,,$
$\Lambda_3=\Lambda_2 \cup\{\mathbf{c}^{\eta+1}=\mathbf{w}^\omega\}.$ 

In the proof of this case, we use the abbreviated definition of $[x^\alpha:\mathtt{t}]$. If we use the following expanded definition instead  \[\begin{array}{lcll}
 {[y^\beta: \mathtt{t}]} & =^{n+2}_{\mu} &  \mathsf{Unfold}_\mathtt{t}(y^{\beta})\\
 \mathsf{Unfold}_\mathtt{t}(y^{\beta}) & =^i_{\nu} & ( \mathsf{Msg}(y^\beta.\nu_t) \multimap[y^{\beta+1}:\mathtt{A}]),

 \end{array}\]
 we need to apply an extra $\mu L$ rule with priority $n+2$ on the predicate $[x^\alpha:\mathtt{t}]$. Immediately after this $\mu L$ rule, we apply the $\nu L$ rule with priority $k<n+1$. This implies that the extra $\mu L$ rule does not play a role in validity of the derivation. As a result, the validity of the derivation given here implies validity of the derivation using the non-abbreviated definition. Furthermore, if we use the non-abbreviated definition we will have $\mathbf{x}^{\alpha+2}: [x^{\alpha+1}:\mathtt{C}]$ at the end (the top) of the cycle. Thus, we decided to use the abbreviated definition to make sure that the generation of position variables steps at the same pace as the generation of channels. This decision will help us to describe the bisimulation between the derivation built here and the typing derivation of processes in a more elegant way.

\item {\bf Case 10.} $(Ry.\mu_\mathtt{t};\mathsf{Q})$ Dual to {\bf Case 9.}
\item {\bf Case 11.} $(\mathbf{case}Lx(\mu_\mathtt{t} \Rightarrow \mathsf{Q}))$ where $\mathtt{t}=^k_{\mu}\mathtt{C}$
\begin{equation}
    \hspace*{-2cm}\infer[\mu L]{\star \;  {\mathbf{x}^\alpha:[x^\alpha:\mathtt{t}]}, \mathbf{c}^\eta:\mathsf{Cfg}_{x^\alpha:\mathtt{t},y^\beta:\mathtt{B}}(\mathbf{case}Lx(\mu_\mathtt{t} \Rightarrow \mathsf{Q})) \vdash_{\Lambda} \mathbf{y}^\beta:{[y^\beta:\mathtt{B}]}}{\infer[\mu L]{\mathbf{x}^\alpha:[x^\alpha:\mathtt{t}], \mathbf{c}^{\eta+1}:\mathsf{Msg}(x^\alpha.\mu_t)\multimap\mathsf{Cfg}_{x^{\alpha+1}:\mathtt{C},y^\beta:\mathtt{B}}(\mathsf{Q}) \vdash_{\Lambda_1} \mathbf{y}^\beta:{[y^\beta:\mathtt{B}]}}{\infer[\otimes L]{\mathbf{x}^{\alpha+1}:\mathsf{Msg}(x^\alpha.\mu_t)\otimes {[x^{\alpha+1}:\mathtt{C}]}, \mathbf{c}^{\eta+1}:\mathsf{Msg}(x^\alpha.\mu_t)\multimap\mathsf{Cfg}_{x^{\alpha+1}:\mathtt{C},y^\beta:\mathtt{B}}(\mathsf{Q}) \vdash_{\Lambda_2} \mathbf{y}^\beta:{[y^\beta:\mathtt{B}]} }{\infer[\multimap L]{\mathbf{w}^{\omega}:\mathsf{Msg}(x^\alpha.\mu_t), \mathbf{x}^{\alpha+1}:{[x^{\alpha+1}:\mathtt{C}]}, \mathbf{c}^{\eta+1}:\mathsf{Msg}(x^\alpha.\mu_t)\multimap\mathsf{Cfg}_{x^{\alpha+1}:\mathtt{C},y^\beta:\mathtt{B}}(\mathsf{Q}) \vdash_{\Lambda_3} \mathbf{y}^\beta:{[y^\beta:\mathtt{B}]}}{\infer[\msc{ID}]{\mathbf{w}^{\omega}:\mathsf{Msg}(x^\alpha.\mu_t)\vdash_{\Lambda_3} \mathbf{z}^{\zeta}:\mathsf{Msg}(x^\alpha.\mu_t)}{}& \deduce{\mathbf{x}^{\alpha+1}:{[x^{\alpha+1}:\mathtt{C}]}, \mathbf{c}^{\eta+1}:\mathsf{Cfg}_{x^{\alpha+1}:\mathtt{C},y^\beta:\mathtt{B}}(\mathsf{Q}) \vdash_{\Lambda_3} \mathbf{y}^\beta:{[y^\beta:\mathtt{B}]}}{\star}}}}}\notag \end{equation}

$\Lambda_1=\Lambda \cup\{\mathbf{c}^{\eta+1}_{n+2}<\mathbf{c}^{\eta}_{n+2}\} \cup \{\mathbf{c}^{\eta+1}_{i}=\mathbf{c}^{\eta}_{i} \mid i \neq n+2\},$ 
$ \Lambda_2=\Lambda_1 \cup\{\mathbf{x}^{\alpha+1}_{k}<\mathbf{x}^{\alpha}_{k}\} \cup \{\mathbf{x}^{\alpha+1}_{i}=\mathbf{x}^{\alpha}_{i} \mid i \neq k\}\,,$
$\Lambda_3=\Lambda_2 \cup\{\mathbf{x}^{\alpha+1}=\mathbf{w}^\omega\}.$

\item {\bf Case 12} $(\mathbf{case}Ry(\nu_\mathtt{t} \Rightarrow \mathsf{Q}))$ Dual to {\bf Case 11}.
\item {\bf Case 13.} $(y \leftarrow \mathsf{Y}\leftarrow x)$
\[\infer[\nu L]{{\star \;  \mathbf{x}^\alpha:[x^\alpha:\mathtt{A}]}, \mathbf{c}^\eta:\mathsf{Cfg}_{x^\alpha:\mathtt{A},y^\beta:\mathtt{B}}(y \leftarrow \mathsf{Y}\leftarrow x) \vdash_{\Lambda} {\mathbf{y}^\beta: [y^\beta:\mathtt{B}]}}{ \deduce{{\mathbf{x}^\alpha: [x^{\alpha}:\mathtt{A}]}, \mathbf{c}^{\eta+1}:\mathsf{Cfg}_{x^{\alpha}:\mathtt{A},y^\beta:\mathtt{B}}(Q) \vdash_{\Lambda \cup\{c^{\eta+1}_{i}=c^{\eta}_{i}\mid i \neq n+1\} } {\mathbf{y}^\beta:[y^{\beta}:\mathtt{B}]}}{\star}}\]
\end{description}

In cases 1-13 a predicate annotated with a position variable $\mathbf{c}$ in a branch can be interpreted as the (potential) computational continuation of the predicate $\mathsf{Cfg}()$ in the conclusion (at the bottom) of the block. Also, the only predicates that occur as a cut formula are of the form $[x^\alpha:\mathtt{A}]$.

Furthermore, in all the cases a rule is applied on $[x^\alpha:\mathtt{A}]$ only if there is a process in the antecedents willing to receive or send a message along channels $x$ or $y$ via a $\mathsf{Msg}$ predicate: there is a predicate  in the antecedents of the form $\mathsf{Msg}(x^\alpha.b) \circ \mathsf{Cfg}_{x^{\alpha+1}:p,\_ }(P)$ (or $\&\{\mathsf{Msg}(x^\alpha.b_\ell) \circ \mathsf{Cfg}_{\_, x^{\alpha+1}:p_{\ell}}(P_{\ell})\}_{\ell \in L}$ in cases 7 and 8 where $\circ \in \{\otimes, \multimap\}$.

Observe that in each circular branch a position variable $\mathbf{v}^\delta$ annotating a predicate $[v^\delta:\mathtt{D}]$ is only related by $<$ with its own other generations $\mathbf{v}^\gamma$. These observations are important in particular in the proof of Theorem~\ref{thm:strongprogress}.
\end{proof}

We need to show that when defined over a guarded program the derivation introduced above is a valid proof. Since a configuration is always finite, it is enough to prove validity of the annotated derivation built using Cases 1-13 for a single process $\mathsf{P}$.

 We use a validity-preserving bisimulation between the annotated derivation and the typing derivation of process $\mathsf{P}$. The notation $\Omega\vDash a \le b$ stands for ``the relation $a \le b$ can be deduced from the reflexive transitive closure of set $\Omega$".
\begin{definition}\label{def:R}
Define relation $\mathcal{R}$ between process typing judgments \[y^\beta:B \vdash_{\Omega} \mathtt{P} ::(x^\alpha:A)\] defined over $\mathbf{\Sigma}$ and annotated sequents in {\small $\mathit{FIMALL}^{\infty}_{\mu,\nu}$}:
\[\begin{array}{lcl}
x^\alpha:\mathtt{A} \vdash_{\Omega} \mathsf{P} :: (y^\beta:\mathtt{B}) & \mathcal{R} & \mathbf{x}^\alpha:{[x^\alpha:\mathtt{A}]}, \mathbf{c}^{\eta}:\mathsf{Cfg}_{x^\alpha:\mathtt{A},y^\beta:\mathtt{B}}(\mathsf{P}) \vdash_{\Omega'} \mathbf{y}^\beta:{[y^\beta:\mathtt{B}]},\\
\cdot \vdash_{\Omega} \mathsf{P} :: (y^\beta:\mathtt{B}) & \mathcal{R} & \mathbf{c}^{\eta}:\mathsf{Cfg}_{\cdot,y^\beta:\mathtt{B}}(\mathsf{P}) \vdash_{\Omega'} \mathbf{y}^\beta:{[y^\beta:\mathtt{B}]},
\end{array}\]
where\footnote{
$n$ is the maximum priority of fixed points in $\mathbf{\Sigma}$.} for $i\le n$, 
\begin{itemize}
    \item $\Omega\vDash x^\alpha_i \le w^\zeta_i$  iff $\Omega' \vDash \mathbf{x}^\alpha_i \le \mathbf{w}^\zeta_i$, and
    \item $\Omega\vDash y^\beta_i \le w^\zeta_i$  iff $\Omega' \vDash \mathbf{y}^\beta_i \le \mathbf{w}^\zeta_i$.

\end{itemize}
\end{definition}

 We define two stepping rules over the typing derivation of a process and the derivation built in the proof of Lemma \ref{lem:derivation}.
\begin{definition}
We define two stepping rules $\hookrightarrow$ and $\Rightarrow$ over processes and annotated sequents, respectively: 
\begin{itemize}
    \item $\bar{x}^\alpha:\omega \vdash_{\Omega} \mathsf{P} ::(y^\beta:\mathtt{B}) \hookrightarrow \bar{z}^\delta:\omega' \vdash_{\Lambda} \mathsf{Q} ::(w^\gamma:\mathtt{D})$ iff there is a rule in the infinitary system of Figure \ref{fig:stp-order} of the form
\[\infer[]{\bar{x}^\alpha:\omega  \vdash_{\Omega} \mathsf{P} ::(y^\beta:\mathtt{B})}{\cdots & \bar{z}^\delta:\omega'\vdash_{\Lambda} \mathsf{Q} ::(w^\gamma:\mathtt{D}) & \cdots}\]
\item {\footnotesize${(i) \; \overline{\mathbf{x}^{\alpha}:{[\bar{x}^\alpha:\omega]}}, \mathbf{c}^{\eta_1}:\mathsf{Cfg}_{\bar{x}^\alpha:\omega,y^\beta:\mathtt{B}}(\mathsf{P}) \vdash_{\Omega} \mathbf{y}^{\beta}:{[y^\beta:\mathtt{B}]} \Rightarrow (ii) \; \overline{\mathbf{z}^{\delta}:{[\bar{z}^\delta:\omega']}}, \mathbf{d}^{\eta_2}:\mathsf{Cfg}_{\bar{z}^\delta:\omega',w^\gamma:\mathtt{D}}(\mathsf{Q}) \vdash_{\Lambda} \mathbf{w}^{\gamma}:{[w^\gamma:\mathtt{D}]}}$}
iff (i) is the conclusion of one of the blocks 1-13 in the proof of Lemma \ref{lem:derivation} and (ii) is a $\star$ assumption of it.
\end{itemize}
\end{definition}

Next we prove that relation $\mathcal{R}$ is a validity preserving bisimulation with regard to the steppings $\hookrightarrow$ and $\Rightarrow$.
\begin{lemma}\label{thm:bisim}
$\mathcal{R}$ forms a bisimulation between the derivation given for {\[\star \, \overline{\mathbf{x}^\alpha:[\bar{x}^\alpha:\omega]}, \mathbf{c}^\eta:\mathsf{Cfg}_{\bar{x}^\alpha:\omega, y^\beta:\mathtt{B}}(\mathsf{P})\vdash_{\Omega'}  \mathbf{y}^\beta:[y^\beta:\mathtt{B}]\]}  and the typing derivation of process \[\bar{x}^\alpha:\omega \vdash_{\Omega} \mathsf{P} :: (y^\beta:\mathtt{B}).\]

\end{lemma}

\begin{proof}
The proof of this bisimulation is straightforward. It follows from the way we built proof blocks for each case in Lemma~\ref{lem:derivation}, and the typing rules for annotated processes in Figure~\ref{fig:stp-order}:
\begin{description}
\item {\bf $\rightarrow$:} 
\[
\begin{tikzcd}
 \bar{x}^\alpha:\omega \vdash_{\Omega} \mathsf{P} :: (y^\beta:\mathtt{B}) \arrow[r, dash, "\mathcal{R}"] \arrow[d, hookrightarrow] &  \overline{\mathbf{x}^\alpha:[\bar{x}^\alpha:\omega]}, \mathbf{c}^\eta:\mathsf{Cfg}_{\bar{x}^\alpha:\omega, y^\beta:\mathtt{B}}(\mathsf{P})\vdash_{\Omega'}  \mathbf{y}^\beta:[y^\beta:\mathtt{B}] \arrow[d, Rightarrow, dashed] \\
  \bar{z}^\gamma:\omega' \vdash_{\Lambda} \mathsf{Q} :: (w^\delta:\mathtt{D}) \arrow[r, dash, dashed, "\mathcal{R}"'] & \overline{\mathbf{z}^\gamma:{[\bar{z}^\gamma:\omega']}}, \mathbf{d}^\theta:\mathsf{Cfg}_{\bar{z}^\gamma:\omega', w^\delta:\mathtt{D}}(\mathsf{Q})\vdash_{\Lambda'} \mathbf{w}^{\delta}: {[w^\delta:\mathtt{D}]}
\end{tikzcd}
\]
The proof is by considering different cases for $\hookrightarrow$:

\begin{description}
\item {\bf Case 1.} ($\mathbf{wait}Lx; \mathsf{Q}$)
\[\begin{array}{lr}
    1.\mbox{Premise} &  x^\alpha:1 \vdash_{\Omega} \mathbf{wait}Lx;\mathsf{Q}::y^\beta:\mathtt{B}  \hookrightarrow \cdot \vdash_{\Omega}\mathsf{Q}:: y^\beta:\mathtt{B}\\ 
2. \mbox{Block 3}  & \mathbf{x}^\alpha:{[x^\alpha:1]}, \mathbf{c}^\eta:\mathsf{Cfg}_{x^\alpha:1, y^\beta:B}(\mathbf{wait}Lx;\mathsf{Q}) \vdash_{\Omega'} \mathbf{y}^\beta:{[y^\beta:\mathtt{B}]} \Rightarrow\\ & \mathbf{c}^{\eta+1}: \mathsf{Cfg}_{\cdot, y^\beta:\mathtt{B}}(\mathsf{Q}) \vdash_{\Lambda'} \mathbf{y}^\beta:{[y^\beta:\mathtt{B}]}\\
3.\mbox{Definition of $\Lambda'$}& \mbox{for}\, i \le n,  \Lambda'\vDash \mathbf{y}^\beta_i \le \mathbf{z}_i^\gamma\; \mbox{iff}\; \Omega'\vDash \mathbf{y}^\beta_i \le \mathbf{z}_i^\gamma.\\
4. \mbox{By assumption and line 3}& \mbox{for}\, i \le n,  \Lambda'\vDash \mathbf{y}^\beta_i \le \mathbf{z}_i^\gamma\; \mbox{iff}\; \Omega\vDash y^\beta_i \le z_i^\gamma \\
5. \mbox{By line 4} & \cdot \vdash_{\Omega}\mathsf{Q}:: y^\beta:\mathtt{B}\; \mathcal{R}\;\mathbf{c}^{\eta+1}: \mathsf{Cfg}_{\cdot, y^\beta:\mathtt{B}}(\mathsf{Q}) \vdash_{\Lambda'} \mathbf{y}^\beta:{[y^\beta:\mathtt{B}]}
\end{array}\]

\item {\bf Case 2.} ($Lx. \nu_\mathtt{t}; \mathsf{Q}$)
{\footnotesize\[\begin{array}{lr}
1. \mbox{Premise} &  x^\alpha:\mathtt{t} \vdash_{\Omega} \mathtt{Lx.\nu_\mathtt{t};\mathsf{Q}}::y^\beta:\mathtt{B}  \hookrightarrow x^{\alpha+1}:\mathtt{C} \vdash_{\Lambda}Q:: y^\beta:\mathtt{B}\\ 
2. \mbox{Block 9}  & \mathbf{x}^\alpha:{[x^\alpha:\mathtt{t}]}, \mathbf{c}^\eta:\mathsf{Cfg}_{x^\alpha:\mathtt{t}, y^\beta:\mathtt{B}}(Lx.\nu_\mathtt{t};\mathsf{Q}) \vdash_{\Omega'} \mathbf{y}^\beta:{[y^\beta:\mathtt{B}]} \Rightarrow\\ & \mathbf{x}^{\alpha+1}:{[x^{\alpha+1}:\mathtt{C}]}, \mathbf{c}^{\eta+1}: \mathsf{Cfg}_{x^{\alpha+1}:\mathtt{C}, y^\beta:\mathtt{B}}(\mathsf{Q}) \vdash_{\Lambda'} \mathbf{y}^\beta:{[y^\beta:\mathtt{B}]}\\
3.\mbox{Definition of $\Lambda'$}& \mbox{for}\, i \le n,  \Lambda'\vDash \mathbf{y}^\beta_i \le \mathbf{z}_i^\gamma\; \mbox{iff}\; \Omega'\vDash \mathbf{y}^\beta_i \le \mathbf{z}_i^\gamma.\\
4.\mbox{Definition of $\Lambda$}& \mbox{for}\, i \le n,  \Lambda\vDash y^\beta_i \le z_i^\gamma\; \mbox{iff}\; \Omega\vDash y^\beta_i \le z_i^\gamma.\\
5.\mbox{By assumption}& \mbox{for}\, i \le n,  \Lambda' \vDash \mathbf{y}^\beta_i \le \mathbf{z}_i^\gamma\; \mbox{iff}\; \Lambda \vDash y^\beta_i \le z_i^\gamma \\
6.\mbox{Definition of $\Lambda'$}& \mbox{for}\, i \le n \;\mbox{and}\; z^\gamma \neq x^{\alpha+1},  \Lambda'\vDash \mathbf{x}^{\alpha+1}_i \le \mathbf{z}_i^\gamma\; \mbox{iff}\; \Omega'\vDash \mathbf{x}^\alpha_i \le \mathbf{z}_i^\gamma.\\
7.\mbox{Definition of $\Lambda$}& \mbox{for}\, i \le n, \;\mbox{and}\; z^\gamma \neq x^{\alpha+1},   \Lambda\vDash x^{\alpha+1}_i \le z_i^\gamma\; \mbox{iff}\; \Omega\vDash x^{\alpha}_i \le z_i^\gamma.\\
8. \mbox{By assumption}& \mbox{for}\, i \le n,  \Lambda' \vDash \mathbf{x}^{\alpha+1}_i \le \mathbf{z}_i^\gamma\; \mbox{iff}\; \Lambda \vDash x^{\alpha+1}_i \le z_i^\gamma \\
9. \mbox{By lines 5 and 9} &  x^{\alpha+1}:\mathtt{C} \vdash_{\Lambda}Q:: y^\beta:\mathtt{B}\; \mathcal{R}\;\mathbf{x}^{\alpha+1}:{[x^{\alpha+1}:\mathtt{C}]}, \mathbf{c}^{\eta+1}: \mathsf{Cfg}_{x^{\alpha+1}:\mathtt{C}, y^\beta:\mathtt{B}}(Q) \vdash_{\Lambda'} \mathbf{y}^\beta:{[y^\beta:\mathtt{B}]}
\end{array}\]}
\item {\bf Cases.} The proof of other cases is similar to the previous ones. 
\end{description}
\item {\bf $\leftarrow$.}
\[
\begin{tikzcd}
  \bar{x}^\alpha:\omega \vdash_{\Omega} \mathsf{P} :: (y^\beta:\mathtt{B}) \arrow[r, dash, "\mathcal{R}"] \arrow[d, hookrightarrow, dashed] &  \overline{\mathbf{x}^\alpha:[\bar{x}^\alpha:\omega]}, \mathbf{c}^\eta:\mathsf{Cfg}_{\bar{x}^\alpha:\omega, y^\beta:\mathtt{B}}(\mathsf{P})\vdash_{\Omega'}  \mathbf{y}^\beta:[y^\beta:\mathtt{B}] \arrow[d, Rightarrow] \\
  \bar{z}^\gamma:\omega' \vdash_{\Lambda} \mathsf{\mathsf{Q}} :: (w^\delta:\mathtt{D}) \arrow[r, dash, dashed, "\mathcal{R}"'] & \overline{\mathbf{z}^\gamma:{[\bar{z}^\gamma:\omega'}]}, \mathbf{d}^\theta:\mathsf{Cfg}_{\bar{z}^\gamma:\omega', w^\delta:\mathtt{D}}(\mathsf{Q})\vdash_{\Lambda'} \mathbf{w}^{\delta}: {[w^\delta:D]}
\end{tikzcd} \]
The proof is by considering different cases for $\Rightarrow$.

\begin{description}
\item {\bf Case 1.} ($\mathbf{wait}Lx; \mathsf{Q}$)
\[\begin{array}{lr}
1. \mbox{Premise}  & \mathbf{x}^\alpha:{[x^\alpha:1]}, \mathbf{c}^\eta:\mathsf{Cfg}_{x^\alpha:1, y^\beta:B}(\mathbf{wait}Lx;\mathsf{Q}) \vdash_{\Omega'} \mathbf{y}^\beta:{[y^\beta:\mathtt{B}]} \Rightarrow\\ & \mathbf{c}^{\eta+1}: \mathsf{Cfg}_{\cdot, y^\beta:\mathtt{B}}(\mathsf{Q}) \vdash_{\Lambda'} \mathbf{y}^\beta:{[y^\beta:\mathtt{B}]}\\
2.\mbox{By $1L$ typing rule} &  x^\alpha:1 \vdash_{\Omega} \mathbf{wait}Lx;\mathsf{Q}::y^\beta:\mathtt{B}  \hookrightarrow \cdot \vdash_{\Omega}\mathsf{Q}:: y^\beta:\mathtt{B}\\ 
3.\mbox{Definition of $\Lambda'$}& \mbox{for}\, i \le n,  \Lambda'\vDash \mathbf{y}^\beta_i \le \mathbf{z}_i^\gamma\; \mbox{iff}\; \Omega'\vDash \mathbf{y}^\beta_i \le \mathbf{z}_i^\gamma.\\
4. \mbox{By assumption and line 3}& \mbox{for}\, i \le n,  \Lambda'\vDash \mathbf{y}^\beta_i \le \mathbf{z}_i^\gamma\; \mbox{iff}\; \Omega\vDash y^\beta_i \le z_i^\gamma \\
5. \mbox{By line 4} & \cdot \vdash_{\Omega}\mathsf{Q}:: y^\beta:\mathtt{B}\; \mathcal{R}\;\mathbf{c}^{\eta+1}: \mathsf{Cfg}_{\cdot, y^\beta:\mathtt{B}}(\mathsf{Q}) \vdash_{\Lambda'} \mathbf{y}^\beta:{[y^\beta:\mathtt{B}]}
\end{array}\]
\item {\bf Case 2.} ($Lx. \nu_t; \mathsf{Q}$)
{\footnotesize\[\begin{array}{lr}
    1. \mbox{Premise}  & \mathbf{x}^\alpha:{[x^\alpha:\mathtt{t}]}, \mathbf{c}^\eta:\mathsf{Cfg}_{x^\alpha:\mathtt{t}, y^\beta:\mathtt{B}}(Lx.\nu_\mathtt{t};Q) \vdash_{\Omega'} \mathbf{y}^\beta:{[y^\beta:\mathtt{B}]} \Rightarrow\\ & \mathbf{x}^{\alpha+1}:{[x^{\alpha+1}:\mathtt{C}]}, \mathbf{c}^{\eta+1}: \mathsf{Cfg}_{x^{\alpha+1}:\mathtt{C}, y^\beta:\mathtt{B}}(Q) \vdash_{\Lambda'} \mathbf{y}^\beta:{[y^\beta:\mathtt{B}]}\\
    2. \mbox{By $\nu L$ typing rule} &  x^\alpha:\mathtt{t} \vdash_{\Omega} \mathtt{Lx.\nu_t;Q}::y^\beta:\mathtt{B}  \hookrightarrow x^{\alpha+1}:\mathtt{C} \vdash_{\Lambda}Q:: y^\beta:\mathtt{B}\\ 
3.\mbox{Definition of $\Lambda'$}& \mbox{for}\, i \le n,  \Lambda'\vDash \mathbf{y}^\beta_i \le \mathbf{z}_i^\gamma\; \mbox{iff}\; \Omega'\vDash \mathbf{y}^\beta_i \le \mathbf{z}_i^\gamma.\\
4.\mbox{Definition of $\Lambda$}& \mbox{for}\, i \le n,  \Lambda\vDash y^\beta_i \le z_i^\gamma\; \mbox{iff}\; \Omega\vDash y^\beta_i \le z_i^\gamma.\\
5. \mbox{By assumption}& \mbox{for}\, i \le n,  \Lambda' \vDash \mathbf{y}^\beta_i \le \mathbf{z}_i^\gamma\; \mbox{iff}\; \Lambda \vDash y^\beta_i \le z_i^\gamma \\
6.\mbox{Definition of $\Lambda'$}& \mbox{for}\, i \le n \;\mbox{and}\; z^\gamma \neq x^{\alpha+1},  \Lambda'\vDash \mathbf{x}^{\alpha+1}_i \le \mathbf{z}_i^\gamma\; \mbox{iff}\; \Omega'\vDash \mathbf{x}^\alpha_i \le \mathbf{z}_i^\gamma.\\
7.\mbox{Definition of $\Lambda$}& \mbox{for}\, i \le n, \;\mbox{and}\; z^\gamma \neq x^{\alpha+1},   \Lambda\vDash x^{\alpha+1}_i \le z_i^\gamma\; \mbox{iff}\; \Omega\vDash x^{\alpha}_i \le z_i^\gamma.\\
8. \mbox{By assumption}& \mbox{for}\, i \le n,  \Lambda' \vDash \mathbf{x}^{\alpha+1}_i \le \mathbf{z}_i^\gamma\; \mbox{iff}\; \Lambda \vDash x^{\alpha+1}_i \le z_i^\gamma \\
9. \mbox{By lines 5 and 9} &  x^{\alpha+1}:\mathtt{C} \vdash_{\Lambda}Q:: y^\beta:\mathtt{B}\; \mathcal{R}\;\mathbf{x}^{\alpha+1}:{[x^{\alpha+1}:\mathtt{C}]}, \mathbf{c}^{\eta+1}: \mathsf{Cfg}_{x^{\alpha+1}:\mathtt{C}, y^\beta:\mathtt{B}}(Q) \vdash_{\Lambda'} \mathbf{y}^\beta:{[y^\beta:\mathtt{B}]}
\end{array}\]}
\item {\bf Cases.} Similar to the previous cases.
\end{description}
\end{description}
\end{proof}

\begin{lemma}\label{thm:validmain} If \[\bar{x}^\alpha:\omega \vdash_{\emptyset} \mathsf{P}:: (y^\beta:B)\] is a guarded process, then a derivation built in Lemma~\ref{lem:derivation} for 
\[\star \, \overline{\mathbf{x^\alpha}:[\bar{x}^\alpha:\omega]}, \mathbf{c}^\eta:\mathsf{Cfg}_{ \bar{x}^\alpha:\omega,y^\beta:\mathtt{B}}(\mathsf{P})\vdash_{\emptyset} \mathbf{y}^\beta:[y^\beta:\mathtt{B}] \] is valid.
\end{lemma}
\begin{proof}
By assumption there is an (infinitary) guarded typing derivation $\mathcal{D}_1$ for process $\bar{x}^\alpha:\omega \vdash_{\emptyset} \mathsf{P}:: (y^\beta:\mathtt{B})$. By Lemma \ref{thm:bisim}, we build a bisimilar (infinite) derivation for $\star \, \overline{\mathbf{x^\alpha}:[\bar{x}^\alpha:\omega]}, \mathbf{c}^\eta:\mathsf{Cfg}_{ \bar{x}^\alpha:\omega,y^\beta:\mathtt{B}}(\mathsf{P})\vdash_{\emptyset} \mathbf{y}^\beta:[y^\beta:\mathtt{B}] $ using Cases (1-13) in Lemma \ref{lem:derivation}. Consider an infinite path $\mathtt{p}_2$ in $\mathcal{D}_2$ and its bisimilar path $\mathtt{p}_1$ in $\mathcal{D}_1$:
\[
\begin{tikzcd}
  \cdots \arrow[d, hookrightarrow,"*"] & \cdots \arrow[d, Rightarrow, "*"] \\
  \bar{x}^\alpha:\omega \vdash_{\Omega} \mathsf{P} :: (y^\beta:\mathtt{B}) \arrow[r, dash, "\mathcal{R}"] \arrow[d, hookrightarrow, "k"] &  \overline{\mathbf{x}^\alpha:[\bar{x}^\alpha:\omega]}, \mathbf{c}^\eta:\mathsf{Cfg}_{\bar{x}^\alpha:\omega, y^\beta:\mathtt{B}}(\mathsf{P})\vdash_{\Omega'}  \mathbf{y}^\beta:[y^\beta:\mathtt{B}] \arrow[d, Rightarrow, "k"] \\
  \bar{z}^\gamma:\omega' \vdash_{\Lambda} \mathsf{Q} :: (w^\delta:\mathtt{D}) \arrow[r, dash, "\mathcal{R}"'] \arrow[d, hookrightarrow, "*"] & \overline{\mathbf{z}^\gamma:{[\bar{z}^\gamma:\omega']}}, \mathbf{d}^\theta:\mathsf{Cfg}_{\bar{z}^\gamma:\omega', w^\delta:\mathtt{D}}(\mathsf{Q})\vdash_{\Lambda'} \mathbf{w}^{\delta}: {[w^\delta:\mathtt{D}]}  \arrow[d, Rightarrow, "*"]\\
  \cdots & \cdots
\end{tikzcd} \]

By definition of $\mathcal{R}$,
\begin{itemize}
    \item if $\Lambda \vDash \mathsf{snap}(z^\gamma)<\mathsf{snap}(x^\alpha)$, then 
$\Lambda' \vDash \mathsf{snap}(\mathtt{z}^\gamma)<\mathtt{snap}(\mathsf{x}^\alpha)$, and
\item if $\Lambda \vDash \mathsf{snap}(w^\delta)<\mathsf{snap}(y^\beta)$, then 
$\Lambda' \vDash \mathsf{snap}(\mathtt{w}^\delta)<\mathtt{snap}(\mathsf{y}^\beta)$.
\end{itemize}
 By the above property if path $\mathtt{p}_1$ is guarded, then $\mathtt{p}_2$ is valid.
\end{proof}

\begin{customthm}{2.}
For the judgment $\bar{x}^\alpha:\omega\Vdash \mathcal{C} ::(y^\beta:\mathtt{B})$ where $\mathcal{C}$ is a configuration of guarded processes, there is a valid cut-free proof for $ \mathcal{C} \in \llbracket \bar{x}^\alpha:\omega \vdash y^\beta:\mathtt{B} \rrbracket$ in {\small $\mathit{FIMALL}^{\infty}_{\mu,\nu}$}. Validity of this proof ensures the strong progress property of $\mathcal{C}$.
\end{customthm}

 \begin{proof}
It is enough to show that the valid proof we built for \[\star {\,\overline{[\bar{x}^\alpha:\omega]}, \mathsf{Cfg}_{\bar{x}^\alpha:\omega, y^\beta:\mathtt{B}}(\mathcal{C})\vdash  [y^\beta:\mathtt{B}]}\] in {\small $\mathit{FIMALL}^{\infty}_{\mu,\nu}$} ensures the strong progress property of configuration $\mathcal{C}$. Here $\bar{x}^\alpha$ and $y^\beta$ are external channels of the configuration $\mathcal{C}$. 

We run the cut elimination algorithm on the valid proof we built in Lemma \ref{lem:derivation}. By the structure of the proof, we know that the first step in the cut elimination algorithm is to apply an external reduction (Flip rule) on $\mathsf{Cfg}_{x^\alpha:\omega, y^\beta:\mathtt{B}}(\mathcal{C})$ to unfold its definition. The tape transforms to \[ {\,\overline{[\bar{x}^\alpha:\omega]}, T \vdash  [y^\beta:\mathtt{B}]}\] where $T$ is the definition given based on the pattern of $\mathcal{C}$.

In fact, we can prove that throughout the cut elimination procedure, we repeatedly get a branching tape of the form\footnote{There may be other judgments of the form $T \vdash T$ on the tape that do not contribute to the computational meaning of the tape. Without loss of generality, we do not present them in the next line, but we will explain how to deal with them in the proof.}
\[{\overline{[\bar{x}^\alpha:\omega]}, T_1 \vdash  [z_1^{\eta_1}:\mathtt{D_1}]  ,[z_1^{\eta_1}:\mathtt{D_1}], T_2 \vdash [z_2^{\eta_2}:\mathtt{D_2}], \cdots [z_m^{\eta_m}:\mathtt{D_m}], T_{m+1} \vdash  [y^\beta:\mathtt{B}]} \] where $T_{i+1}$ is the definition of a predicate $\mathsf{Cfg}_{z_{i}^{\eta_i}:\mathtt{D}_i, z_{i+1}^{\eta_{i+1}}:\mathtt{D}_{i+1}}(\mathcal{C}_{i+1})$, and it explains the behavior of the computational continuation of $\mathcal{C}$ with regards to  channels $\bar{z}_i^{\eta_i}$ and $\bar{z}_{i+1}^{\eta_{i+1}}$.  Moreover, if an external reduction rule is applied on $[x^\alpha:\omega]$ or $[y^\beta:\mathtt{B}]$, there is a process willing to send or receive a message along $x^\alpha$ or $y^\beta$. 
\footnote{Put $z_0^{\eta_0}=x^\alpha$ and $z_{m+1}^{\eta_{m+1}}=y^{\beta}$. Note that channels $z_i^{\eta_i}$ for $1\le i \le m$ are the internal channels of the configuration, and  channels $x^\alpha$ and $y^\beta$ are the external channels.}

This property holds after the very first rule of cut elimination (an external reduction on $\mathsf{Cfg}(\mathcal{C})$ which leads to the tape ${\,\overline{[\bar{x}^\alpha:\omega]}, T \vdash  [y^\beta:\mathtt{B}]}$). We want to prove that the property explained in the previous paragraph holds as an invariant on the tapes being produced by the cut elimination algorithm if we apply the cut elimination algorithm in the following order \footnote{Our cut elimination algorithm is non-deterministic in the sense that there might be two applicable reductions at each time. However, we proved that no matter what rule we choose, the algorithm will terminate on valid proofs. Here, we single out one particular way of running the cut-elimination property.}:





 \begin{description}

\item {\bf Step 1.  Communication along an external channel:} an external reduction can be applied on predicates $[x^\alpha:\omega]$ or $[y^\beta:\mathtt{B}]$.\\
 In Figure~\ref{fig:extred} we provide the steps of our cut elimination algorithm when a rule can be applied on a predicate $[x^\alpha:\oplus\{\ell:\mathtt{A}_\ell\}_{\ell \in L}]$. The cases for other types are similar. First observe that by the structure of the proof, this is only the case if in the configuration a process communicates along a left external channel $x^\alpha$ of type $\oplus\{\ell:\mathtt{A}_\ell\}_{\ell \in L}$.  Moreover, in Line 4 of Figure~\ref{fig:extred} the algorithm creates multiple branches: the continuation of processes $\mathbf{case}Lx(\ell \Rightarrow \mathsf{Q}_\ell)_{\ell \in L}$ depends on the  potential label $k \in L$ that it receives along the external channel $x^\alpha$. The tape we have at the end of each branch corresponds to a potential configuration in the computation. On line (8) we unfold the definition of  $\mathsf{Cfg}_{x^{\alpha+1};\mathtt{A}_k,v^\delta:\mathtt{D}}(Q_k)$ predicate to get the invariant we are looking for. We start over the cut elimination procedure for the new tape.
 
 On line (7) we create an extra branch containing a single (green) judgment 
 \[{\mathsf{Msg}(x^\alpha.k) \vdash \mathsf{Msg}(x^\alpha.k)}.\] This tape can be closed by a single ID-elim rule.
 
 \begin{figure}
    \centering
 { \[\hspace*{-3cm}
 {\begin{array}{ll}
(1) \qquad   {[x^\alpha:\oplus\{\ell:\mathtt{A}_\ell\}_{\ell \in L}]}, {\color{red}\mathsf{Cfg}_{x^\alpha:\oplus\{\ell:A_\ell\}_{\ell \in L}, v^\delta:\mathtt{D}}(\mathbf{case}Lx(\ell \Rightarrow \mathsf{Q}_\ell)_{\ell \in L})} \vdash {[v^\delta:\mathtt{D}]} \\ 
 \hspace{9cm} \Downarrow \mbox{External reduction}(\mu) \notag\\
\hline\\
 (2) \qquad {\color{red}[x^\alpha:\oplus\{\ell:\mathtt{A}_\ell\}_{\ell \in L}]}, \&\{\ell:  \mathsf{Msg}(x^\alpha.\ell) \multimap \mathsf{Cfg}_{x^{\alpha+1}:\mathtt{A}_\ell, v^\delta:\mathtt{D}}(\mathsf{Q}_\ell)\}_{\ell\in L} \vdash {[v^\delta:\mathtt{D}]} \\
 \hspace{9cm}  \Downarrow \mbox{External reduction}(\mu) \notag\\
   \hline\\
  (3) \qquad{\color{red}\oplus\{\ell:(\mathsf{Msg}(x^\alpha.\ell) \otimes [x^{\alpha+1}:\mathtt{A}_\ell])\}_{\ell \in L}}, \&\{\ell:  \mathsf{Msg}(x^\alpha.\ell) \multimap \mathsf{Cfg}_{x^{\alpha+1}:\mathtt{A}_\ell, v^\delta:\mathtt{D}}(\mathsf{Q}_\ell)\}_{\ell\in L} \vdash {[v^\delta:\mathtt{D}]}  \\
   \hspace{9cm}\Downarrow \color{purple}{\mbox{External reduction}(\oplus)} \notag\\
  \hline\\ (4) \qquad\forall k \in L \qquad (\mathsf{Msg}(x^\alpha.k) \otimes [x^{\alpha+1}:\mathtt{A}_k]), {\color{red}\&\{\ell:  \mathsf{Msg}(x^\alpha.\ell) \multimap \mathsf{Cfg}_{x^{\alpha+1}:\mathtt{A}_\ell, v^\delta:\mathtt{D}}(\mathsf{Q}_\ell)\}_{\ell\in L}} \vdash {[v^\delta:\mathtt{D}]} \\
  \hspace{9cm} \Downarrow \mbox{External reduction}(\&) \notag\\
   \hline\\
(5) \qquad\forall k \in L\qquad {\color{red}(\mathsf{Msg}(x^\alpha.k) \otimes [x^{\alpha+1}:\mathtt{A}_k])},  (\mathsf{Msg}(x^\alpha.k) \multimap \mathsf{Cfg}_{x^{\alpha+1}:\mathtt{A}k, v^\delta:\mathtt{D}}(\mathsf{Q}_k) \vdash {[v^\delta:\mathtt{D}]} \\
       \hspace{9cm}  \Downarrow \mbox{External reduction}(\otimes) \notag\\\hline\\
 (6) \qquad   \forall k \in L\qquad  \mathsf{Msg}(x^\alpha.k), [x^{\alpha+1}:\mathtt{A}_k], {\color{red} \mathsf{Msg}(x^\alpha.k) \multimap \mathsf{Cfg}_{x^{\alpha+1}:\mathtt{A}k, v^\delta:\mathtt{D}}(\mathsf{Q}_k)} \vdash {[v^\delta:\mathtt{D}]} \\
  \hspace{9cm}  \Downarrow \mbox{External reduction}(\multimap) \notag\\\hline\\
(7) \qquad \forall k \in L\qquad  [x^{\alpha+1}:\mathtt{A}_k], {\color{red}\mathsf{Cfg}_{x^{\alpha+1}:\mathtt{A}k, v^\delta:\mathtt{D}}(\mathsf{Q}_k)} \vdash {[v^\delta:\mathtt{D}]} \quad
         \hspace{1.68cm}{\color{green}\mathsf{Msg}(x^\alpha.k) \vdash \mathsf{Msg}(x^\alpha.k)} 
 (\mbox{Identity elimination}) \notag\\ 
\hspace{6cm}    \Downarrow \mbox{External Reduction}(\mu) \notag\\  \hline\\
(8) \qquad    \forall k \in L\qquad  [x^{\alpha+1}:\mathtt{A}_k], T_k \vdash {[v^\delta:\mathtt{D}]} \\
\end{array}}
\]}
    \caption{A run of the cut elimination algorithm when there is a process communicating along an external channel.}
    \label{fig:extred}
\end{figure}

 \item {\bf Step 2. An external reduction rule can be applied  on a $T_i$}:
 \begin{itemize}
 \item {\bf 2.1. Identity:} Consider all judgments on the tape which are of the form
\[[u^\gamma:\mathtt{F}], u^\gamma=w^\delta \vdash [w^\delta:\mathtt{F}].\]
We first apply external reductions ($=L$) on the antecedents of the form $u^\gamma=w^\delta$. This rule renames channels $u^\gamma$ and $w^\delta$ with their most general unifier $z^\eta$. 

We ignore the remainder of the judgments of the form $[z^\eta:\mathtt{F}] \vdash [z^\eta:\mathtt{F}]$ until {\bf Step 2.4}.  The rest of the tape preserves the property of interest, since we only renamed computationally identical channels in it.

     \item {\bf 2.2. Spawn:} Consider all  the judgments on the tape that are of the form \[[\bar{u}^\gamma:\omega'], \exists v. \exists \zeta. (\mathsf{Cfg}_{\bar{u}^\gamma:\omega', v^\zeta:\mathtt{E}}(\mathcal{C}_1) \otimes \mathsf{Cfg}_{v^\zeta:\mathtt{E}, w^\delta:\mathtt{F}}(\mathcal{C}_2)) \vdash [w^\delta:\mathtt{F}].\]
We apply two external reduction rules ($\exists$ and $\otimes$) on them to get judgments of the form $[\bar{u}^\gamma:\omega'],  (\mathsf{Cfg}_{\bar{u}^\gamma:\omega', v^\zeta:\mathtt{E}}(\mathcal{C}_1), \mathsf{Cfg}_{v^\zeta:\mathtt{E}, w^\delta:\mathtt{F}}(\mathcal{C}_2))) \vdash [w^\delta:\mathtt{F}]$.

\item {\bf 2.3.} We apply identity elimination on all the remainder identity judgments that we have ignored (including the ones from {\bf Step 2.1}).
\item {\bf 2.4.} We apply a Merge rule on all  judgments remaining from {\bf Step 2.2} which are of the form 
\[[\bar{u}^\gamma:\omega'],  (\mathsf{Cfg}_{\bar{u}^\gamma:\omega', v^\zeta:\mathtt{E}}(\mathcal{C}_1), \mathsf{Cfg}_{v^\zeta:\mathtt{E}, w^\delta:\mathtt{F}}(\mathcal{C}_2))) \vdash [w^\delta:\mathtt{F}].\]
It replaces them with two judgments connected with a fresh internal channel $v^\zeta$:
\[[\bar{u}^\gamma:\omega'], (\mathsf{Cfg}_{\bar{u}^\gamma:\omega', v^\zeta:\mathtt{E}}(\mathcal{C}_1) \vdash  [v^\zeta:\mathtt{E}]  \qquad  [v^\zeta:\mathtt{E}],  \mathsf{Cfg}_{v^\zeta:\mathtt{E}, w^\delta:\mathtt{F}}(\mathcal{C}_2)) \vdash [w^\delta:\mathtt{F}].\]
$\mathcal{C}_1$ is the configuration (or a process) spawned and $\mathcal{C}_2$ is the continuation. With two other external reductions on the $\mathsf{Cfg}$ predicates, we unfold the definition of the predicates based on their pattern and get back to a tape satisfying the invariant. We start over the procedure from {\bf Step 1}. \footnote{We need to break down this step to be compatible with the order enforced by the cut elimination algorithm.}
\end{itemize}

\item {\bf Step 3.  Communication along an internal channel:} an internal reduction can be applied on predicates of the form $[w^\gamma:\mathtt{A}]$.\\
We provide the steps of the cut elimination algorithm for a case in which the protocol of communication is an internal choice ($\oplus$) in Figure \ref{fig:intred}. The other cases are similar.\footnote{Our system is based on an asynchronous semantics. However, cut elimination is a synchronous procedure. When we run our cut elimination algorithm on this  particular derivation, it models a synchronized sending and receiving messages between processes.} It is straightforward to observe that both processes $\mathsf{P}$ and $\mathsf{Q}_k$ are the computational continuations of the original processes in the configuration: $(Rw.k;\mathsf{P})$ sends the label $k$ along the internal channel $w^\gamma$ and steps to $\mathsf{P}$. Process $\mathbf{case}Lw(\ell \Rightarrow \mathsf{Q}_\ell)_{\ell \in L}$  when receiving label $k$ along channel $w^\gamma$ steps to $\mathsf{Q}_k$. Similar to {\bf Step 1.} with two external reductions on the $\mathsf{Cfg}$ predicates, we get back to a tape satisfying the invariant. 

Similar to {\bf Step 1.} here we get an extra (green) tape with the single judgment ${{\mathsf{Msg}(w^\gamma.k) \vdash \mathsf{Msg}(w^\gamma.k)}}$ that can be closed by an ID-elim. We also add a similar extra judgment  ${{\mathsf{Msg}(w^\gamma.k) \vdash \mathsf{Msg}(w^\gamma.k)}}$ to the current tape by a $\otimes$ principal reduction. To be compatible with the order enforced by the cut elimination algorithm, we ignore this judgment until the next {\bf Step 2.3}.
\begin{figure}
    \centering
 { \[\hspace*{-2cm}
 {\begin{array}{ll}
    {[z^\eta:\mathtt{C}]}, {\color{red}\mathsf{Cfg}_{z^\eta: \mathtt{C},w^\gamma:\oplus\{\ell:A_\ell\}_{\ell \in L}}(Rw.k;\mathsf{P})} \vdash {[w^\gamma:\oplus\{\ell:\mathtt{A}_\ell\}_{\ell \in L}]}\\ \hspace{2cm} {[w^\gamma:\oplus\{\ell:\mathtt{A}_\ell\}_{\ell \in L}]}, {\color{red}\mathsf{Cfg}_{w^\gamma:\oplus\{\ell:A_\ell\}_{\ell \in L}, v^\delta:\mathtt{D}}(\mathbf{case}Lw(\ell \Rightarrow \mathsf{Q}_\ell)_{\ell \in L})} \vdash {[v^\delta:\mathtt{D}]} \\ 
 \hspace{9cm} \Downarrow \mbox{External reduction}(\mu) \notag\\
\hline\\
    {[\bar{z}^\eta:\omega]}, \mathsf{Msg}(w^\gamma.k) \otimes \mathsf{Cfg}_{ \bar{z}^\eta:\omega, w^{\gamma+1}:\mathtt{A}_k}(\mathsf{P}) \vdash {\color{blue}[w^\gamma:\oplus\{\ell:\mathtt{A}_\ell\}_{\ell \in L}]} \\\hspace{2cm} {\color{blue}[w^\gamma:\oplus\{\ell:\mathtt{A}_\ell\}_{\ell \in L}]}, \&\{\ell:  \mathsf{Msg}(w^\gamma.\ell) \multimap \mathsf{Cfg}_{w^{\gamma+1}:\mathtt{A}_\ell, v^\delta:\mathtt{D}}(\mathsf{Q}_\ell)\}_{\ell\in L} \vdash {[v^\delta:\mathtt{D}]} \\
 \hspace{9cm}  \Downarrow \mbox{principal reduction}(\mu) \notag\\
   \hline\\
    {[\bar{z}^\eta:\omega]}, \mathsf{Msg}(w^\gamma.k) \otimes \mathsf{Cfg}_{ \bar{z}^\eta:\omega, w^{\gamma+1}:\mathtt{A}_k}(\mathsf{P}) \vdash {\color{blue}\oplus\{\ell:(\mathsf{Msg}(w^\gamma.\ell) \otimes [w^{\gamma+1}:\mathtt{A}_\ell])\}_{\ell \in L}}\\\hspace{2cm} {\color{blue}\oplus\{\ell:(\mathsf{Msg}(w^\gamma.\ell) \otimes [w^{\gamma+1}:\mathtt{A}_\ell])\}_{\ell \in L}}, \&\{\ell:  \mathsf{Msg}(w^\gamma.\ell) \multimap \mathsf{Cfg}_{w^{\gamma+1}:\mathtt{A}_\ell, v^\delta:\mathtt{D}}(\mathsf{Q}_\ell)\}_{\ell\in L} \vdash {[v^\delta:\mathtt{D}]}  \\
   \hspace{9cm}\Downarrow \mbox{principal reduction}(\oplus) \notag\\
  \hline\\{[\bar{z}^\eta:\omega]}, \mathsf{Msg}(w^\gamma.k) \otimes \mathsf{Cfg}_{ \bar{z}^\eta:\omega, w^{\gamma+1}:\mathtt{A}_k}(\mathsf{P}) \vdash (\mathsf{Msg}(w^\gamma.k) \otimes [w^{\gamma+1}:\mathtt{A}_k])\\\hspace{2cm} (\mathsf{Msg}(w^\gamma.k) \otimes [w^{\gamma+1}:\mathtt{A}_k]), {\color{red}\&\{\ell:  \mathsf{Msg}(w^\gamma.\ell) \multimap \mathsf{Cfg}_{w^{\gamma+1}:\mathtt{A}_\ell, v^\delta:\mathtt{D}}(\mathsf{Q}_\ell)\}_{\ell\in L}} \vdash {[v^\delta:\mathtt{D}]} \\
  \hspace{9cm} \Downarrow \mbox{External reduction}(\&) \notag\\
   \hline\\
    {[\bar{z}^\eta:\omega]}, {\color{red}\mathsf{Msg}(w^\gamma.k) \otimes \mathsf{Cfg}_{ \bar{z}^\eta:\omega, w^{\gamma+1}:\mathtt{A}_k}(\mathsf{P})} \vdash (\mathsf{Msg}(w^\gamma.k) \otimes [w^{\gamma+1}:\mathtt{A}_k])\\\hspace{2cm} (\mathsf{Msg}(w^\gamma.k) \otimes [w^{\gamma+1}:\mathtt{A}_k]),  (\mathsf{Msg}(w^\gamma.k) \multimap \mathsf{Cfg}_{w^{\gamma+1}:\mathtt{A}k, v^\delta:\mathtt{D}}(\mathsf{Q}_k) \vdash {[v^\delta:\mathtt{D}]} \\
    \hspace{9cm}  \Downarrow \mbox{External reduction}(\otimes) \notag\\ \hline\\ {[\bar{z}^\eta:\omega]}, \mathsf{Msg}(w^\gamma.k), \mathsf{Cfg}_{ \bar{z}^\eta:\omega, w^{\gamma+1}:\mathtt{A}_k}(\mathsf{P}) \vdash {\color{blue}(\mathsf{Msg}(w^\gamma.k) \otimes [w^{\gamma+1}:\mathtt{A}_k])}\\\hspace{2cm} {\color{blue}(\mathsf{Msg}(w^\gamma.k) \otimes [w^{\gamma+1}:\mathtt{A}_k])},  (\mathsf{Msg}(w^\gamma.k) \multimap \mathsf{Cfg}_{w^{\gamma+1}:\mathtt{A}k, v^\delta:\mathtt{D}}(\mathsf{Q}_k) \vdash {[v^\delta:\mathtt{D}]} \\
       \hspace{9cm}  \Downarrow \mbox{Principal reduction}(\otimes) \notag\\\hline\\
    {\mathsf{Msg}(w^\gamma.k) \vdash \mathsf{Msg}(w^\gamma.k)} \quad{[\bar{z}^\eta:\omega]} \mathsf{Cfg}_{ \bar{z}^\eta:\omega, w^{\gamma+1}:\mathtt{A}_k}(\mathsf{P}) \vdash  [w^{\gamma+1}:\mathtt{A}_k]\\\hspace{2cm} \mathsf{Msg}(w^\gamma.k), [w^{\gamma+1}:\mathtt{A}_k], {\color{red} \mathsf{Msg}(w^\gamma.k) \multimap \mathsf{Cfg}_{w^{\gamma+1}:\mathtt{A}k, v^\delta:\mathtt{D}}(\mathsf{Q}_k)} \vdash {[v^\delta:\mathtt{D}]} \\
  \hspace{9cm}  \Downarrow \mbox{External reduction}(\multimap) \notag\\\hline\\
    {\mathsf{Msg}(w^\gamma.k) \vdash \mathsf{Msg}(w^\gamma.k)} \quad{[\bar{z}^\eta:\omega]}, \mathsf{Cfg}_{ \bar{z}^\eta:\omega, w^{\gamma+1}:\mathtt{A}_k}(\mathsf{P}) \vdash  [w^{\gamma+1}:\mathtt{A}_k]\\\hspace{2cm}  [w^{\gamma+1}:\mathtt{A}_k], \mathsf{Cfg}_{w^{\gamma+1}:\mathtt{A}k, v^\delta:\mathtt{D}}(\mathsf{Q}_k) \vdash {[v^\delta:\mathtt{D}]} \quad {\color{green}{\mathsf{Msg}(w^\gamma.k) \vdash \mathsf{Msg}(w^\gamma.k)}}(\mbox{identity elimination})\\
\hspace{9cm}    \Downarrow \mbox{External reduction}\times 2 \notag\\ \hline\\
    {{\mathsf{Msg}(w^\gamma.k) \vdash \mathsf{Msg}(w^\gamma.k)}}  \quad
    {[\bar{z}^\eta:\omega]},T_{P} \vdash  [w^{\gamma+1}:\mathtt{A}_k]\\\hspace{2cm}   [w^{\gamma+1}:\mathtt{A}_k], T_{Q_k} \vdash {[v^\delta:\mathtt{D}]} \\
\end{array}}
\]}
    \caption{A run of the cut elimination algorithm on communicating processes.}
    \label{fig:intred}
\end{figure}

 \end{description}
By the definition of predicates $\mathsf{Cfg}$ and $[x^\alpha: \mathtt{A}]$, and the observations we made in the proof of Lemma \ref{lem:derivation}, these are the only steps we need to consider for the cut elimination algorithm.

Consider the cut-free proof returned by our algorithm.  By the property proved above, it is enough to show that an external reduction will be applied on $[x^\alpha:\omega]$ or $[y^\beta:\mathtt{B}]$. We use linearity and validity of the cut-free output derivation. If there are no infinite branches in the proof then by linearity of the calculus we know that a rule is applied on $\overline{[\bar{x}^\alpha:\omega]}$ and $[y^\beta:\mathtt{B}]$. In the infinite case, recall that the predicate $\mathsf{Cfg}$ for a recursive process is defined coinductively; no subformula of it in the antecedents can be a part of an infinite $\mu$-trace. Thus an external reduction (flip rule)  has to be applied on $[x^\alpha:\omega]$ or $[y^\beta:\mathtt{B}]$ to produce a judgment of the derivation. This completes the proof as it shows that the configuration will eventually communicate with one of its external channels. 

We can take one step further, and show that the configuration either terminates or it will eventually communicate with one of its external channels by \emph{receiving a message}. Consider a branch in the cut-free valid proof as described above. If the branch is finite it is straightforward to see that the computation terminates. By a similar reasoning to the previous paragraph, in an infinite branch either $[x^\alpha:\omega]$ has to be a part of an (infinite) $\mu$-trace or $[y^\beta:\mathtt{B}]$ has to be a part of an (infinite) $\nu$-trace. Without loss of generality assume that $[x^\alpha:\omega]$ is a part of a $\mu$-trace.  Since the traces are infinite, there has to be an occurrence of a least fixed point type $\mathtt{t}$ in a predicate $[x^{\gamma}:\mathtt{t}]$ on the branch. As a result, there will be a process in the computation such that it communicates along $x^\gamma$, and by the type of $x^\gamma$, we know that it will be receiving a fixed point unfolding message.
\end{proof} 

\end{document}